\documentclass[12pt,oneside]{sigmasthesis}

\title{Fair Division of Indivisible Items}			
\type{Dissertation}				
\author{Kevin Hsu}			
\authordetails{MSc, University of Victoria, 2020 \\ BSc (Honours), University of Victoria, 2018}	
\degree{Doctor of Philosophy}			
\department{Department of Computer Science}	
\date{2025}				

\panel{
\panelist{Dr.\ Valerie King}{Supervisor}{Department of Computer Science} 
\panelist{Dr.\ Yun Lu}{Departmental Member}{Department of Computer Science} 
\panelist{Dr.\ Jon Noel}{Outside Member}{Department of Mathematics and Statistics} 
}

\usepackage{amsmath, amsthm, amssymb}
\usepackage{graphicx}
\usepackage{tikz}
\usepackage{booktabs} 
\usepackage{verbatim} 
\usepackage{dsfont} 
\usepackage{etoolbox} 

\usepackage[authoryear]{natbib}
\usepackage{algorithm}
\usepackage{thm-restate}
\usepackage{mathtools}
\usepackage[noend]{algpseudocode}
\usepackage{amsfonts}
\usepackage{tikz}
\usepackage{subcaption}

\usepackage{kbordermatrix, blkarray}

\usepackage{pifont}
\newcommand{\xmark}{\ding{55}}

\usepackage[english]{babel}
\usepackage[autostyle, english = american]{csquotes}
\MakeOuterQuote{"}

\usetikzlibrary{decorations.markings}
\usetikzlibrary{arrows.meta}

\tikzset{arc/.style={decoration={
			markings,
			mark=at position 0.6 with {\arrow{>[scale=2.5,
					length=2,
					width=2.5]}}},postaction={decorate}}}

\tikzset{every loop/.style={}}

\tikzset{redvertex/.style={black, draw=black, circle, fill, scale=0.5}}
\tikzset{blkvertex/.style={black, draw=black, circle, scale=0.5}}

\tikzset{sqvertex/.style={black, draw=black, rectangle}}

\tikzset{circlevertex/.style={black, draw=black, circle, fill, scale=0.5}}

\tikzset{emptyvertex/.style={black, draw=black, circle, scale=0.5}}

\tikzset{diamondvertex/.style={black, draw=black, diamond, fill, scale=0.5}}

\newcommand{\N}{\mathbb{N}}

\newcommand\restr[2]{{
  \left.\kern-\nulldelimiterspace 
  #1 
  \vphantom{\big|} 
  \right|_{#2} 
  }}

\DeclareMathOperator*{\argmax}{arg\,max}
\DeclareMathOperator{\NW}{NW}

\newlistof[section]{abbrevs}{abr}{List of Abbreviations}
\newcommand{\abbrevs}[2]{
	\refstepcounter{abbrevs}#1 (#2)
	\addcontentsline{abr}{abbrevs}{\numberline{}#2 - #1}}	
	
\newcommand{\abbrevsmulti}[3]{
		\refstepcounter{abbrevs}#1 (#2)
		\addcontentsline{abr}{abbrevs}{\numberline{}#2 - #3}}	

\newtheorem{theorem}{Theorem}[chapter] 
\newtheorem{lemma}[theorem]{Lemma} 
\newtheorem{proposition}[theorem]{Proposition}
\newtheorem{corollary}[theorem]{Corollary}
\newtheorem{observation}[theorem]{Observation}

\newtheorem*{problem}{Problem}

\theoremstyle{definition}
\newtheorem{definition}[theorem]{Definition} 

\newtheorem{example}[theorem]{Example}

\begin{document}
\frontmatter 	

\newbool{terr-ack}
\booltrue{terr-ack} 
\maketitle{terr-ack}	

\makecommittee	


\begin{abstract}

We study the fair division of indivisible items. In the general model, the goal is to find a way to allocate $m$ indivisible items to $n$ agents, while satisfying some fairness criteria. We are particularly interested in the fairness criteria of maximin share (MMS), envy-freeness up to one item (EF1), and envy-freeness up to any item (EFX). Additionally, we study a recently-introduced graphical model that represents the fair division problem as a multigraph, in which vertices correspond to agents and edges correspond to items. The graphical model stipulates that an item can have non-zero marginal utility to an agent only if its corresponding edge is incident to the agent's corresponding vertex, capturing the idea of proximity between agents and items. It is desirable to allocate edges only to their endpoints. Such allocations are called orientations, as they correspond naturally to graph orientations.

Our first contribution concerns MMS allocations of mixed manna (i.e.\ a mixture of goods and chores) in the general model. It is known that MMS allocations of goods exist when $m \leq n+5$. We generalize this result by showing that when $m \leq n+5$, MMS allocations of mixed manna exist as long as $n \leq 3$, there is an agent whose MMS threshold is non-negative, or every item is a chore. Remarkably, our result leaves only the case in which every agent has a negative MMS threshold unanswered.

Our second contribution concerns EFX orientations of multigraphs of goods. We show that deciding whether EFX orientations exist for multigraphs is NP-complete, even for symmetric bi-valued multigraphs. Complementary to this, we show symmetric bi-valued multigraphs that do not contain non-trivial odd multitrees have EFX orientations that can be found in polynomial time.

Our third contribution concerns EF1 and EFX orientations of graphs and multigraphs of chores. We obtain polynomial-time algorithms for deciding whether such graphs have EF1 and EFX orientations, resolving a previous conjecture and showing a fundamental difference between goods and chores division. In addition, we show that the analogous problems for multigraphs are NP-hard.

\end{abstract}

\newpage\addtoToC{Table of Contents}\tableofcontents 
\newpage\addtoToC{List of Figures}\listoffigures 

\newpage\addtoToC{List of Abbreviations}\listofabbrevs

\begin{acknowledgements} 
	\vskip 2\baselineskip
	
	\noindent I would like to thank my supervisory committee and external examiner Nisarg Shah for the time and effort they invested in carefully examining this dissertation and their helpful comments and suggestions. In particular, I am grateful to Nisarg Shah for raising an interesting question during the oral examination that inspired the addition of Chapter~\ref{chapter:po-and-orientations} to this work.
	
	I would also like to thank my amazing supervisor Valerie King for her encouragement and unwavering support throughout the years. Under her guidance, I developed not only professionally, but also personally thanks to her inspiring wisdom and kindness. I could not have asked for a better supervisor and cannot thank her enough.
	
	I would also like to thank my family and friends, whose continued support made this work possible. In particular, I would like to thank my friends Philip and Charlie for attending my oral examination and supporting me on such an important occasion.
	
	Finally, I would like to thank the numerous anonymous reviewers for their insightful comments and suggestions during the peer review process at the conferences to which the works included in this dissertation were submitted.
	
\end{acknowledgements}

\mainmatter 


\chapter{Introduction} 

How does one divide things between a group of agents in a fair manner? Various forms of this problem have presented themselves since time immemorial and this problem has become the basis of a rich literature at the intersection of computer science, economics, and mathematics. A reason for its importance and popularity is its exceptionally wide applicability, including household chores division, inheritance division, divorce settlements, course allocation, et cetera. This problem may appear simple at first glance, but is greatly complicated by the fact that the agents may have conflicting interests, and by the implicit question of what fairness is. Moreover, certain divisions may not be feasible depending on the specific setting. Thus, for each setting under consideration, one may ask what level of fairness is achievable and how to find a division that achieves it. 

The study of fair division generally consists of the following steps. First, one chooses a real-life application. Then, one designs a suitable mathematical model that accurately represents the relevant parameters and constraints. Finally, one studies different fairness criteria for the model with the goal of knowing which fairness criteria can be attained under what conditions, and how to efficiently compute an allocation that satisfies the candidate fairness criteria.

The first formal model of fair division was introduced in 1948 by \citet{steinhaus1948problem} and is called the {\em cake-cutting model}. In this model, the items being divided are assumed to be arbitrarily divisible, much like cake, and are represented by the unit interval $[0, 1]$. There are $n$ agents, each of whose preference is represented by a utility function $u_i: \mathcal{P}([0, 1]) \to \mathbb{R}_{\geq 0}$ that assigns a numerical value to each possible cut of cake. By possibly scaling the utility functions, we assume that $u_i([0,1]) = 1$ for each agent $i$. The goal is to find a partition of the cake into $n$ parts, and to allocate the parts to the agents while satisfying some fairness criterion. Various fairness criteria for cake-cutting have been studied, most important among which are {\em envy-freeness} and {\em proportionality}. An {\em envy-free} allocation is one in which no agent envies another, that is, no agent believes the slice of another agent to have greater utility than their own slice. A {\em proportional} allocation is one in which every agent receives their proportional share, that is, every agent believes their slice is worth at least $1/n$. Beyond fairness, another consideration for the cake-cutting model is that of {\em contiguity} --- it is better for each agent to be allocated a contiguous slice of cake, rather than a union of many small slices.

For the case of 2 agents, an envy-free and proportional allocation can be found using the folklore {\em divide-and-choose} (also called {\em cut-and-choose}) algorithm. In this algorithm, one agent cuts the cake into what they perceive to be two equal halves. Then, the other agent chooses their preferred half. The remaining half is given to the first agent. When there are more than 2 agents, the existence of envy-free allocations has been shown using Lyapunov's convexity theorem about vector measures \citep{dubins1961cut} and Sperner's lemma from combinatorics \citep{edward1999rental}. Unfortunately, because these proofs are non-constructive, they do not result in an algorithm for computing an envy-free allocation. This leads researchers to attempt to design such an algorithm and to ask what the computational complexity of this problem is. Algorithms for cake-cutting are usually described using the Robertson-Webb query model \citep{woeginger2007complexity}. In this model, a set of special points in $[0,1]$ called {\em cut-points} is maintained and initially includes the two boundaries 0 and 1 of the cake. Two types of queries can be made to the agents. The first, the cut query, presents a cut-point $a$ and a target value $t \in [0,1]$ to an agent $i$, and asks for the leftmost point $b \in [a,1]$ such that $u_i([a,b]) = t$, if such a point exists. Such a point is then added to the set of cut-points. The second, the evaluation query, presents two cut-points $a<b$ to an agent $i$ and asks for the value of $u_i([a,b])$.

The complexity of a cake-cutting algorithm can be described by the number of queries it requires. Unfortunately, designing an efficient algorithm for computing an envy-free algorithm has proven to be difficult. While the best-known lower bound of the query complexity of this problem is $\Omega(n^2)$ \citep{procaccia2009thou}, the state-of-the-art algorithm requires $O\left(n^{n^{n^{n^{n^n}}}}\right)$ queries \citep{aziz2016discrete}. In contrast, proportional allocations can be computing using $O(n \log n)$ queries even for contiguous cuts \citep{even1984note}, with a matching lower bound \citep{woeginger2007complexity,edmonds2011cake}.

A fundamental assumption of the cake-cutting model is that the items being divided are arbitrarily divisible. However, there are many scenarios to which this assumption does not apply. For example, when allocating classrooms to concurrent lectures, it is natural to require that a classroom is not shared between multiple concurrent lectures. Similarly, when allocating delivery tasks to workers of a food-delivery platform, a single task should be assigned to a single worker. Other examples of items that are inherently indivisible include movie tickets, individual pieces of Halloween candy, rooms in a shared apartment, et cetera. This consideration prompts the study of the fair division of {\em indivisible} items, which is the focus of our work. This indivisible setting is modelled using $n$ agents and $m$ items. Similarly to cake-cutting, the preference of each agent $i$ is represented by a utility function $u_i$ that assigns a numerical value to each subset of items.

Classical fairness criteria such as envy-freeness and proportionality are often unattainable in the indivisible setting --- in the canonical example of dividing a single item between two agents, one agent is necessarily envious. Thus, relaxations of these criteria have been introduced. In this work, we focus on three such criteria: {\em \abbrevsmulti{envy-freeness up to one good}{EF1}{envy-freeness up to one good/chore/item}} \citep{lipton2004approximately,budish2011combinatorial}, {\em \abbrevsmulti{envy-freeness up to any good}{EFX}{envy-freeness up to any good/chore/item}} \citep{caragiannis2019unreasonable,gourves2014near}, and {\em \abbrevs{maximin share}{MMS}} \citep{budish2011combinatorial}. We informally describe each of them below for the special case in which the items being divided are {\em goods}, that is, they have non-negative utility to every agent. The formal definitions are deferred to Chapter~\ref{chapter:theoretical-overview}.

EF1 allocations were implicitly introduced in 2004 by \citet{lipton2004approximately} and again more formally in 2011 by \citet{budish2011combinatorial}. Because envy between agents is often unavoidable, an EF1 allocation allows each agent (say Alice) is allowed to envy another (say Bob) as long as Alice's envy toward Bob disappears when she ignores what she perceives as the best good allocated to Bob.

EFX allocations were introduced in 2014 by \citet{gourves2014near} under the name of {\em nearly envy-free allocations} and again under its present name in 2019 by \citet{caragiannis2019unreasonable}. EFX allocations are more stringent than EF1 allocations, in that they only allow Alice to envy Bob as long as Alice's envy toward Bob disappears when she ignores {\em any} good allocated to Bob. Thus, EFX allocations are EF1. However, the converse is not always true. Consider the example of dividing two \$1 coins and one \$5 note between Alice and Bob. The allocation in which Alice receives a \$1 coin and Bob receives both a \$1 coin and a \$5 note is EF1 because Alice's envy toward Bob disappears when she ignores the \$5 note allocated to Bob, but is not EFX because her envy persists when she ignores the \$1 coin allocated to Bob.

MMS allocations were introduced in 2011 by \citet{budish2011combinatorial}. MMS allocations differ fundamentally from EF1 and EFX allocations in that MMS is a  {\em threshold-based} fairness notion, similar to proportionality, in which each agent must receive a bundle whose utility meets a threshold value. Informally, for each agent $i$, we define their {\em MMS threshold} to be the utility that they would receive if they were to choose the partition of the goods into $n$ bundles, but must receive the bundle with the least utility in their judgment. Different agents may have different MMS thresholds depending on how they value each good. An MMS allocation is one that provides each agent with their respective MMS threshold. MMS allocations generalize the spirit behind the divide-and-choose algorithm for division between 2 agents to $n>2$ agents in the sense that the agent who chooses the division must choose their bundle last. They are particularly interesting as they relate to the {\em original position} (also known as the {\em veil of ignorance}) introduced by the American philosopher \citet{rawls1971}, which is the idea that principles of justice need to be decided without knowledge of one's position in society, so that one does not shape justice in one's own favour.

Beyond fairness, certain settings of fair division have additional mathematical structure that must be considered. A common such example is spatial proximity. Consider the example of allocating delivery tasks to workers of a food-delivery platform. Because delivery workers and restaurants have physical locations, it is desirable to assign delivery tasks to workers who are physically close to the restaurant where the order is to be collected. Allocating classrooms to different lectures is another example. It is more desirable for a class provided by a given university department to take place in a classroom that is close to the department rather than in one far away.

Motivated by such examples, \citet{christodoulou2023fair} introduced the {\em graphical model} of fair division in 2023. This model of fair division represents the problem as a graph $G$, in which each vertex represents an agent and each edge represents an item. An item is allowed to have non-zero marginal utility to an agent only if its representative edge is incident to the agent's representative vertex. Thus, each agent only "cares" about the items in their vicinity. In this setting, it is particularly desirable to find allocations that allocate each item to an agent who "cares" about it. Such allocations are called {\em orientations} because they exactly correspond to orientations of the graph $G$ in a natural way --- each edge is oriented toward the endpoint representing the agent who receives the item represented by the edge.

\section{Organization}

This dissertation is organized as follows. In Chapter~\ref{chapter:theoretical-overview}, we formally define the models of fair division of indivisible items and various fairness considerations in order to familiarize the reader with important background knowledge. We also present a selection of the canonical results and techniques that are well-known in the field in this chapter.

In Chapter~\ref{chap:our-contribution}, we state our contributions and contextualize them by presenting previous work by other researchers that are most directly related to ours.

We present the details of our work in Chapters~\ref{chapter:MMS-Mixed-Manna}-\ref{chapter:Fair-Orientations-Chores}. Each of these chapters is independent and pertains to a different topic.

Finally, we conclude the dissertation with a technical discussion on the relationship between Pareto optimality and orientations in Chapter~\ref{chapter:po-and-orientations}.

\chapter{Overview of Previous Work}\label{chapter:theoretical-overview}

In this chapter, we formally define two models of fair division and various fairness concepts, while highlighting a selection of fundamental results in the literature. This chapter is meant to serve as a brief introduction of the well-known results from previous work by other researchers. For a more comprehensive introduction to the fair division of indivisible items, we refer the reader to the survey by \citet{amanatidis2022fair}. The definitions covered in this chapter are standard, and can also be found in \citep{amanatidis2022fair}.

\section{The General Model}

An {\em instance} $I$ of the fair division problem is a tuple $(N, M, U)$ where $N = [n]$ is a set of $n$ agents, $M = \{o_1, o_2, \dots, o_m\}$ is a set of $m$ items, and $U$ is a collection of $n$ utility functions $u_i:\mathcal{P}(M) \rightarrow \mathbb{R}$, each corresponding to an agent. An agent's utility function represents its preferences by assigning a numerical value to each subset of items. For the case of a single item $o_j$, we will write $u_i(o_j)$ instead of $u_i(\{o_j\})$ for brevity.

Items generally can have any real-valued utility. Two important types of items are {\em goods} and {\em chores}.

\begin{definition}[Goods and Chores]
	We say an item $o_j$ is a {\em good} to an agent $i$ if $u_i(S \cup \{o_j\}) \geq u_i(S)$ for every set $S \subseteq M$. Similarly, we say an item $o_j$ is a {\em chore} to an agent $i$ if $u_i(S \cup \{o_j\}) \leq u_i(S)$ for every subset $S \subseteq M$.
\end{definition}

Importantly, it is possible for an item to be neither a good nor a chore to the same agent.

\begin{example}[An Item that is Neither a Good nor a Chore]\label{ex:neither-good-nor-chore}
	Consider an agent $i$ and three items: a dog, a cat, and dog food. Because dog food is useful to have when there is a dog, we have $u_i(\{\text{dog}, \text{dog food}\}) > u_i(\{\text{dog}\})$. On the other hand, if one only has a cat, then having dog food is undesirable because it is useless and wastes storage space, so we have $u_i(\{\text{cat}, \text{dog food}\}) < u_i(\{\text{cat}\})$. Thus, the marginal utility of dog food is either positive or negative depending on the situation, so dog food is neither a good nor a chore.
\end{example}

Moreover, it is possible that different agents disagree on whether an item is a good or a chore. For example, Alice might like chocolate, but Bob might be allergic to it. This suggests the following definition.

\begin{definition}[Objective and Subjective Items]
	We say an item $o_j$ is an {\em objective good} if it is a good to every agent. Similarly, we say $o_j$ is an {\em objective chore} if it is a chore to every agent. If an item is not objective, it is said to be {\em subjective}.
\end{definition}

It is useful to classify instances depending on the types of items they contain.

\begin{definition}[Instance Types]
	An instance is called a {\em goods instance} (resp.\ {\em chores instance}) if every item is a good (resp.\ chore) to every agent. An instance that is neither a goods instance nor a chores instance is called a {\em mixed instance}.
\end{definition}

A combination of goods and chores is often referred to as {\em mixed manna} in the literature. 

We now formally define how to allocate items to agents.

\begin{definition}[Partial Allocations]
	A {\em partial allocation} $\pi$ is an $n$-tuple $(\pi_1, \pi_2, \dots, \pi_n)$ where each $\pi_i$ is a subset of $M$ and no pair of distinct subsets $\pi_i, \pi_j$ intersect.
\end{definition}

\begin{definition}[Allocations]
	A {\em complete allocation} (or simply, an {\em allocation}) is a partial allocation $\pi = (\pi_1, \pi_2, \dots, \pi_n)$ such that $\cup_{i \in N} \pi_i = M$.
\end{definition}

We interpret each subset $\pi_i$ in a (partial) allocation as being allocated to the agent $i$. For convenience, we refer to these subsets as {\em bundles}.

\begin{definition}[Bundles]
	Let $\pi = (\pi_1, \pi_2, \dots, \pi_n)$ be a partial allocation. The subset $\pi_i$ is called the {\em bundle of agent} $i$ or the {\em bundle allocated to agent} $i$.
\end{definition}

We often make the assumption that utility functions are {\em additive}. Additivity is a very commonly made assumption in the literature.

\begin{definition}[Additivity]
	A utility function $u_i$ is said to be {\em additive} if for any pair of disjoint subsets $S, T \subseteq M$, we have $u_i(S \cup T) = u_i(S) + u_i(T)$. If every utility function is additive, we call $I$ an {\em additive instance}.
\end{definition}

If every agent has an additive utility function, we can represent the instance $I$ using an $n \times m$ matrix $M(I)$ in which the entry $m_{ij}$ takes the value $u_i(o_j)$. This representation is useful when describing examples of instances.

\begin{example}[Matrix Representation]
	The following matrix represents an instance $I$ containing 3 agents and 4 items, in which every agent has an additive utility function. Each row represents the utility function of an agent. For example, row 1 indicates that $u_1(o_1) = 1, u_1(o_2) = 2, u_1(o_3) = 3$, and $u_1(o_4) = 4$. Moreover, $o_3$ is the unique objective good in this instance.
	$$\begin{blockarray}{ccccc}
		& o_1 & o_2 & o_3 & o_4 \\
		\begin{block}{c[cccc]}
			1&1 & 2 & 3 & 4 \\
			2&-1 & -1 & 1 & -1 \\
			3&5 & 2 & 1 & 0 \\
		\end{block}
	\end{blockarray}$$
\end{example}

We define some other commonly made assumptions about utility functions.

\begin{definition}[Monotonicity]
	A utility function $u_i$ is said to be
	\begin{itemize}
		\item {\em monotone increasing} if for any pair of subsets $S \subseteq T$ of $M$, we have $u_i(S) \leq u_i(T)$; and
		\item {\em monotone decreasing} if for any pair of subsets $S \subseteq T$ of $M$, we have $u_i(S) \geq u_i(T)$.
	\end{itemize}
	A utility function is said to be {\em monotone} if it is either monotone increasing or monotone decreasing. If every utility function is monotone, we call $I$ a {\em monotone instance}.
\end{definition}

The purpose of defining two versions of monotonicity is to account for both goods instances and chores instances. In either case, additive utility functions are monotone. On the other hand, additive utility functions for mixed instances  are non-monotone. This is because including an additional good in a bundle possibly increases its utility, while including an additional chore possibly decreases its utility. However, a variant of monotonicity has also been defined for mixed instances.

\begin{definition}[Double Monotonicity]
	A utility function $u_i$ is said to be {\em doubly monotone} if for each agent $i$ and each item $o_j$, the item $o_j$ is a good or a chore to agent $i$. If every utility function is doubly monotone, we call $I$ a {\em doubly monotone instance}.
\end{definition}

Essentially, double monotonicity prevents the existence of an item $o_j$ that is neither a good nor a chore to an agent $i$, i.e.\ there to be sets $S, T \subseteq M$ such that $u_i(S \cup \{o_j\}) > u_i(S)$ and $u_i(T \cup \{o_j\}) < u_i(T)$, such as the item in Example~\ref{ex:neither-good-nor-chore}. On the other hand, this above cannot occur with additive utility functions. Thus, additive utility functions are doubly monotone.

The following two assumptions are relaxations of additivity, and are less commonly made.

\begin{definition}[Subadditivity]
	A utility function $u_i$ is said to be {\em subadditive} if for any pair of disjoint subsets $S, T \subseteq M$, we have $u_i(S \cup T) \leq u_i(S) + u_i(T)$. If every utility function is subadditive, we call $I$ a {\em subadditive instance}.
\end{definition}

\begin{definition}[Superadditivity]
	A utility function $u_i$ is said to be {\em superadditive} if for any pair of disjoint subsets $S, T \subseteq M$, we have $u_i(S \cup T) \geq u_i(S) + u_i(T)$. If every utility function is superadditive, we call $I$ a {\em superadditive instance}.
\end{definition}

Researchers have also considered types of instances that restrict the number of values that the utility of an item can have.

\begin{definition}[Binary Instances]
	An instance is said to be {\em binary} if for each agent $i$ and each item $o_j$, we have we have $u_i(o_j) \in \{0, 1\}$.
\end{definition}

\begin{definition}[$k$-valued Instances]
	An instance is said to be {\em $k$-valued} if there exists a set $S$ of size $k$, such that for each agent $i$ and each item $o_j$, we have $u_i(o_j) \in S$.
\end{definition}

Another particularly important type of instances is the type in which agents rank the items in the same order.

\begin{definition}[Same-Order Preference Instances]
	An instance is said to have {\em \abbrevs{same-order preference}{SOP}} (or, to be {\em fully correlated}) if there exists a permutation $\sigma: M \to M$ such that for each agent $i$, we have
	$$
	u_i(o_{\sigma(1)}) \geq u_i(o_{\sigma(2)}) \geq \dots \geq u_i(o_{\sigma(m)})
	$$
\end{definition}

If $I$ is an SOP instance, we can assume that the item utilities are already ordered in non-increasing order without loss of generality, i.e.\ $u_i(o_1) \geq u_i(o_2) \geq \dots \geq u_i(o_m)$ for each agent $i$.

\section{Graph Theory}

We briefly familiarize the reader with some basic definitions and results from graph theory that our results depend on. The content of this section is standard in any introductory graph theory course and is included for the sake of completeness. For this reason, we do not include citations for the basic definitions, and instead refer the interested reader to \citet{chartrand2024graphs} for a more comprehensive introduction to graph theory. A knowledgeable reader can disregard this section.

\subsection{Graphs}

A {\em graph} $G$ is a pair $(V(G), E(G))$ where $V(G)$ is a set of {\em vertices} and $E(G) \subseteq \{\{v, w\} \mid v,w \in V(G)\}$ is a set of unordered pairs of vertices called {\em edges}. Let $e = \{v, w\}$ be an edge. We call the vertices $v, w$ the {\em endpoints} of $e$, and say $e$ is {\em incident} to $v$ and $w$. We also say that $v$ and $w$ are {\em adjacent} to each other. In the special case that $v=w$, we call $e$ a {\em self-loop} at the vertex $v$. When specifying the edge between two vertices $v, w$, we often write $vw$ instead of $\{v, w\}$ as a shorthand.

\begin{definition}[Neighbours and Neighbourhoods]
	Let $v$ be a vertex. The {\em neighbourhood} $N(v)$ of $v$ is the set of vertices $v$ is adjacent to. The vertices in $N(v)$ are called the {\em neighbours} of $v$.
\end{definition}

We now define subgraphs and induced subgraphs. These definitions capture the idea of containment between graphs.

\begin{definition}[Subgraphs]
	Let $G$ be a graph. A {\em subgraph} of $G$ is a graph $H$ such that $V(H) \subseteq V(G)$ and $E(H) \subseteq E(G)$. If $H$ is a subgraph of $G$, we say $G$ {\em contains $H$ as a subgraph}.
\end{definition}

\begin{definition}[Induced Subgraphs]
	Let $G$ be a graph and $S \subseteq V(G)$ be a subset of vertices. The {\em subgraph of $G$ induced by $S$} is the subgraph $H$ such that $V(H) = S$ and $E(H) = \{vw \in E(G) \mid v, w \in S\}$. We refer to any subgraph $H$ of $G$ of this form as an {\em induced subgraph} of $G$. If $H$ is an induced subgraph of $G$, we say $G$ {\em contains $H$ as an induced subgraph}.
\end{definition}

\begin{definition}[Edge-Induced Subgraphs]
	Let $G$ be a graph and $S \subseteq E(G)$ be a set of edges. The {\em subgraph of $G$ induced by $S$} is the subgraph $H$ such that $V(H)$ is the set of endpoints of the edges of $S$, and $E(H) = S$. We refer to any subgraph $H$ of $G$ of this form as an {\em edge-induced subgraph} of $G$. If $H$ is an edge-induced subgraph of $G$, we say $G$ {\em contains $H$ as an edge-induced subgraph}.
\end{definition}

Induced subgraphs are sometimes referred to as {\em vertex-induced subgraphs} in order to emphasize their distinction from edge-induced subgraphs.

We now define paths and cycles in graphs. To do this, we first introduce walks.

\begin{definition}[Walks]\label{def:walk}
	A {\em walk} $W$ in a graph $G$ is a sequence $(v_0, e_1, v_1, e_2, \dots, e_k, v_k)$ such that
	\begin{itemize}
		\item $v_i \in V(G)$ for each $0 \leq i \leq k$;
		\item $e_i \in E(G)$ for each $1 \leq i \leq k$; and
		\item $e_i = \{v_{i-1}, v_i\}$ for each $1 \leq i \leq k$.
	\end{itemize}
	The {\em length} of $W$ is the number $k$ of edges it contains. We say that $W$ is a walk {\em from $v_0$ to $v_k$} or that $W$ is a walk {\em between $v_0$ and $v_k$}.
\end{definition}

Because a graph contains at most one edge between any pair of vertices, it is possible to equivalently define walks as sequences of vertices instead of explicitly including the edges between consecutive vertices. The reason we prefer the above definition is because we will reuse it for multigraphs, which are generalizations of graphs that allow for multiple edges between each pair of vertices.

\begin{definition}[Paths]\label{def:path}
	A {\em path} $P$ in a graph $G$ is a walk $(v_0, e_1, v_1, e_2, \dots, e_k, v_k)$ such that no vertex is repeated, i.e.\ $i \neq j$ implies $v_i \neq v_j$ for all $1 \leq i, j \leq k$.
\end{definition}

\begin{definition}[Cycles]\label{def:cycle}
	A {\em cycle} $C$ in a graph $G$ is a walk $(v_0, e_1, v_1, e_2, \dots, e_k, v_k)$ such that $v_0 = v_k$ and no vertex nor edge is repeated.
\end{definition}

A self-loop is a cycle of length 1. Similarly to paths, cycles can also be specified using the sequence of edges between its consecutive edges.

We now define some important classes of graphs that are commonly considered.

\begin{definition}[Simple Graphs]
	A {\em simple graph} is a graph that contains no self-loops.
\end{definition}

\begin{definition}[Bipartite Graphs]
	A graph $G$ is said to be {\em bipartite} if there exists a partition $A \cup B$ of $V(G)$ such that every edge $e$ has an endpoint in both $A$ and $B$. If $G$ is bipartite, we often write $G = (A \cup B, E)$ instead of $G = (V(G), E(G))$ to specify the partition.
\end{definition}

In order to define the class of connected graphs, we first define components.

\begin{definition}[Components]
	A {\em component} $K$ of a graph $G$ is a maximal subset of $V(G)$ such that between every pair of distinct vertices $v, w \in K$, there exists a path in $G$ between $v$ and $w$.
\end{definition}

\begin{definition}[Connected Graphs]
	A graph $G$ is said to be {\em connected} if it contains a unique component and {\em disconnected} otherwise. Equivalently, $G$ is connected if for each pair of distinct vertices $v, w \in V(G)$, there exists a path between $v$ and $w$, and disconnected otherwise.
\end{definition}

\subsection{Directed Graphs and Graph Orientations}

We now define directed graphs and graph orientations. A directed graph is essentially a graph whose edges have directions. More formally, a {\em directed graph} (or simply a {\em digraph}) $D$ is a pair $(V(D), A(D))$, where $V(D)$ is a set of {\em vertices} and $A(D) \subseteq \{ (v, w) \mid v, w \in V(D) \}$ is a set of {\em directed edges} (also called {\em arcs}). When the context is clear, we often refer to directed edges simply as {\em edges}. Let $e=(v,w)$ be a directed edge. We say $e$ is {\em from} $v$ {\em to} $w$ and that it is {\em directed toward} $w$. We also say that $v$ is {\em adjacent to} $w$ and vice versa.

\begin{definition}[Neighbours and Neighbourhoods]
	Let $v$ be a vertex in a directed graph $G$. The {\em in-neighbourhood} $v$ is the set $N^-(v) \coloneqq \{w \in V(G) \mid (w, v) \in E(G)\}$ and the {\em out-neighbourhood} $v$ is the set $N^+(v) \coloneqq \{w \in V(G) \mid (v, w) \in E(G)\}$. The vertices in $N^-(v)$ and $N^+(v)$ are called the {\em in-neighbours} and {\em out-neighbours} of $v$, respectively.
\end{definition}

\begin{definition}[Underlying Graphs]
	The {\em underlying graph} of a directed graph $D$ is the graph $G$ such that $V(G) = V(D)$ and $E(G) = \{\{v, w\} \mid (v,w) \in A(D) \text{ or } (w,v) \in A(D)\}$. Less formally, $G$ is the graph obtained from $D$ by ignoring edge directions.
\end{definition}

\begin{definition}[Graph Orientations]
	An {\em orientation} of a graph $G$ is a directed graph $D$ whose underlying graph is $G$.
\end{definition}

Directed graphs can be represented visually using arrows to indicate the direction of the edges.

\begin{example}[A Directed Graph]
	The following example is a directed graph on 5 vertices and 7 edges. The edge between 1 and 2 is directed toward 2, as indicated by the direction of the arrow.
	$$
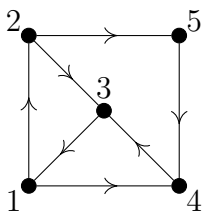
\begin{tikzpicture}
		\node[redvertex] (1) at (0,0) {};
		\node[redvertex] (2) at (0,2) {};
		\node[redvertex] (3) at (1,1) {};
		\node[redvertex] (4) at (2,0) {};
		\node[redvertex] (5) at (2,2) {};
		\node at (-0.2,-0.2) {1};
		\node at (-0.2,2.2) {2};
		\node at (1,1.3) {3};
		\node at (2.2,-0.2) {4};
		\node at (2.2,2.2) {5};
		\draw[arc] (1) to (2);
		\draw[arc] (2) to (3);
		\draw[arc] (4) to (3);
		\draw[arc] (5) to (4);
		\draw[arc] (3) to (1);
		\draw[arc] (1) to (4);
		\draw[arc] (2) to (5);
	\end{tikzpicture}$$
\end{example}

Directed walks, directed paths, and directed cycles in directed graphs can be defined analogously to their counterparts in undirected graphs, by requiring that the directed edges along them are all directed toward their respective subsequent vertices.

\begin{definition}[Directed Walks]
	A {\em directed walk} $W$ in a directed graph $D$ is a sequence $(v_0, e_1, v_1, e_2, \dots, e_k, v_k)$ such that
	\begin{itemize}
		\item $v_i \in V(D)$ for each $0 \leq i \leq k$;
		\item $e_i \in E(D)$ for each $1 \leq i \leq k$; and
		\item $e_i = (v_{i-1}, v_i)$ for each $1 \leq i \leq k$.
	\end{itemize}
	The {\em length} of $W$ is the number $k$ of edges it contains. We say that $W$ is a walk {\em from $v_0$ to $v_k$} or that $W$ is a walk {\em between $v_0$ and $v_k$}.
\end{definition}

\begin{definition}[Directed Paths]
	A {\em directed path} $P$ in a directed graph $D$ is a directed walk $(v_0, e_1, v_1, e_2, \dots, e_k, v_k)$ such that no vertex is repeated, i.e.\ $i \neq j$ implies $v_i \neq v_j$ for all $1 \leq i, j \leq k$.
\end{definition}

\begin{definition}[Directed Cycles]
	A {\em directed cycle} $C$ in a directed graph $D$ is a directed walk $(v_0, e_1, v_1, e_2, \dots, e_k, v_k)$ such that $v_0 = v_k$ and no vertex nor edge is repeated.
\end{definition}

\subsection{Multigraphs}

We are also interested in multigraphs. Essentially, multigraphs are graphs that possibly have multiple edges between the same pair of vertices. More formally, a multigraph $G$ is a pair $(V(G), E(G))$ where $V(G)$ is a set of {\em vertices} and $E(G) \subseteq \{\{v, w\} \mid v,w \in V(G)\}$ is a {\em multiset} of unordered pairs of vertices called {\em edges}. Endpoints, incident edges, adjacent vertices, and self-loops are defined similarly to the case of graphs.

\begin{definition}[Parallel Edges]
	Two edges $e_1, e_2$ are said to be {\em parallel} if they share the same pair of endpoints. An edge is always parallel with itself.
\end{definition}

\begin{definition}[Multiplicity]
	The {\em multiplicity} of an edge $e$ is the number of edges that are parallel with it. The {\em multiplicity} $q(G)$ of a multigraph $G$ is the maximum multiplicity of an edge it contains.
\end{definition}

A graph is simply a multigraph of multiplicity 1.

\begin{example}[A Multigraph]
	The following multigraph has multiplicity 2 because the edge $e$ has multiplicity 2. It also contains a self-loop at the vertex $v$.
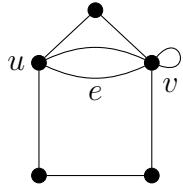
	$$\begin{tikzpicture}
		\node[redvertex] (1) at (0,0) {};
		\node[redvertex] (2) at (1.5,0) {};
		\node[redvertex] (3) at (1.5,1.5) {};
		\node[redvertex] (4) at (0,1.5) {};
		\node[redvertex] (5) at (0.75,2.2) {};
		\node (le) at (0.75, 1.1) {$e$};
		\node (l3) at (1.75,1.2) {$v$};
		\node (l4) at (-0.3,1.5) {$u$};
		\draw (1) to (2) (2) to (3) (1) to (4) (4) to (5) (5) to (3);
		\draw (4) to[bend left=25] (3)
		(4) to[bend right=25] (3);
		\draw (3) edge[out=-20, in=50, loop] (3);
	\end{tikzpicture}$$
\end{example}

Walks, paths, and cycles in multigraphs are defined in the exact same way as they are in graphs (see Definitions~\ref{def:walk}, \ref{def:path}, and \ref{def:cycle}.)

Similarly to graphs, edges in multigraphs can also be replaced with directed edges to form {\em directed multigraphs} in an analogous fashion to directed graphs. Equivalently, directed multigraphs can also be defined by allowing for the arc set $A(D)$ of a directed graph $D$ to be a multiset rather than a set.

\subsection{Graph Colourings}

An important property of a graph is its {\em chromatic number}. In order to define this property, we first define colourings.

\begin{definition}[Colourings]
	Let $G$ be a graph. A {\em colouring} of $G$ is a map $c: V(G) \to \N$. We say $c$ is a {\em $k$-colouring} if $|c(V(G))| = k$.
\end{definition}

\begin{definition}[Proper Colourings]
	A colouring $c$ of a graph $G$ is said to be {\em proper} if $c$ assigns each pair of adjacent vertices different colours, that is, $c(v) \neq c(w)$ whenever $vw \in E(G)$.
\end{definition}

\begin{definition}[Chromatic Number]
	The {\em chromatic number} $\chi(G)$ of a graph $G$ is the minimum integer $k$ such that there exists a proper $k$-colouring of $G$.
\end{definition}

The following fact is folklore.

\begin{proposition}
	A graph $G$ is bipartite if and only if $\chi(G) \leq 2$. \qed
\end{proposition}

\subsection{Matchings}

An important tool that we use to obtain one of our results is Hall's theorem \citep{hall1935representatives}. In order to state Hall's theorem, we first provide some relevant definitions.

\begin{definition}[Matchings]
	Let $G$ be a graph. A {\em matching} $M$ in $G$ is a subset of $E(G)$ such that for each vertex $v$, there exists at most one edge $e \in M$ with $v$ as an endpoint. If such an edge exists, we say $M$ {\em matches} $v$.
\end{definition}

\begin{definition}[Perfect Matchings]
	A matching $M$ of a graph $G$ is said to be {\em perfect} if it matches every vertex of $G$.
\end{definition}

\begin{definition}[Neighbourhood of a Set of Vertices]
	Let $G$ be a graph and $S \subseteq V(G)$ be a subset of vertices. We define the {\em neighbourhood} $N(S)$ of $S$ to be the union of the neighbourhoods of the vertices in $S$, i.e.\ $N(S) \coloneqq \cup_{v \in S} N(v)$.
\end{definition}

We can now state Hall's theorem. The version that we state is not the most general, but is sufficient for our purposes.

\begin{theorem}[Hall's theorem \citep{hall1935representatives}]\label{thm:hall}
	Let $G = (A \cup B, E)$ be a bipartite graph such that $|A| = |B|$. There exists a perfect matching of $G$ if and only if for every subset $S \subseteq A$, we have $|S| \leq |N(S)|$. \qed
\end{theorem}

\section{Fairness and Efficiency}

We now formally define several popular fairness criteria considered in the literature, and highlight some relevant important results.

\subsection{EF and PROP Allocations}

We first define the two classical fairness criteria of envy-freeness and proportionality.

\begin{definition}[Envy]
	Let $\pi = (\pi_1, \pi_2, \dots, \pi_n)$ be an allocation. We say that agent $i$ {\em envies} agent $j$ if $u_i(\pi_i) < u_i(\pi_j)$.
\end{definition}

\begin{definition}[Envy-Free Allocations]
	We say an allocation $\pi = (\pi_1, \pi_2, \dots, \pi_n)$ is {\em \abbrevs{envy-free}{EF}} if no agent envies another.
\end{definition}

\begin{definition}[Proportional Allocations]
	We say an allocation $\pi = (\pi_1, \pi_2, \dots, \pi_n)$ is {\em \abbrevs{proportional}{PROP}} if for each agent $i$, we have $u_i(\pi_i) \geq u_i(M)/n$.
\end{definition}

In the following, we show the logical relationship between envy-freeness and proportionality. Most importantly, envy-freeness implies proportionality when every agent has a subadditive utility function.

\begin{proposition}\label{prop:EF-implies-PROP}
	For subadditive instances, EF implies PROP.
\end{proposition}
\begin{proof}
	Suppose $\pi = (\pi_1, \pi_2, \dots, \pi_n)$ is EF. We show that $u_i(\pi_i) \geq u_i(M)/n$ for each agent $i$. Fix an agent $i$. Because $\pi$ is EF, agent $i$ does not envy any agent $j$. So, the inequality $u_i(\pi_i) \geq u_i(\pi_j)$ holds for each agent $j$. By summing together the $n$ inequalities of this form, we obtain
	\begin{align*}
		n \cdot u_i(\pi_i) &\geq u_i(\pi_1) + u_i(\pi_2) + \dots u_i(\pi_n) \\
		&\geq u_i(M)
	\end{align*}
	where the second inequality follows from the subadditivity of $u_i$. It follows that $u_i(\pi_i) \geq u_i(M)/n$ as required.	
\end{proof}

The converse is not necessarily true. It is possible for an allocation to be PROP but not EF, even if every agent has an additive utility function, which is a special case of subadditive utility functions.

\begin{example}[An Allocation that is PROP but not EF]\label{ex:prop-but-not-ef}
	Consider the 3-agent 3-good instance represented by the following matrix.
	$$\begin{blockarray}{cccc}
		& o_1 & o_2 & o_3  \\
		\begin{block}{c[ccc]}
			1&\mathbf{\underline{1}} & 2 & 0  \\
			2&0 & \mathbf{\underline{1}} & 2  \\
			3&2 & 0 & \mathbf{\underline{1}} \\
		\end{block}
	\end{blockarray}$$

	We claim the allocation $\pi = (\{o_1\}, \{o_2\}, \{o_3\})$ indicated by the bold underlined matrix entries, in which each agent $i$ is allocated $o_i$ is PROP but not EF. For each agent $i$, we have $u_i(\pi_i) = 1 = u_i(M)/n$, so $\pi$ is PROP. On the other hand, because $u_1(\pi_1) = u_1(o_1) = 1 < 2 = u_1(o_2) = u_1(\pi_2)$, agent 1 envies agent 2, so $\pi$ is not EF.
\end{example}

However, for the case of $n=2$ where agents have superadditive utility functions, EF and PROP are equivalent.

\begin{proposition}\label{prop:superadditive-2-agents}
	For superadditive instances with exactly two agents, PROP implies EF.
\end{proposition}
\begin{proof}
	Suppose $\pi = (\pi_1, \pi_2)$ is PROP. By symmetry, it suffices to show that $u_1(\pi_1) \geq u_1(\pi_2)$, i.e.\ agent 1 does not envy agent 2. Because $\pi$ is PROP, we have $u_1(\pi_1) \geq u_1(M)/2$. On the other hand, because $u_1$ is superadditive, we have $u_1(M) = u_1(\pi_1 \cup \pi_2) \geq u_1(\pi_1) + u_1(\pi_2)$, so $u_1(\pi_1 \cup \pi_2) - u_1(\pi_1) \geq u_1(\pi_2)$. Thus, 
	\begin{align*}
		u_1(\pi_1) &\geq u_1(M)/2 \\
		&= u_1(M) - u_1(M)/2 \\
		&= u_1(\pi_1 \cup \pi_2) - u_1(M)/2 \\
		&\geq u_1(\pi_1 \cup \pi_2) - u_1(\pi_1) \\
		&\geq u_1(\pi_2)
	\end{align*}
\end{proof}

One remaining case is when $n \geq 3$ and agents have superadditive utility functions. Example~\ref{ex:prop-but-not-ef} already showed that it is possible for an allocation to be PROP but not EF. We now exhibit an allocation that is EF but not PROP.

\begin{example}[An Allocation that is EF but not PROP]\label{ex:ef-but-not-prop}
	Consider the instance with 3 agents and 3 goods $o_1, o_2, o_3$. We define the utility functions of every agent $i$ as follows:
	$$
	u_i(S) = \begin{cases}
		|S| & \text{if } |S| \leq 2 \\
		9 & \text{if } |S| = 3
	\end{cases}
	$$
	It is straightforward to verify that $u_i$ is superadditive. We claim the allocation $\pi = (\{o_1\}, \{o_2\}, \{o_3\})$ is EF but not PROP. For each pair of agents $i, j$, we have $u_i(\pi_i) = 1 = u_i(\pi_j)$, so $\pi$ is EF. On the other hand, for each agent $i$, we have $u_i(\pi_i) = 1 < 3 = u_i(M)/3$, so $\pi$ is not PROP.
\end{example}

For a visual summary of the relationship between EF and PROP, see Table~\ref{table:ef-prop}.

\begin{table}
	\centering
\begin{tabular}{|c||c|c|}
	\hline
	Utility Functions & $n=2$ & $n\geq 3$ \\
	\hline
	\hline
	Additive & EF $\Leftrightarrow$ PROP (Props~\ref{prop:EF-implies-PROP} and \ref{prop:superadditive-2-agents}) & EF $\implies$ PROP (Prop~\ref{prop:EF-implies-PROP}) \\
	\hline
	Subadditive & EF $\implies$ PROP (Prop~\ref{prop:EF-implies-PROP}) & EF $\implies$ PROP (Prop~\ref{prop:EF-implies-PROP})\\
	\hline
	Superadditive & PROP $\implies$ EF (Prop~\ref{prop:superadditive-2-agents})& \xmark (Examples~\ref{ex:prop-but-not-ef} and \ref{ex:ef-but-not-prop})\\
	\hline
\end{tabular}
\caption{Relationship between EF and PROP depending on $n$ and the type of utility functions. "$\implies$" denotes logical implication. "$\Leftrightarrow$" denotes logical equivalence. "\xmark" denotes the lack of logical implication.}
\label{table:ef-prop}
\end{table}

\subsection{EF1 Allocations}

Unfortunately, both EF and PROP allocations often fail to exist. Even for the simple example of dividing one good between two agents, both EF and PROP allocations fail to exist because inevitably one agent receives nothing. This has led to many other alternative fairness criteria to be introduced. We describe one such criterion which relaxes envy-freeness by allowing an agent (say, Alice) to envy another (say, Bob) as long as Alice's envy toward Bob can be eliminated by disregarding either some good in Bob's bundle or some chore in Alice's bundle. This ensures that while envy inevitably exists, it is relatively small.

\begin{definition}[EF1 Allocations \citep{lipton2004approximately,budish2011combinatorial}]
	We say an allocation $\pi = (\pi_1, \pi_2, \dots, \pi_n)$ is {\em envy-free up to one item (EF1)} if for each pair of agents $i, j$, one of the following conditions hold:
	\begin{itemize}
		\item there exists an item $o_j \in \pi_j$ such that $u_i(\pi_i) \geq u_i(\pi_j \setminus \{o_j\})$, or
		\item there exists an item $o_i \in \pi_i$ such that $u_i(\pi_i \setminus \{o_i\}) \geq u_i(\pi_j)$.
	\end{itemize}
\end{definition}

\citet{caragiannis2019unreasonable} introduced the elegant {\em round-robin} algorithm (Algorithm~\ref{alg:RoundRobin}) for computing an EF1 allocation for goods instances and chores instances. This algorithm accepts an arbitrary linear ordering $\sigma$ of the agents called a {\em picking order}. Then, the algorithm proceeds in multiple rounds. In each round, the agents take turns selecting their preferred item among those remaining, in the order dictated by $\sigma$. The rounds continue until every item has been selected.

\begin{algorithm}
	\caption{\citep{caragiannis2019unreasonable} $\textsc{RoundRobin}(I, \sigma)$}
	\label{alg:RoundRobin}
	
	\hspace*{\algorithmicindent} \textbf{Input:} An additive goods instance or an additive chores instance $I = (N, M, U)$, and a picking order $\sigma$.
	
	\hspace*{\algorithmicindent} \textbf{Output:} An EF1 allocation $\pi = (\pi_1, \pi_2, \dots, \pi_n)$.
	\begin{algorithmic}[1] 
		
		\For{each agent $i = 1, 2, \dots, n$}
		\State $\pi_i \gets \emptyset$ 
		\EndFor
		
		\State $R \gets M$ \Comment{Remaining items}
		
		\State $i \gets$ the first agent in $\sigma$
		\While{$R$ is not empty}
		\State Select any $o^* \in \argmax_{o \in R} u_i(o)$ 
		\State $\pi_i \gets \pi_i \cup \{o^*\}$ \Comment{Give $o^*$ to agent $i$}
		\State $R \gets R \setminus \{o^*\}$
		\State $i \gets$ the next agent in $\sigma$ in a cyclic fashion
		
		\EndWhile

		\State \Return $(\pi_1, \pi_2, \dots, \pi_n)$
	\end{algorithmic}
\end{algorithm}

\begin{example}[Round-Robin]
	Consider the 3-agent 4-good instance represented by the following matrix.
	$$\begin{blockarray}{ccccc}
		& o_1 & o_2 & o_3 & o_4 \\
		\begin{block}{c[cccc]}
			1&1 & 2 & \mathbf{\underline{0}} & \mathbf{\underline{5}}\\
			2&\mathbf{\underline{2}} & 1 & 0 & 2\\
			3& 1 & \mathbf{\underline{1}} & 1 & 0 \\
		\end{block}
	\end{blockarray}$$
	We use the picking order $(1, 2, 3)$. Because there are 3 agents and 4 goods, round-robin requires 2 rounds. In the first round, agent 1 is first to pick, and picks $o_4$ because agent 1 judges it to have utility 5, which is the highest of the remaining goods.
	
	Agent 2 is second to pick, and must pick from the remaining items $o_1, o_2, o_3$. Of these, agent 2 picks $o_1$ as it has the highest utility of 2.
	
	Agent 3 follows and must pick from the remaining two items $o_2, o_3$. Both of them have the same utility of 1 to agent 3, so agent 3 arbitrarily picks $o_2$.
	
	Finally, we proceed to round 2 and begin with agent 1 again. Agent 1 picks the last remaining item $o_3$. This results in the allocation $(\{o_3, o_4\}, \{o_1\}, \{o_2\})$ which we indicate using the bold and underlined matrix entries.
\end{example}

\begin{theorem}[\citet{caragiannis2019unreasonable}]\label{thm:round-robin}
	If $I$ is an additive goods instance or an additive chores instance, then round-robin (Algorithm~\ref{alg:RoundRobin}) returns an EF1 allocation.
\end{theorem}
\begin{proof}
	We prove the theorem for goods instances $I$. Chores instances can be proved analogously. Without loss of generality, assume the picking order is $(1, 2, \dots, n)$. Let $\pi$ denote the allocation resulting from the round-robin algorithm and fix a pair of agents $j < i$. Because $j$ always picks a good before $i$ according to the picking order, agent $j$ does not envy agent $i$. On the other hand, it is possible that $i$ envies $j$. However, by ignoring the first good $o_j$ that $j$ picks during the algorithm, it becomes as if $i$ always picks a good before $j$, so $i$ no longer envies $j$. In other words, we have $u_i(\pi_i) \geq u_i(\pi_j \setminus \{o_j\})$ for the good $o_j \in \pi_j$. Thus, $\pi$ is EF1.
\end{proof}

Theorem~\ref{thm:round-robin} shows that round-robin computes an EF1 allocation under two conditions: (1) the agents have additive utility functions and (2) $I$ is a goods instance or a chores instance. It turns out that both of these conditions are necessary, as shown by Examples~\ref{ex:round-robin-fail-non-additive} and \ref{ex:round-robin-fail-mixed}.

\begin{example}[Round-Robin Fails for Non-Additive Utility Functions]\label{ex:round-robin-fail-non-additive}
	Round-robin fails to compute an EF1 allocation for non-additive instances, even when there are only two agents. Consider the instance $I = (N, M, U)$ containing 2 agents and 6 goods $o_1, o_2, \dots, o_6$. The utility functions are defined as follows.
	
	For agent 1, we let $u_1(S) = |S|$ for each subset $S \subset M$. For agent 2, we let
	$$
	u_2(S) = \begin{cases}
		10 & \text{if } |S \cap \{o_1, o_3, o_5\}| \geq 2 \\
		|S| & \text{otherwise}
	\end{cases}
	$$
	Agent 1 has an additive utility function whereas agent 2 does not. Suppose we use the picking order $(1, 2)$. At each step, each agent is presented with the set of remaining items which are all tied with one another because $u_i(S) = 1$ for any singleton $S$. If the ties are always broken by choosing the item with smallest index, then round-robin produces the allocation $\pi = (\{o_1, o_3, o_5\}, \{o_2, o_4, o_6\})$.
	
	Agent 2 envies agent 1 because $u_2(\pi_2) = 3 < 10 = u_2(\pi_1)$. Moreover, for any item $o \in \pi_1$, the set $\pi_1 \setminus \{o\}$ still contains at least two of $o_1, o_3, o_5$, so $u_2(\pi_2) = 3 < 10 = u_2(\pi_1 \setminus \{o\})$. It follows that $\pi$ is not EF1.
\end{example}

The following example found by \citet{aziz2022fair} shows that round-robin does not always produce an EF1 allocation for mixed instances.

\begin{example}[Round-Robin Fails for Mixed Instances \citep{aziz2022fair}]\label{ex:round-robin-fail-mixed}
	Round-robin fails to compute an EF1 allocation for the following 2-agent instance.
	$$\begin{blockarray}{ccccc}
		& o_1 & o_2 & o_3 & o_4 \\
		\begin{block}{c[cccc]}
			1&\mathbf{\underline{2}} & -3 & \mathbf{\underline{-3}} & -3 \\
			2&2 & \mathbf{\underline{-3}} & -3 & \mathbf{\underline{-3}}\\
		\end{block}
	\end{blockarray}$$

	Using the picking order $(1, 2)$, round-robin produces the allocation $\pi$ in which agent 1 receives a good of utility $2$ and a chore of utility $-3$, and agent 2 receives two chores of utility $-3$, indicated by the bold underlined entries. This allocation is not EF1 because even if agent 2 ignores one of their two chores, agent 2 still envies agent 1, because $u_2(\pi_2 \setminus \{o_4\}) = -3 < -1 = u_2(\pi_1)$.
\end{example}

\citet{aziz2022fair} devised a variant of round-robin called {\em double round-robin} (Algorithm~\ref{alg:DoubleRoundRobin}) for computing an EF1 allocation for mixed instances. Essentially, double round-robin consists of two runs of round-robin, in which the second run uses the picking order of the first in reverse. To begin, the algorithm partitions $M$ into two parts $M^+$ and $M^-$, where $M^-$ is the set of objective chores. Then, double round-robin proceeds in two phases. In the first phase, the objective chores are allocated using round-robin with the picking order $(1, 2, \dots, n)$. Then, in the second phase, the remaining items are allocated using round-robin with the reversed picking order $(n, n-1, \dots, 1)$. The resulting allocation is returned.

Double round-robin is equivalent to ordinary round-robin in the special cases of goods instances and chores instances.

\begin{algorithm}
	\caption{\citep{aziz2022fair} $\textsc{DoubleRoundRobin}(I)$}
	\label{alg:DoubleRoundRobin}
	
	\hspace*{\algorithmicindent} \textbf{Input:} An additive instance $I = (N, M, U)$.
	
	\hspace*{\algorithmicindent} \textbf{Output:} An EF1 allocation $\pi$.
	\begin{algorithmic}[1] 
		
		\State $M^- \gets \{o \in M \mid \forall i \in N, u_i(o) \leq 0\}$ \Comment{i.e.\ objective chores}
		\State $M^+ \gets M \setminus M^-$ \label{line:double:define-M+}
		
		\State Allocate $M^-$ using round-robin with picking order $(1, 2, \dots, n)$ \label{line:double:phase1}\Comment{Phase 1}
		\State Allocate $M^+$ using round-robin with picking order $(n, n-1, \dots, 1)$ \label{line:double:phase2} \Comment{Phase 2}
		
		\State \Return the resulting allocation $\pi$
	\end{algorithmic}
\end{algorithm}

Before giving the proof of correctness of double round-robin, we first make a useful observation.

\begin{observation}\label{obs:ef1-dummy-item}
	Let $I = (N, M, U)$ be an additive instance and $I^+$ be the instance obtained from $I$ by introducing a dummy item with zero utility to every agent. Then, the following statements hold.
	\begin{enumerate}
		\item If $I$ has an EF1 allocation, then $I^+$ does as well.
		\item Deleting the dummy item from an EF1 allocation of $I^+$ yields an EF1 allocation of $I$.
	\end{enumerate}
\end{observation}
\begin{proof}
	(1): Suppose $\pi = (\pi_1, \pi_2, \dots, \pi_n)$ is an EF1 allocation of $I$ and let $\pi^+$ be the allocation of $I^+$ obtained by adding the dummy item to an arbitrary bundle of $\pi$. Recall that the EF1 condition between a pair of agents $i, j$ requires there to be an item $o_j \in \pi_j$ such that $u_i(\pi_i) \geq u_i(\pi_j \setminus \{o_j\})$ or an item $o_i \in \pi_i$ such that $u_i(\pi_i \setminus \{o_i\}) \geq u_i(\pi_j)$. Adding an extra dummy item to either $\pi_i$ or $\pi_j$ does not effect the validity of either of these inequalities because such a dummy item has zero utility to every agent, so the EF1 condition holds between $i$ and $j$ for $\pi^+$.
	
	(2): Similarly to the above, deleting a dummy item from the bundle of $\pi^+$ containing it does not affect either of the two inequalities in the EF1 condition. Thus, (2) holds.
\end{proof}

\begin{theorem}[\citet{aziz2022fair}]\label{thm:double-round-robin}
	If $I$ is an additive instance, then double round-robin (Algorithm~\ref{alg:DoubleRoundRobin}) returns an EF1 allocation.
\end{theorem}
\begin{proof}
	Observation~\ref{obs:ef1-dummy-item} allows us to assume that $|M^-| = an$ for some integer $a \geq 0$ without loss of generality by creating sufficiently many dummy items and adding them to the set $M^-$ immediately after Line~\ref{line:double:define-M+} and deleting them from the output.
	
	Fix a pair of agents $i < j$. Let $c_t(i)$ and $c_t(j)$ denote the $t$-th item allocated to agents $i$ and $j$, respectively, in phase 1 (Line~\ref{line:double:phase1}). Similarly, let $g_t(i)$ and $g_t(j)$ denote the $t$-th item allocated to agents $i$ and $j$, respectively, in phase 2 (Line~\ref{line:double:phase2}). Here, $c$ stands for "chore" and $g$ stands for "good", reflecting the fact that phase 1 allocates objective chores.
	
	We first consider agent $i$'s envy toward agent $j$. Since $i$ always picks before $j$ in phase 1, agent $i$ does not envy agent $j$ with respect to the set of items $M^-$ allocated during phase 1. Hence,
	\begin{equation}\label{eq:double-1}
		u_i(\pi_i \cap M^-) \geq u_i(\pi_j \cap M^-)
	\end{equation}
	However, we also need to consider the items in $M^+$ allocated in phase 2. Since $j$ picks before $i$ and both $j$ and $i$ receives the same number of items in phase 2, any envy that $i$ has toward $j$ with respect to $M^+$ disappears if $i$ ignores the first item that $j$ picks, which is $g_1(j)$. So, we have
	\begin{equation}\label{eq:double-2}
		u_i(\pi_i \cap M^+) \geq u_i((\pi_j \cap M^+) \setminus \{g_1(j)\})
	\end{equation}
	By combining Equations~\ref{eq:double-1} and \ref{eq:double-2}, we obtain $u_i(\pi) \geq u_i(\pi_j \setminus \{g_1(j)\})$.
	
	It remains to consider agent $j$'s envy toward agent $i$. Although the analysis here is similar to the above, we include it for completeness. Agent $j$ does not envy agent $i$ with respect to $M^+$ because $j$ picks before $i$ in phase 2. Hence,
	\begin{equation}\label{eq:double-3}
		u_j(\pi_j \cap M^+) \geq u_j(\pi_i \cap M^+)
	\end{equation}
	On the other hand, it is possible that $j$ envies $i$ with respect to the set $M^-$ of items allocated in phase 1, because $i$ picks before $j$ in phase 1. Note that $M^-$ only contains objective chores because any dummy item that we introduced also counts as an objective chore. For each round $t$ of phase 1 except for the last round, $j$ believes the chore $c_t(j)$ that $j$ picks is at least as good as the chore $c_{t+1}(i)$ that $i$ picks in the subsequent round. In other words, for each chore $c(j)$ that $j$ receives during phase 1, except for the chore $c_a(j)$ that $j$ receives in the final round of phase 1, there exists a different corresponding chore $c(i)$ that $i$ receives such that $u_j(c(j)) \geq u_j(c(i))$. Thus, any envy that $j$ has toward $i$ with respect to $M^-$ disappears if $j$ ignores the chore $c_a(j)$, that is,
	\begin{equation}\label{eq:double-4}
		u_j((\pi_j \cap M^-) \setminus \{c_a(j)\} \geq u_j(\pi_i \cap M^+)
	\end{equation}
	By combining Equations~\ref{eq:double-3} and \ref{eq:double-4}, we obtain $u_j(\pi_j \setminus \{c_a(j)\}) \geq u_j(\pi_i)$. Thus, $\pi$ is EF1.
\end{proof}

Another well-known algorithm for finding EF1 allocations is {\em envy-cycle elimination} introduced by \citet{lipton2004approximately}. Envy-cycle elimination is composed of $m$ rounds, each of which allocates a single additional good to an unenvied agent, until all goods have been allocated.

In order to find an unenvied agent, the algorithm crucially depends on an operation defined on the {\em envy graph} of a partial allocation. Given a partial allocation $\pi$, the {\em envy graph} $G(\pi)$ is defined as the directed graph with vertex set $N$ where the vertices correspond to the agents, and there is a directed edge $(i, j)$ if and only if $u_i(\pi_i) < u_i(\pi_j)$ (i.e.\ $i$ envies $j$). Suppose $C = (i_i, i_2, \dots, i_k)$ is a directed cycle in $G$ (called an {\em envy cycle}). To {\em shift the bundles of $\pi$ along $C$} means to reassign the bundles belonging to the agents on the envy cycle $C$ so that each agent $i_p$ on $C$ now receives the bundle belonging to agent $i_{p+1}$, where $i_{k+1} \coloneqq i_1$. The agents who are not on the envy cycle $C$ are unaffected. In other words, each agent on $C$ receives the bundle of the agent whom they envy with respect to $C$.

Clearly, the shifting operation strictly decreases the envy that any agent on $C$ has toward any other agent on $C$, so the sum of all the envy between any two agents (i.e.\ $\sum_{i, j \in N} \max(u_i(\pi_j) - u_i(\pi_i), 0)$) strictly decreases. It follows that repeatedly applying the shifting operation until there are no more envy cycles results in a partial allocation $\pi$ such that $G(\pi)$ contains no envy cycle. In this case, $G(\pi)$ contains a vertex with zero in-degree, which corresponds to an unenvied agent.

\begin{algorithm}
	\caption{\citep{lipton2004approximately} $\textsc{EnvyCycleElimination}(I)$}
	\label{alg:EnvyCycleElimination}
	
	\hspace*{\algorithmicindent} \textbf{Input:} A monotone goods instance $I = (N, M, U)$.
	
	\hspace*{\algorithmicindent} \textbf{Output:} An EF1 allocation $\pi = (\pi_1, \pi_2, \dots, \pi_n)$.
	\begin{algorithmic}[1] 
		
		\For{each agent $i = 1, 2, \dots, n$}
		\State $\pi_i \gets \emptyset$ 
		\EndFor
		
		\For{each good $o \in M$}

		\State{Construct the envy graph $G(\pi)$}
		
		\While{$G$ contains an envy cycle}\label{line:envy-cycle-while}
		
		\State $C \gets$ an arbitrary envy cycle
		
		\State Shift the bundles of $\pi$ along $C$
		\State{Reconstruct the envy graph $G(\pi)$}
		
		\EndWhile
		
		\State $i \gets$ a vertex of $G(\pi)$ of zero in-degree

		\State $\pi_i \gets \pi_i \cup \{o\}$ \label{line:envy-cycle-allocate}

		\EndFor
		
		\State \Return $\pi = (\pi_1, \pi_2, \dots, \pi_n)$
	\end{algorithmic}
\end{algorithm}

\begin{theorem}[\citet{lipton2004approximately}]\label{thm:envy-cycle-elimination}
	If $I$ is a monotone goods instance, then envy-cycle elimination (Algorithm~\ref{alg:EnvyCycleElimination}) returns an EF1 allocation.
\end{theorem}
\begin{proof}
	The initial partial allocation $(\{o_1\}, \emptyset, \dots, \emptyset)$ is clearly EF1. We show that if the partial allocation $\pi$ at the beginning of a round of envy-cycle elimination is EF1, then the resulting partial allocation $\pi'$ at the end of the round is also EF1.
	
	In each round of envy-cycle elimination, the algorithm constructs the envy graph $G(\pi)$ and continually shifts the bundles of $\pi$ along an envy-cycle $C$ of $G(\pi)$, until no more envy-cycles exist. This process terminates because the total envy between all pairs of agents, i.e.\ $\sum_{i, j \in N} \max(u_i(\pi_j) - u_i(\pi_i), 0)$, strictly decreases each time the bundles are shifted along an envy-cycle. Moreover, shifting the bundles of $\pi$ along $C$ preserves EF1 because each agent receives a bundle that is at least as good as the one they had before, and the $n$ bundles themselves remain the same. Thus, when the while-loop on Line~\ref{line:envy-cycle-while} terminates, $\pi$ is a partial EF1 allocation such that $G(\pi)$ contains no envy-cycles. Because $G(\pi)$ contains no envy-cycles, it contains a vertex $i$ of zero in-degree. By the definition of $G(\pi)$, the vertex $i$ corresponds to an unenvied agent. The algorithm allocates an arbitrary unallocated good $o_j$ to $i$ on Line~\ref{line:envy-cycle-allocate}. Doing so possibly causes other agents to envy $i$, but any envy that results can be alleviated by ignoring the item $o_j$, so the resulting partial allocation is EF1.
\end{proof}

Unfortunately, envy-cycle elimination can fail to produce an EF1 allocation for monotone chores instances as noted by \citet{bhaskar2021approximate}. \citet{bhaskar2021approximate} constructed a concrete example containing 3 agents and 6 chores that illustrates this effect, in which the agents have additive utility functions.

Fortunately, this problem is resolved by \citet{bhaskar2021approximate} by making a small modification to the envy-cycle elimination algorithm. Rather than shifting bundles along an arbitrary envy-cycle, we always select a {\em top-trading envy-cycle} defined as follows. The {\em top-trading envy graph} $G^\text{top}(\pi)$ of a partial allocation is the directed graph with vertex set $N$ and a directed edge $(i, j)$ if and only if $i$ envies $j$ {\em the most} out of every agent. More specifically, $(i, j)$ is a directed edge if and only if $u_i(\pi_i) < u_i(\pi_j)$ and $u_i(\pi_j) = \max_{k \in N}(u_i(\pi_k))$. It is possible for there $i$ to have multiple out-neighbours. Note that $G^\text{top}(\pi)$ is a subgraph of the envy graph $G(\pi)$. This modified version of envy-cycle elimination is called {\em top-trading envy-cycle elimination} (Algorithm~\ref{alg:TopTradingEnvyCycleElimination}).

We briefly explain the problem that ordinary envy-cycle elimination encounters with chores allocation, and why it is avoided by making this modification. Suppose that while shifting bundles along an envy-cycle $C$ in ordinary envy-cycle elimination, an agent $i$ receives the bundle $\pi_j$ from an agent $j$ whom $i$ envies. Further suppose that there is a third agent $k$ whom $i$ envies both before and after the bundle shifting operation. Even though $u_i(\pi_j) > u_i(\pi_i)$, it could be that the the worst chore $c_i$ in $\pi_i$ is worse than the worst chore $c_j$ in $\pi_j$. If this were the case, it is possible that ignoring the chore $c_i$ in $\pi_i$ alleviates $i$'s envy toward $k$, but ignoring the chore $c_j$ in $\pi_j$ does not alleviate $i$'s envy toward $k$. This causes the allocation resulting from the operation to not be EF1.

On the other hand, by shifting bundles along a {\em top-trading} envy-cycle $C$, the agent $i$ always receives a bundle $\pi_j$ from an agent $j$ that $i$ envies the most. After receiving such a bundle, $i$ no longer envies anybody else, so the allocation remains EF1.

\begin{theorem}[\citet{bhaskar2021approximate}]
	If $I$ is a monotone chores instance, then top-trading envy-cycle elimination (Algorithm~\ref{alg:TopTradingEnvyCycleElimination}) returns an EF1 allocation. \qed
\end{theorem}

\begin{algorithm}
	\caption{\citep{bhaskar2021approximate} $\textsc{TopTradingEnvyCycleElimination}(I)$}
	\label{alg:TopTradingEnvyCycleElimination}
	
	\hspace*{\algorithmicindent} \textbf{Input:} A chores instance $I = (N, M, U)$ in which agents have monotone utility functions.
	
	\hspace*{\algorithmicindent} \textbf{Output:} An EF1 allocation $\pi = (\pi_1, \pi_2, \dots, \pi_n)$.
	\begin{algorithmic}[1] 
		
		\For{each agent $i = 1, 2, \dots, n$}
		\State $\pi_i \gets \emptyset$ 
		\EndFor

		\For{each chore $o \in M$}
		\State{Construct the top-trading envy graph $G^\text{top}(\pi)$}
		
		\While{$G^\text{top}(\pi)$ contains a top-trading envy-cycle}
		
		\State $C \gets$ an arbitrary top-trading envy-cycle
		
		\State Shift the bundles of $\pi$ along $C$
		\State{Reconstruct the envy graph $G^\text{top}(\pi)$}
		
		\EndWhile
		
		\State $i \gets$ a vertex of $G^\text{top}(\pi)$ of zero in-degree
		
		\State $\pi_i \gets \pi_i \cup \{o\}$

		\EndFor
		
		\State \Return $\pi = (\pi_1, \pi_2, \dots, \pi_n)$
	\end{algorithmic}
\end{algorithm}

In addition, \citet{bhaskar2021approximate} also presented a further variant (Algorithm~\ref{alg:DoubleEnvyCycleElimination}) of the envy-cycle elimination algorithm that produces an EF1 allocation for doubly monotone instances, thus allowing for both goods and chores. For convenience, we refer to this variant as {\em double envy-cycle elimination} because we are unaware of an established name for this algorithm in the literature.

Similarly to double round-robin (Algorithm~\ref{alg:DoubleRoundRobin}), double envy-cycle elimination (Algorithm~\ref{alg:DoubleEnvyCycleElimination}) also proceeds in two phases, one for objective chores and one for the other items. In the first phase, the set of items that are a good to at least one agent are allocated using envy-cycle elimination, where the envy graph is defined only using the vertices that correspond to agents for whom the item to be allocated is a good. Specifically, if the current round is attempting to allocate an item $o$, then we construct the envy graph using only agents $i$ to whom $o$ is a good, and ignore the other agents. In the second phase, the objective chores are allocated using top-trading envy-cycle elimination (Algorithm~\ref{alg:TopTradingEnvyCycleElimination}).

\begin{algorithm}
	\caption{\citep{bhaskar2021approximate} $\textsc{DoubleEnvyCycleElimination}(I)$}
	\label{alg:DoubleEnvyCycleElimination}
	
	\hspace*{\algorithmicindent} \textbf{Input:} A doubly monotone instance $I = (N, M, U)$.
	
	\hspace*{\algorithmicindent} \textbf{Output:} An EF1 allocation $\pi = (\pi_1, \pi_2, \dots, \pi_n)$.
	\begin{algorithmic}[1] 
		
		\State $M^- \gets \{o \in M \mid \forall i \in N, u_i(o) \leq 0\}$ \Comment{i.e.\ objective chores}
		\State $M^+ \gets M \setminus M^-$
		
		\For{each agent $i = 1, 2, \dots, n$}
		\State $\pi_i \gets \emptyset$ 
		\EndFor
		
		\Statex \textit{// Phase 1}
		\For{each item $o \in M^+$}
		\State $V \gets \{i \in N \mid o \text{ is a good to } i\}$
		\State $G \gets$ the subgraph of the envy graph $G(\pi)$ induced by $V$
		
		\While{$G$ contains an envy-cycle}
		\State $C \gets$ an arbitrary envy-cycle
		\State Shift the bundles of $\pi$ along $C$
		\State{Reconstruct the graph $G$}
		\EndWhile
		
		\State $i \gets$ a vertex of $G$ of zero in-degree
		\State $\pi_i \gets \pi_i \cup \{o\}$
		
		\EndFor

		\Statex \textit{// Phase 2}
		\For{each item $o \in M^-$}
		\State Construct the top-trading envy graph $G^\text{top}(\pi) \gets$
		
		\While{$G^\text{top}(\pi)$ contains a top-trading envy-cycle}
		\State $C \gets$ an arbitrary top-trading envy-cycle
		\State Shift the bundles of $\pi$ along $C$
		\State{Reconstruct the graph $G^\text{top}(\pi)$}
		\EndWhile
		
		\State $i \gets$ a vertex of $G$ of zero in-degree
		\State $\pi_i \gets \pi_i \cup \{o\}$
		
		\EndFor
		
		\State \Return $\pi = (\pi_1, \pi_2, \dots, \pi_n)$
	\end{algorithmic}
\end{algorithm}

\begin{theorem}[\citet{bhaskar2021approximate}]
	If $I$ is a doubly monotone instance, then double envy-cycle elimination (Algorithm~\ref{alg:DoubleEnvyCycleElimination}) returns an EF1 allocation. \qed
\end{theorem}

Notably, double envy-cycle elimination (Algorithm~\ref{alg:DoubleEnvyCycleElimination}) is a strict improvement over double round-robin (Algorithm~\ref{alg:DoubleRoundRobin}). This is because while double round-robin only does so for additive instances as shown by Theorem~\ref{thm:double-round-robin} and Example~\ref{ex:round-robin-fail-non-additive}, double envy-cycle elimination produces EF1 allocations for doubly monotone instances, which are more general than additive instances. (Recall that double round-robin is equivalent to round-robin for goods instances.)

\subsection{EFX Allocations}

EFX allocations are another popular relaxation of EF allocations. They are similar to EF1 allocations in spirit, but are much more stringent. The definition of EFX allocations we give in this section only apply to {\em additive instances}. Adapting the definition to account for non-additive instances require additional technical considerations that are beyond the scope of this work. For a discussion on EFX allocations for non-additive instances, we refer the reader to \citet{berczi2024envy}.

As warm-up, we first give the original definition of EFX allocations for goods instances, before giving its generalizations to chores instances.

\begin{definition}[EFX$_-$ Allocations of Goods \citep{gourves2014near,caragiannis2019unreasonable}]
	We say that an allocation $\pi = (\pi_1, \pi_2, \dots, \pi_n)$ of goods is {\em envy-free up to any positively-valued good (EFX$_-$)} if for each pair of agents $i, j$ and each item $o_j \in \pi_j$ such that $u_i(o_j) > 0$, we have $u_i(\pi_i) \geq u_i(\pi_j \setminus \{o_j\})$.
\end{definition}

The idea is to allow agent $i$ to envy agent $j$ to a small degree --- agent $i$'s envy must be alleviated if agent $i$ ignores a good in agent $j$'s bundle with the least positive marginal utility to $i$. A stronger version of EFX called EFX$_0$ was introduced by \citet{plaut2020almost}. In the stronger version, agent $i$'s envy must be alleviated if agent $i$ ignores a good in agent $j$'s bundle with the least marginal utility to $i$, even such a good has zero marginal utility.

\begin{definition}[EFX$_0$ Allocations of Goods \citep{plaut2020almost}]
	We say that an allocation $\pi = (\pi_1, \pi_2, \dots, \pi_n)$ of goods is {\em envy-free up to any good (EFX$_0$)} if for each pair of agents $i, j$ and each item $o_j \in \pi_j$, we have $u_i(\pi_i) \geq u_i(\pi_j \setminus \{o_j\})$.
\end{definition}

In this dissertation, we mainly consider the stronger version, EFX$_0$.

The idea behind EFX allocations can adapted to chores instances naturally by allowing agent $i$ to envy agent $j$ as long as agent $i$'s envy can be alleviated by ignoring a chore in agent $i$'s bundle. Similarly to goods, two versions of EFX can be defined.

\begin{definition}[EFX$_-$ Allocations of Chores]
	We say an allocation $\pi = (\pi_1, \pi_2, \dots, \pi_n)$ of chores is {\em envy-free up to any negatively-valued chore (EFX$_-$)} if for each pair of agents $i, j$ and each item $o_i \in \pi_i$ such that $u_i(o_i) < 0$, we have $u_i(\pi_i \setminus \{o_i\}) \geq u_i(\pi_j)$.
\end{definition}

\begin{definition}[EFX$_0$ Allocations of Chores]
	We say an allocation $\pi = (\pi_1, \pi_2, \dots, \pi_n)$ of chores is {\em envy-free up to any chore (EFX$_0$)} if for each pair of agents $i, j$ and each item $o_i \in \pi_i$, we have $u_i(\pi_i \setminus \{o_i\}) \geq u_i(\pi_j)$.
\end{definition}

We remark that the notion of EFX can be further generalized to the mixed instances, but we do not explicitly discuss this case as very little is known, and it is beyond the scope of our work. We refer the interested reader to \citet{berczi2024envy}.

EFX$_0$ allocations can alternatively be defined using {\em strong envy}.

\begin{definition}[Strong Envy]
	Let $\pi = (\pi_1, \pi_2, \dots, \pi_n)$ be an allocation. We say that agent $i$ {\em strongly envies} agent $j$ if one of the following conditions hold:
	\begin{itemize}
		\item for some item $o_j \in \pi_j$ that is a good to $i$, we have $u_i(\pi_i) < u_i(\pi_j \setminus \{o_j\})$; or
		\item for some item $o_i \in \pi_i$ that is a chore to $i$, we have $u_i(\pi_i \setminus \{o_i\}) < u_i(\pi_j)$.
	\end{itemize}
\end{definition}

The next proposition follows immediately.

\begin{proposition}
	For both goods and chores instances, an allocation is EFX$_0$ if and only if no agent strongly envies another. \qed
\end{proposition}

It follows immediately from the definition that both EFX$_0$ and EFX$_-$ allocations are also EF1. On the other hand, EF1 allocations are not necessarily EFX$_-$ or EFX$_0$.

\begin{example}[An EF1 Allocation that is not EFX]
	Consider the instance represented by the following matrix.
	$$\begin{blockarray}{cccc}
		& o_1 & o_2 & o_3 \\
		\begin{block}{c[ccc]}
			1&3 & \mathbf{\underline{2}} & 1 \\
			2&\mathbf{\underline{3}} & 2 & \mathbf{\underline{1}}\\
		\end{block}
	\end{blockarray}$$
	The allocation $\pi = (\{o_2\}, \{o_1, o_3\})$ indicated by the bold underlined entries is EF1 but neither EFX$_-$ nor EFX$_0$. In $\pi$, agent 1 envies agent 2, but not vice versa. This is because agent 1's envy toward agent 2 can be alleviated by ignoring the good $o_1 \in \pi_2$ but not the good $o_3 \in \pi_2$.
\end{example}

Unfortunately, the existence problem of EFX allocations has proven to be notoriously difficult and has even been referred to as "fair division's most enigmatic question" by \citet{procaccia2020technical}. To date, the existence of EFX allocations have only been established in several limited settings, and no counterexample is known for additive instances to the best of our knowledge.

We now highlight some positive results that are currently known.

\begin{theorem}[\citet{plaut2020almost}]\label{thm:efx-goods}
	A goods instance has an EFX$_0$ allocation if any of the following hold.
	\begin{enumerate}
		\item agents have identical monotone utility functions;
		\item $n=2$ and agents have monotone utility functions.
	\end{enumerate}
\end{theorem}
\begin{proof}
	(1): \citet{plaut2020almost} showed that a special kind of allocations called leximin++ allocation is EFX$_0$ in this case. A leximin++ allocation is defined using an optimization problem. First, maximize the utility that the worst-off agent receives, and then subject to that, maximize the size of the bundle allocated to the worst-off agent. Then, maximize the utility that the second worst-off agent receives, and then subject to that, maximize the size of the bundle allocated to the second worst-off agent. Continue this procedure until all agents have been considered. An allocation is leximin++ if it can be obtained this way.
	
	To see that a leximin++ allocation $\pi$ is EFX$_0$, assume $u_i(\pi_1) \leq u_i(\pi_2) \leq \dots \leq u_i(\pi_n)$ for all agents $i$ without loss of generality. If $\pi$ is not EFX$_0$, then there is a pair $i < j$ and a good $o_j \in \pi_j$ such that $u_i(\pi_i) < u_i(\pi_j \setminus \{o_j\}$. In this case, one can move the item $o_j$ from $\pi_j$ to $\pi_i$ and increase both the utility and the size of the $i$-th worst bundle in $\pi$, contradicting the choice of $\pi$ as a leximin++ allocation.
	
	(2): A divide-and-choose approach can be applied here. By pretending both agents have the same utility function, we can use (1) to obtain an allocation $\pi$ in which agent 1 does not strongly agent 2 regardless of which of the two bundles agent 1 receives. By letting agent 2 pick their preferred bundle in $\pi$ first and giving agent 1 the other bundle, we ensure that agent 1 does not strongly envy agent 2 and that agent 2 does not envy agent 1. Such an allocation is EFX$_0$.
\end{proof}

EFX$_0$ allocations also exist for additive 2-valued goods instances (i.e.\ there exist $a \geq b \geq 0$ such that for each agent $i$ and each good $o_j$, we have $u_i(o_j) \in \{a, b\}$).

\begin{theorem}[\citet{amanatidis2021maximum}]\label{thm:efx-additive-2-valued}
	EFX$_0$ allocations exist for additive 2-valued goods instances. \qed
\end{theorem}

Another important setting in which EFX$_0$ allocations have been shown to exist is when the agents rank the goods in the same order. In fact,
envy-cycle elimination (Algorithm~\ref{alg:EnvyCycleElimination}) produces such an allocation.

\begin{theorem}[\citet{plaut2020almost}]
	If $I$ is an additive goods instances with SOP, then envy-cycle elimination (Algorithm~\ref{alg:EnvyCycleElimination})  returns an EFX$_0$ allocation. \qed
\end{theorem}

EFX$_0$ allocations are also known to exist when there are a small number of agents and when there are only a few more goods than agents.

\begin{theorem}[\citet{chaudhury2020efx}]
	Additive goods instances for which $n \leq 3$ have EFX$_0$ allocations. \qed
\end{theorem}

\begin{theorem}[\citet{mahara2024extension}]\label{thm:mahara-n+3}
	Monotone goods instances for which $m \leq n+3$ have EFX$_0$ allocations. \qed
\end{theorem}

Very little is known about the existence of EFX$_0$ allocations for four agents and beyond. However, one direction that has been explored is to allow for some goods to remain unallocated. In fact, leaving a linear number of goods unallocated is sufficient.

\begin{theorem}[\citet{berger2022almost}]\label{thm:additive-n-2}
	Additive goods instances have a partial EFX$_0$ allocation that leaves at most $n-2$ goods unallocated. \qed
\end{theorem}

\citet{mahara2024extension} generalized the above to the class of monotone goods instances.

\begin{theorem}[\citet{mahara2024extension}]\label{thm:monotone-n-2}
	Monotone goods instances have a partial EFX$_0$ allocation that leaves at most $n-2$ goods unallocated. \qed
\end{theorem}

Theorem~\ref{thm:additive-n-2} can be improved for the special case of $n=4$.

\begin{theorem}[\citet{berger2022almost}]
	Additive goods instances for which $n = 4$ have a partial EFX$_0$ allocation $\pi$ that leaves at most one good unallocated. Moreover, if $o_j$ is the unallocated good, then $u_i(o_j) \leq u_i(\pi_i)$ for every agent $i$. \qed
\end{theorem}

Although less is known about chores division, some results have been recently discovered. The following is a chores analogue to Theorem~\ref{thm:efx-goods}.

\begin{theorem}[\citet{berczi2024envy}]
	A chores instance has an EFX$_0$ allocation if any of the following hold.
\begin{enumerate}
	\item agents have identical monotone utility functions;
	\item $n=2$ and agents have monotone utility functions. \qed
\end{enumerate} 
\end{theorem}

Very recently, some researchers have been successful in discovering examples of instances that do not have EFX allocations. \citet{christoforidis2024pursuit} found a 3-agent 6-chore instance in which agents have superadditive utility functions, for which EFX$_0$ allocations fail to exist. \citet{berczi2024envy} found a non-monotone 2-agent 3-item instance in which agents have identical utility functions that, for which EFX$^+_-$  allocations\footnote{EFX$^+_-$ is a generalization of EFX to non-monotone mixed instances. As it is beyond the scope of this work, we refer the reader to \citet{berczi2024envy} for the precise definition.} fail to exist. \citet{hosseini2023fairly} found a 4-agent 7-item instance in which agents have lexicographic\footnote{We refer the reader to \citet{hosseini2023fairly} for the definition of lexicographic utility functions.} utility functions that fails to have an EFX$^+_-$ allocation. Interestingly, all monotone 4-agent 7-good instances have EFX$_0$ allocations as a consequence of Theorem~\ref{thm:mahara-n+3}.

Importantly, all of these examples involve either chores or non-monotone utility functions. Thus, the most important case of additive goods instances is still open at the time of writing.

\subsection{MMS Allocations}\label{subsec:mms}

We now shift our attention to MMS allocations. The MMS criterion differs fundamentally from both EF1 and EFX allocations in the sense that it belong to the class of {\em threshold-based} fairness criteria, as opposed to {\em comparison-based} criteria such as EF1 and EFX. For an instance $I = (N, M, U)$ of the fair allocation problem, we use $\Pi_N(M)$ to denote the set of all possible allocations of $M$ among the agents in $N$. 

\begin{definition}[MMS Thresholds \citep{budish2011combinatorial}]
	The {\em maximin share (MMS) threshold} of an agent $i$ in an instance $I$ is defined to be
	$$
	u_i^\text{MMS}(I) \coloneqq \max_{\pi \in \Pi_N(M)} \min_{j \in N} u_i(\pi_j)
	$$
\end{definition}

In other words, the MMS threshold of an agent is the best possible utility that they can receive, if they were to choose an allocation but must receive the worst bundle. For additive instances in particular, the MMS threshold of an agent $i$ is trivially bounded above by $u_i(M)/n$.

\begin{observation}\label{obs:mms-upper-bound}
	If $I$ is an additive instance, then for each agent $i$, we have $$u_i^\text{MMS}(I) \leq \frac{u_i(M)}{n}$$
\end{observation}
\begin{proof}
	Let $\pi = (\pi_1, \pi_2, \dots, \pi_n)$ be an allocation that maximizes the minimum utility of a bundle according to agent $i$, that is, the value of $\min_{j \in N} u_i(\pi_j)$ is maximized. By the definition of MMS thresholds, we have $u_i(\pi_j) \geq u_i^\text{MMS}(I)$ for each bundle $\pi_j$ of $\pi$. The additivity of $u_i$ implies $u_i(M) = \sum_{j \in N} u_i(\pi_j)$. Thus, we have
	\begin{align*}
		u_i(M) &= \sum_{j \in N} u_i(\pi_j) \\
		&\geq \sum_{j \in N} u_i^\text{MMS}(I) \\
		&= n \cdot u_i^\text{MMS}(I)
	\end{align*}
	The desired inequality follows from dividing both sides by $n$.
\end{proof}

The MMS threshold is defined independently for each agent, so different agents can have different MMS thresholds, as exemplified by Example~\ref{ex:mms-thresholds}.

\begin{example}[MMS Thresholds]\label{ex:mms-thresholds}
	Consider the additive instance $I$ given by the following matrix.
	$$\begin{blockarray}{ccccccc}
		& o_1 & o_2 & o_3 & o_4 & o_5 & o_6 \\
		\begin{block}{c[cccccc]}
			1&1 & 2 & 3 & 4 & 5 & 6 \\
			2&1 & 10 & 6 & 0 & 0 & 0 \\
			3&10 & 1 & 1 & 1 & 1 & 1 \\
		\end{block}
	\end{blockarray}$$
	Since there are 3 agents, an allocation must contain 3 bundles. For agent 1, letting the bundles be $\{o_1, o_6\}$, $\{o_2, o_5\}$, and $\{o_3, o_4\}$ ensures that each bundle has utility 7, which matches the upper bound given by Observation~\ref{obs:mms-upper-bound} because $u_i(M)/3 = 7$. Thus, $u_1^\text{MMS}(I) = 7$.
	
	For agent 2, the last three items $o_4, o_5, o_6$ have no utility, so they do not affect the minimum utility of a bundle. The worst bundle in the allocation $(\{o_1\}, \{o_2\}, \{o_3, o_4, \dots, o_6\})$ has utility 1. It is clearly not possible to improve the worst bundle, so $u_2^\text{MMS}(I) = 1$.
	
	For agent 3, the worst bundle in the allocation $(\{o_1\}, \{o_2, o_3\}, \{o_4, o_5, o_6\})$ has utility 2. This is the best possible worst utility because if three of the items with utility 1 are put together into one bundle, then there are only two such items left. Thus, $u_3^\text{MMS}(I) = 2$.
\end{example}

An MMS allocation is one that provides each agent a bundle whose utility meets their respective MMS threshold.

\begin{definition}[MMS Allocations \citep{budish2011combinatorial}]
	We say that an allocation $\pi = (\pi_1, \pi_2, \dots, \pi_n)$ is {\em maximin share (MMS)} for the instance $I$ if for each agent $i$, we have $u_i(\pi_i) \geq u_i^\text{MMS}(I)$.
\end{definition}

\begin{example}[MMS Allocation]\label{ex:mms-alloc}
	In Example~\ref{ex:mms-thresholds}, we determined that in the additive instance given by the following matrix, the MMS thresholds for agents 1, 2, 3 to be 7, 1, 2, respectively.
	$$\begin{blockarray}{ccccccc}
		& o_1 & o_2 & o_3 & o_4 & o_5 & o_6 \\
		\begin{block}{c[cccccc]}
			1&\underline{\mathbf{1}} & 2 & 3 & 4 & 5 & \underline{\mathbf{6}} \\
			2&1 & \underline{\mathbf{10}} & 6 & 0 & 0 & 0 \\
			3&10 & 1 & \underline{\mathbf{1}} & \underline{\mathbf{1}} & \underline{\mathbf{1}} & 1 \\
		\end{block}
	\end{blockarray}$$
	The allocation $\pi = (\{o_1, o_6\}, \{o_2\}, \{o_3, o_4, o_5\})$ indicated by the bold and underlined entries is MMS because $u_1(\pi_1) = 7 \geq 7 = u_1^\text{MMS}(I)$, $u_2(\pi_2) = 10 \geq 1 = u_2^\text{MMS}(I)$, and $u_3(\pi_3) = 3 \geq 2 = u_3^\text{MMS}(I)$.
\end{example}

It is often useful to refer to allocations that are MMS in the perspective of an agent.

\begin{definition}[MMS for an Agent]
		We say an allocation $\pi = (\pi_1, \pi_2, \dots, \pi_n)$ is {\em maximin share (MMS) for the agent $i$} if for each bundle $\pi_j$, we have $u_i(\pi_j) \geq u_i^\text{MMS}(I)$.
\end{definition}

Any allocation $\pi$ that maximizes $\min_{j \in N} u_i(\pi_j)$ is MMS for the agent $i$. On the other hand, such an allocation is not necessarily unique, as shown by Example~\ref{ex:multi-mms}.

\begin{example}[Multiple MMS Allocations for the Same Agent]\label{ex:multi-mms}
	Let $I$ be the 2-agent 4-good additive instance given by the following matrix.
	$$\begin{blockarray}{ccccc}
		& o_1 & o_2 & o_3 & o_4 \\
		\begin{block}{c[cccc]}
			1&1 & 2 & 3 & 5 \\
			2&1 & 2 & 3 & 5\\
		\end{block}
	\end{blockarray}$$
	Both $\pi = (\{o_1, o_4\}, \{o_2, o_3\})$ and $\pi' = (\{o_1, o_2, o_3\}, \{o_4\})$ are MMS for agent 1. This is because in both of them, the worst bundle of the two has utility 5 to agent 1, which is the best possible utility of a worst bundle when there needs to be two bundles.
\end{example}

Unfortunately, MMS allocations present two problems. First, computing the MMS threshold of an agent can be shown to be NP-hard via a straightforward reduction from the \textsc{Partition} problem, although a \abbrevs{polynomial-time approximation scheme}{PTAS} was found by \citet{woeginger1997polynomial}. Second, MMS allocations can fail to exist even for additive instances as shown by \citet{KPW16}. On the positive side, instances for which they fail to exist seem to be rare --- \citet{BL16} generated 48000 random additive instances with $3 \leq n \leq 5$ and $n \leq m \leq 11$, by drawing item utilities from the uniform distribution over $[0,1]$ and the normal distribution with mean 0.5 and standard deviation 0.1, and showed that MMS allocations exist for all of them.

Because of this, it is interesting to find conditions that guarantee the existence of MMS allocations. It is straightforward to show that MMS allocations exist under the following conditions. The first condition is given without citation because it follow immediately from the definition of MMS allocations.

\begin{theorem}\label{thm:mms-2agents}
	An additive goods instance $I$ has an MMS allocation if
	\begin{enumerate}
		\item Every agent has the same utility function; or
		\item $n = 2$ (\citet{BL16}).  
	\end{enumerate}
\end{theorem}
\begin{proof}
	(1): If every agent has the same utility function, then an allocation that is MMS for one agent is MMS for every agent. Such an allocation is MMS.
	
	(2): The divide-and-choose algorithm can be used to find an MMS allocation when $n=2$ in the following way. Let $\pi$ be an MMS allocation for agent 2. By possibly exchanging $\pi_1$ and $\pi_2$, we can assume $u_1(\pi_1) \geq u_1(\pi_2)$ without loss of generality. So, $u_1(\pi_1) \geq u_1(M)/2$. Moreover, $u_1^\text{MMS}(I) \leq u_1(M)/2$ by Observation~\ref{obs:mms-upper-bound}. Combining these two inequalities yields $u_1(\pi_1) \geq u_1^\text{MMS}(I)$. On the other hand, we have $u_2(\pi_2) \geq u_2^\text{MMS}(I)$ by the choice of $\pi$. Thus, $\pi$ is MMS.
\end{proof}

It is also known that MMS allocations exist when the number of goods is only slightly greater than the number of agents.

\begin{theorem}\label{thm:mms-bounds}
	Let $I$ be an additive goods instance. Then, $I$ has an MMS allocation if:
	\begin{enumerate}
		\item $m \leq n+3$ (\citet{BL16});
		\item $m \leq n+4$ (\citet{KPW16}); or
		\item $m \leq n+5$ (tight) (\citet{feige2021tight}). \qed
	\end{enumerate}
\end{theorem}

\citet{feige2021tight} showed the last bound to be tight by exhibiting an additive 3-agent 9-good instance for which an MMS allocation fails to exist.

\subsection{Pareto Optimality}\label{subsection:PO}

Thus far, we have limited our attention to the goal of fairness. However, another equally important consideration is that of economic efficiency. This is the idea of maximizing potential and minimizing waste. Fairness alone is always trivial to achieve. Consider the extreme example of goods division in which all of the goods are purposefully discarded. In such an allocation, every agent receives nothing, so nobody is envious of another. However, this does not seem to be a desirable solution because the utility of the goods is wasted. For this reason, researchers have considered various notions of economic efficiency in conjunction with fairness.

In fact, we have already seen such a notion. Almost all of the results that we have presented so far require allocations to be complete instead of partial, that is, every item must be allocated. By doing so, one can avoid the problem of goods being wasted illustrated by the above example. Although most of the literature focus on complete allocations, a notable exception is the idea of donating a small subset of goods to an external charity in order to achieve a higher level of fairness than what is otherwise possible. For example, this is done in Theorems~\ref{thm:additive-n-2} and \ref{thm:monotone-n-2}. For a selection of other such examples, see \citep{caragiannis2019envy,caragiannis2022little,chaudhury2021improving,chaudhury2021little}.

Another commonly-sought criterion is that of Pareto optimality (also called Pareto efficiency). This is the idea that if it is possible to redistribute items in a way that at least one agent is better off than before, and no agent is worse off than before, then it is desirable to do so. Such a redistribution is known as a Pareto improvement.

\begin{definition}[Pareto Improvements and Pareto Domination]
	Let $\pi = (\pi_1, \pi_2, \dots, \pi_n)$ and $\pi' = (\pi_1', \pi_2', \dots, \pi_n')$ be allocations. We say that $\pi$ is a {\em Pareto improvement} over $\pi'$ (or, $\pi$ {\em Pareto-dominates} $\pi'$) if the following two statements hold:
	\begin{itemize}
		\item There exists an agent $i$ such that $u_i(\pi_i) > u_i(\pi_i')$; and
		\item For each agent $i$, we have $u_i(\pi_i) \geq u_i(\pi_i')$.
	\end{itemize}
\end{definition}

\begin{definition}[Pareto Optimality/Efficiency]
	Let $\pi$ be an allocation. If no allocation is a Pareto improvement over $\pi$, then we say $\pi$ is {\em \abbrevs{Pareto optimal}{PO}} (or, $\pi$ is {\em Pareto efficient}).
\end{definition}

A stronger variant of PO, called {\em fractional Pareto optimality (fPO)} has also been studied. In order to define fPO, we need to define {\em fractional allocations}. Essentially, while an allocation requires that each item be assigned to exactly one agent, a fractional allocation allows for each item to be split arbitrarily between agents. In this sense, fractional allocations are a relaxation of allocations. The definition of Pareto improvements can be generalized for fractional allocations by letting $\pi$ and $\pi'$ be fractional allocations.

\begin{definition}[Fractional Pareto Optimality/Efficiency]
	Let $\pi$ be an allocation. If no fractional allocation is a fractional Pareto improvement over $\pi$, then we say $\pi$ is {\em \abbrevs{fractionally Pareto optimal}{fPO}} (or, $\pi$ is {\em fractionally Pareto efficient}).
\end{definition}

As mentioned above, economic efficiency is usually considered in conjunction with fairness. An important example of this is due to \citet{caragiannis2019unreasonable}, who showed that EF1+PO allocations always exist for additive goods instances. Here, the notation EF1+PO means that EF1 and PO are satisfied simultaneously. \citet{caragiannis2019unreasonable} obtained their result using an elegant solution based on maximizing a social welfare function called {\em Nash welfare} (also called {\em Nash social welfare}).

\begin{definition}[Nash Welfare \citep{nash1950bargaining,kaneko1979nash}]
	The {\em Nash welfare} of an allocation $\pi = (\pi_1, \pi_2, \dots, \pi_n)$ is the product
	$$
	\NW(\pi) \coloneqq \prod_{i \in N} u_i(\pi_i)
	$$
\end{definition}

An important technical consideration is that it is sometimes the case that the Nash welfare of every allocation is zero. For goods instances, this can occur when there are fewer goods than agents. Because of this, the definition of a maximum Nash welfare allocation first maximizes the number of agents who receive a bundle with positive utility, then maximizes the Nash welfare subject to that.

\begin{definition}[MNW Allocations  \citep{caragiannis2019unreasonable}]
	An allocation $\pi$ is said to be {\em \abbrevs{maximum Nash welfare}{MNW}} if it maximizes $\NW(\pi)$ among all possible allocations. In the case that every allocation has zero Nash welfare, we first find a largest subset $S \subseteq N$ of agents to whom we can guarantee positive utility. Then, we say $\pi$ is MNW if it maximizes the product $\prod_{i \in S} u_i(\pi_i)$.
\end{definition}

\begin{theorem}[\citet{caragiannis2019unreasonable}]\label{thm:MNW-EF1-PO}
	For additive goods instances, MNW implies EF1+PO. \qed
\end{theorem}

\citet{amanatidis2021maximum} showed that for the special case of additive 2-valued goods instances, MNW allocations are in fact EFX$_0$.

\begin{theorem}[\citet{amanatidis2021maximum}]\label{thm:2-valued-efx}
	For additive 2-valued goods instances, MNW implies EFX$_0$. \qed
\end{theorem}

As a consequence of Theorems~\ref{thm:MNW-EF1-PO} and \ref{thm:2-valued-efx}, we obtain the following corollary, which is an improvement over Theorem~\ref{thm:efx-additive-2-valued}.

\begin{corollary}\label{cor:efx-additive-2-valued}
	For additive 2-valued goods instances, MNW implies EFX$_0$+PO. \qed
\end{corollary}

Unfortunately, computing an MNW allocation is NP-hard \citep{ramezani2009nash,roos2010complexity}, so Theorem~\ref{thm:MNW-EF1-PO} does not lead to an efficient algorithm for computing an EF1+PO allocation. Efforts have been made to find ways to approximate MNW in polynomial-time. \citet{cole2015approximating} found a polynomial-time algorithm with an approximation factor of approximately 0.346. Later, \citet{barman2018finding} improved the approximation factor to 0.692. Unfortunately, \citet{lee2017apx} showed that the problem of maximizing Nash welfare is APX-hard, implying that no polynomial-time algorithm can achieve an approximation ratio better than approximately 0.999 unless $P = NP$. \citet{garg2018approximating} later improved the upper bound of the best possible approximation ratio to $1/\sqrt{8/7} < 0.936$. To the best of our knowledge, these approximation bounds are the state-of-the-art at the time of writing.

On the other hand, some researchers have discovered restricted settings for which MNW allocations of goods can be computed efficiently. \citet{amanatidis2021maximum} showed that computing MNW allocations can be done in polynomial time for binary additive instances, but remains NP-hard for 3-valued instances. \citet{garg2022tractable} showed that computing MNW allocations can be done in polynomial time for monotone instances in which each agent considers at most $k = 2$ items to have positive utility, but becomes NP-hard when $k=3$.

Another line of research is to compute EF1+PO directly, without relying on MNW allocations. \citet{barman2018finding} found a pseudopolynomial-time algorithm for computing EF1+PO allocations for additive goods instances. \citet{murhekar2021fair} improved the guarantee to EF1+fPO for the same setting.

On the chores side, \citet{garg2022fair} found a strongly polynomial-time algorithm for computing an EF1+PO allocation for 2-valued additive instances. \citet{garg2023new} found strongly polynomial-time algorithms for computing EF1+fPO allocations for additive instances for which $n \leq 3$, and for additive instances for which there are only two {\em types} of agents (two agents are said to be of the same {\em type} if they share the same utility function). Most recently, \citet{mahara2025existence} showed that EF1+fPO allocations exist for all additive instances, and that EF1+PO allocations can computed in polynomial-time if the number of agents is constant.

Researchers have also studied EFX+PO allocations \citep{hosseini2023fairly,garg2023new,camacho2023generalized,garg2023computing,tao2025existence} and MMS+PO allocations \citep{mcglaughlin2020improving,kulkarni2021indivisible,ebadian2022fairly}.

\section{The Graphical Model}

The graphical model of fair division introduced by \citet{christodoulou2023fair} is subsumed by the general model defined above, and can be considered to be a special case.

Let $I = (N, M, U)$ be a instance of the fair division problem. It is often the case that not every agent cares about every item. Such a situation can arise when items are located in physical space and agents have geographical constraints. This suggests the following definition.

\begin{definition}[Relevant Items]
	An item $o_j$ is said to be {\em relevant to the agent} $i$ if it has non-zero marginal utility to agent $i$, that is, there exists a subset $S \subseteq M$ of items such that $u_i(S \cup \{o_j\}) \neq u_i(S)$.
\end{definition}

The graphical model of fair division makes the assumption that each item is relevant to at most two agents, leading to the so-called {\em graphical instances}.

\begin{definition}[Graphical Instances]
	An instance $I$ of the fair division problem is said to be {\em graphical} if each item is relevant to at most two agents.
\end{definition}

Graphical instances are so named because they can be  represented as a multigraph $G$ on $n$ vertices and $m$ edges, such that each vertex (resp.\ edge) is identified with an agent (resp.\ item), and an edge is incident to a vertex if the edge's corresponding item is relevant to the vertex's corresponding agent. For convenience, we call $G$ the {\em underlying multigraph} of the instance.

\begin{definition}[Underlying Multigraph]
	The multigraph representing a graphical instance is called the {\em underlying multigraph} of the instance.
\end{definition}

It is convenient to describe a graphical instance $I=(N,M,U)$ using its underlying multigraph by writing $I=(G,U)$.

\begin{example}[An Underlying Multigraph]
	The following multigraph represents a graphical instance with 5 agents and 8 items. There are two items that are relevant to both agents 2 and 3 represented by the two parallel edges between them, and one item that is relevant only to agent 3 represented by the self-loop at vertex 3.
	$$
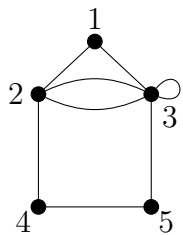
\begin{tikzpicture}
			\node[redvertex] (1) at (0,0) {};
			\node[redvertex] (2) at (1.5,0) {};
			\node[redvertex] (3) at (1.5,1.5) {};
			\node[redvertex] (4) at (0,1.5) {};
			\node[redvertex] (5) at (0.75,2.2) {};
			\node (l3) at (1.75,1.2) {$3$};
			\node (l4) at (-0.3,1.5) {$2$};
			\node (l5) at (0.75, 2.5) {$1$};
			\node (l1) at (-0.2, -0.2) {$4$};
			\node (l2) at (1.7, -0.2) {$5$};
			\draw (1) to (2) (2) to (3) (1) to (4) (4) to (5) (5) to (3);
			\draw (4) to[bend left=25] (3)
			(4) to[bend right=25] (3);
			\draw (3) edge[out=-20, in=50, loop] (3);
		\end{tikzpicture}$$
\end{example}

It is desirable to allocate items only to agents to whom they are relevant --- if an agent does not find an item to be relevant, it is wasteful to allocate it to them. These types of allocations are particularly important and are given a special name.

\begin{definition}[Orientations]
	An allocation $\pi = (\pi_1, \pi_2, \dots, \pi_n)$ is said to be an {\em orientation} if for each bundle $\pi_i$ and each item $o_j \in \pi_i$, the item $o_j$ is relevant to the agent $i$.
\end{definition}

Allocations $\pi$ of a graphical instance $I$ that are orientations correspond to orientations of the underlying multigraph $G$ of $I$ in a natural way. Specifically, $\pi$ corresponds to the orientation of $G$ in which an edge $e$ is directed toward a vertex $i$ if and only if the item corresponding to $e$ is in the bundle $\pi_i$ of $\pi$.

\begin{example}[An Orientation]
	The following graph orientation represents an allocation that is an orientation, in which the item represented by the edge between 1 and 2 is allocated to agent 2, and the item represented by the edge between 2 and 5 is allocated to agent 5.
	$$
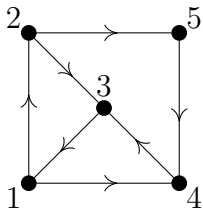
\begin{tikzpicture}
		\node[redvertex] (1) at (0,0) {};
		\node[redvertex] (2) at (0,2) {};
		\node[redvertex] (3) at (1,1) {};
		\node[redvertex] (4) at (2,0) {};
		\node[redvertex] (5) at (2,2) {};
		\node at (-0.2,-0.2) {1};
		\node at (-0.2,2.2) {2};
		\node at (1,1.3) {3};
		\node at (2.2,-0.2) {4};
		\node at (2.2,2.2) {5};
		\draw[arc] (1) to (2);
		\draw[arc] (2) to (3);
		\draw[arc] (4) to (3);
		\draw[arc] (5) to (4);
		\draw[arc] (3) to (1);
		\draw[arc] (1) to (4);
		\draw[arc] (2) to (5);
	\end{tikzpicture}$$
\end{example}

We note that our definition of graphical instances is slightly more general than the original from \citep{christodoulou2023fair}. The original definition requires each item to be relevant to exactly two agents and there to be at most one common relevant item between each pair of agents. Thus, while our definition of graphical instances correspond to multigraphs with self-loops, the graphical instances from \citep{christodoulou2023fair} correspond to simple graphs. We choose the more general definition because some of our results apply to the more general setting.

Recently, many researchers have been interested in the existence of fair allocations and orientations in the graphical model, and the computational complexity of finding such allocations. Below, we survey a small selection of particularly relevant results.

By a {\em graph of goods} (resp.\ {\em graph of chores}), we mean the underlying graph of a graphical instance containing only goods (resp.\ chores).

\begin{theorem}[\citet{christodoulou2023fair}]
	For simple graphs of goods, EFX$_0$ allocations always exist, but deciding whether EFX$_0$ orientations exist is NP-hard. \qed
\end{theorem}

\citet{zeng2024structure} found surprising connections between EFX$_0$ orientations and the chromatic number of a graph.

\begin{theorem}[\citet{zeng2024structure}]
	Bipartite simple graphs of goods have EFX$_0$ orientations, regardless of the utilities of the edges. On the other hand, EFX$_0$ orientations can fail to exist for simple graphs whose chromatic number is 3 or greater.
\end{theorem}

Some results regarding chores and mixed manna have also been obtained.

\begin{theorem}[\citet{zhou2024complete}]
	There exists a polynomial-time algorithm that decides whether an EFX$_-$ orientation exists for a simple graph of chores.
\end{theorem}

\citet{zhou2024complete} also studied the EFX orientation of mixed manna, which is the types of instances that are allowed to contain both goods and chores. They obtained a variety of complexity results for different variants of EFX in this setting. As EFX allocations of mixed manna is beyond the scopes of our work, we refer the reader to \citep{zhou2024complete} for details.

\chapter{Our Contribution}\label{chap:our-contribution}

Our contribution is threefold. In this chapter, we highlight the previous work in the literature that motivated our work, and state our contributions formally.

\section{Existence of MMS Allocations of Mixed Manna}

Our first contribution (Chapter~\ref{chapter:MMS-Mixed-Manna}) concerns the existence of MMS allocations for additive mixed instances and has appeared in the proceedings of ECAI 2024 (see \citep{hsu2024existence}). Recall that these are the instances in which every agent has an additive utility function, and both goods and chores are permitted to exist. We establish the existence of MMS allocations for these instances under a small set of conditions.

As briefly mentioned in Chapter~\ref{chapter:theoretical-overview}, \citet{BL16} showed that when dividing goods, MMS allocations exist if $n=2$ or $m \leq n+3$, where $n, m$ are the numbers of agents and goods, respectively. \citet{KPW16} improved this bound to $m \leq n+4$ while also showing that for all $n \geq 3$, there exists an $n$-agent $(3n+4)$-good fair division instance that has no MMS allocations. A tight result was obtained by \citet{feige2021tight} who improved the bound to $m \leq n+5$ and showed that this is the best possible upper bound on $m$ of the form $n+k$ where $k \in \mathbb{Z}$ by exhibiting an instance with $(n, m) = (3, 9)$ for which an MMS allocation fails to exist.

All of these positive results depend crucially on the following proposition.

\begin{restatable}[\citet{BL16}]{proposition}{restateDeleteOneGood}\label{prop:delete_one_good}
	Let $I$ be an additive goods instance and $I^-$ be the instance obtained from $I$ by deleting an arbitrary agent and an arbitrary good. Then, for each agent $i$ that is not deleted, we have $u_i^\text{MMS}(I^-) \geq u_i^\text{MMS}(I)$. \qed
\end{restatable}

Proposition~\ref{prop:delete_one_good} allows us to give a good $o_j$ to an agent $i$ if $u_i(o_j) \geq u_i^\text{MMS}(I)$, and consider the resulting smaller instance instead, because the MMS thresholds of the remaining agents either stay the same or increase. Unfortunately, Proposition~\ref{prop:delete_one_good} fails when chores come into play, as shown by the following example.

\begin{example}[Proposition~\ref{prop:delete_one_good} Fails for Chores]\label{ex:reduction-fails-for-chores}
	Consider the four additive instances $I, I^+_1, I^+_2, I^+_3$ represented by the matrices below. The MMS threshold of agent 1 is shown on the right. Giving the fourth item to agent 2 in any of $I^+_1, I^+_2, I^+_3$ yields $I$. Observe that the MMS threshold of agent 1 decreases, stays the same, and increases if we do so in $I^+_1, I^+_2, I^+_3$, respectively.
	\begin{align*}
		M(I) &= \begin{bmatrix}
			-1 & -1 & -1
		\end{bmatrix} &u_1^\text{MMS}(I) &= -3 \\
		M(I^+_1) &= \begin{bmatrix}
			-1 & -1 & -1 & -1 \\
			-1 & -1 & -1 & -1
		\end{bmatrix} &u_1^\text{MMS}(I^+_1) &= -2 \\
		M(I^+_2) &= \begin{bmatrix}
			-1 & -1 & -1 & -3 \\
			-1 & -1 & -1 & -3
		\end{bmatrix} &u_1^\text{MMS}(I^+_2) &= -3 \\
		M(I^+_3) &= \begin{bmatrix}
			-1 & -1 & -1 & -4 \\
			-1 & -1 & -1 & -4
		\end{bmatrix} &u_1^\text{MMS}(I^+_3) &= -4
	\end{align*}
\end{example}

Thus, the question of whether MMS allocations exist for chores and mixed instances under the same condition of $m \leq n+5$ arises naturally. Our first contribution makes significant progress toward answering this question.

To state our contribution, we first give some definitions. Let $I$ be an additive instance. A {\em negative agent} (resp.\ {\em non-negative agent}) is an agent whose MMS threshold is negative (resp.\ non-negative). A {\em goods agent} (resp.\ {\em chores agent}) is one for whom every item is a good (resp.\ chore). We show the following.

\begin{restatable}{theorem}{firstMainThm}\label{thm:contribution-1}
	Let $I$ be an additive instance of the fair division problem. Suppose $m \leq n+5$ and at least one of the following conditions hold.
	\begin{enumerate}
		\item $I$ contains $n \leq 3$ agents.
		\item $I$ contains a non-negative agent.
		\item $I$ contains only chores agents.
	\end{enumerate}
	Then, $I$ admits an MMS allocation.
\end{restatable}

Notably, Theorem~\ref{thm:contribution-1}(3) is an existence result for the chores setting. On the other hand, because goods agents are necessarily non-negative, Theorem~\ref{thm:contribution-1}(2) strictly generalizes the $m \leq n+5$ result concerning goods division due to \citet{feige2021tight}.

The tools we develop in Chapter~\ref{chapter:MMS-Mixed-Manna} additionally imply that the only remaining case is the instances with $n \geq 4$ that contain only negative mixed agents. Settling this case positively would imply that the three auxiliary conditions in Theorem~\ref{thm:contribution-1} can be dropped. We discuss this in further detail in Section~\ref{sec:discussion-1}.

\section{EFX Orientations of Multigraphs}

Our second contribution (Chapter~\ref{chapter:EFX-Multigraphs}) concerns EFX$_0$ orientations of multigraphs in which every edge represents a good. We show that the problem of deciding whether an EFX$_0$ orientation exists is NP-hard, even for very restricted settings. To complement this, we also show that multigraphs that avoid a certain structure always have an EFX$_0$ orientation that can be found in polynomial-time. This work will appear in ECAI 2025 \citep{hsu2024efxmanuscript} and is also available on ArXiv \citep{hsu2024efxarxiv}.

Recall that the study of graphical instances was initiated recently by \citet{christodoulou2023fair}. These are the instances that can be represented by a graph in which vertices represent agents, edges represent goods, and a good can only have positive utility to an agent if its corresponding edge is incident to the agent's corresponding vertex. Graphical instances capture the setting where agents only care about goods that are "close" to them, which can arise when the goods are located in physical space and agents have geographical constraints. In such situations, a special kind of allocations called {\em orientations} are particularly desirable --- these are the allocations that allocate each good to an agent to whom it has positive utility, so that no good goes to waste.

Although EFX$_0$ orientations are desirable, their existence is harder to decide than EFX$_0$ allocations. Indeed, \citet{christodoulou2023fair} showed that although EFX$_0$ allocations always exist for simple graphs, deciding whether EFX$_0$ orientations exist is NP-hard. This leads one to ask whether there exists classes of graphs for which the existence of an EFX$_0$ orientation can be decided efficiently.

Recent work has begun to address this. First, \citet{zeng2024structure} found surprising connections between EFX$_0$ orientations and the chromatic number of a graph, showing that EFX$_0$ orientations exist for all bipartite simple graphs, and can only fail to exist for simple graphs whose chromatic number is 3 or greater.

The original graphical model introduced by \citet{christodoulou2023fair} was only defined for graphs. However, this model can be readily generalized to the setting of multigraphs. Our work is among the first to study the generalized setting, which has recently garnered much attention among researchers. For a non-exhaustive list of examples, see \citep{afshinmehr2024efx,amanatidis2024pushing,bhaskar2024efx,christodoulou2025exact,sgouritsa2025existence}.

The first part of our contribution is a hardness result regarding EFX$_0$ orientations. Since simple graphs are multigraphs of multiplicity $1$ without self-loops, the aforementioned result due to \citet{zeng2024structure} regarding simple graphs can be restated as: bipartite multigraphs $G$ of multiplicity $1$ have EFX$_0$ orientations. Thus, it seems plausible that bipartiteness is also sufficient for multigraphs of higher multiplicities $q$ to have EFX$_0$ orientations. Similarly, the problem of finding EFX$_0$ orientations could be tractable if there are few values that utility functions can assume, or if those values fall within a small interval. Our result implies that the problem remains NP-hard under a restrictive setting that incorporates the above considerations. Moreover, it suggests a structure that we call a {\em \abbrevs{non-trivial odd multitree}{NTOM}} that, as we will see, guarantees the existence of EFX$_0$ orientations when avoided.

To give our result, we first give a few definitions informally while deferring the formal definitions to Chapter~\ref{chapter:EFX-Multigraphs}. By a {\em bi-valued symmetric multigraph} we mean a multigraph $G$ such that for each edge $e$, both endpoints agree on the utility of $e$, which can only be one of two possible utilities. A {\em heavy edge} is an edge whose utility is the greater possibility of the two. A {\em heavy component} is a component of $G$ if we only considered the heavy edges and ignored the other edges.

A {\em non-trivial odd multitree (NTOM)} is a multigraph $T$ with the following properties:
\begin{enumerate}
	\item (non-trivial) $T$ contains at least 2 vertices;
	\item (odd) every edge of $T$ has odd multiplicity; and
	\item (multitree) the skeleton of $T$ is a tree (i.e.\ $T$ becomes a tree if we condense each set of parallel edges into a single edge).	
\end{enumerate}

We show the following hardness result.

\begin{restatable}{theorem}{restateContributionTwoA}\label{thm:contribution-2a}
	For any fixed $q \geq 2$, deciding whether a bi-valued symmetric multigraph $G$ of multiplicity $q$ has an EFX$_0$ orientation is NP-complete, even if the following hold:
	\begin{enumerate}
		\item $G$ is bipartite;
		\item $\alpha > q\beta$;
		\item The heavy edges of each heavy component of $G$ induce an NTOM.
	\end{enumerate}
\end{restatable}

We remark that \citet{deligkas2024ef1} also showed that it is NP-complete to decide whether a multigraph $G$ has an EFX$_0$ orientation via an elegant reduction from \textsc{Partition}, even if $G$ contains only 10 vertices. However, our work is independent and provides additional graph theoretical insight.

Theorem~\ref{thm:contribution-2a} suggests that the existence of a heavy component whose heavy edges induce an NTOM is a barrier to finding EFX$_0$ orientations. One might hope that EFX$_0$ orientations can still be found as long as there are few NTOMs. Unfortunately, we are able to find a simple class of examples showing this not to be the case.

\begin{observation}\label{obs:one-tree-example}
	For any $q \geq 1$, there exists a multigraph of multiplicity $q$ containing a unique heavy component whose heavy edges induce an NTOM, that fails to have an EFX$_0$ orientation.
\end{observation}
\begin{proof}
	Let $G$ be the multigraph shown in Figure~\ref{fig:one-tree-example}. By symmetry, $G$ has a unique orientation. By setting $\alpha > q\beta$, the vertex not receiving the heavy edge strongly envies the vertex that does.
\end{proof}

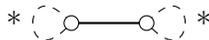
\begin{figure}
	\centering
	\begin{tikzpicture}
		\node[blkvertex] (1) at (0,0) {};
		\node[blkvertex] (2) at (1,0) {};
		\draw[thick] (1) to (2);
		\draw[dashed] (1) to [loop, in=135, out=225, looseness=15] node[left] {*} (1);
		\draw[dashed] (2) to [loop, in=45, out=315, looseness=15] node[right] {*} (2);
	\end{tikzpicture}
	
	\caption{Example of an NTOM with two vertices having no EFX$_0$ orientations. The solid edge represents a heavy edge. The dashed self-loops marked with * each represents $q \geq 1$ light edges. }
	\label{fig:one-tree-example}
\end{figure}

On the other hand, we show that if $G$ does not contain such a heavy component, then it has an EFX$_0$ orientation that can be found in polynomial time.

\begin{restatable}{theorem}{restateContributionTwoB}\label{thm:contribution-2b}
	Let $G$ be a bi-valued symmetric multigraph of multiplicity $q \geq 1$. If $G$ does not contain a heavy component whose heavy edges induce an NTOM, then $G$ has an EFX$_0$ orientation that can be found in polynomial time.
\end{restatable}

A limitation of Theorem~\ref{thm:contribution-2b} is that its highly restrictive setting limits its generality. However, such restrictions seem to be necessary in light of Observation~\ref{obs:one-tree-example}. Moreover, similar assumptions have been made in the existing literature. For example, symmetric valuations have been considered in \citep{christodoulou2023fair,afshinmehr2024efx,deligkas2024ef1}, and bi-valued valuations have been considered in \citep{amanatidis2021maximum,feige2022maximin,garg2023computing}.

\section{Polynomial-Time Algorithms for Fair Orientations of Chores}

Our third contribution (Chapter~\ref{chapter:Fair-Orientations-Chores}) concerns EFX orientations of graphs and multigraphs in which every edge represents a chore. We remark that the work presented in Chapter~\ref{chapter:Fair-Orientations-Chores} is joint work with Valerie King. This work will appear in ECAI 2025 \citep{hsu2025polynomialmanuscript} and is also available on ArXiv \citep{hsu2025polynomialarxiv}.

As previously mentioned, much interest regarding fair orientations of graphs has recently arisen among researchers. Among the many recent works, \citet{zhou2024complete} studied the EFX orientation problem for mixed instances, which is the types of fair division instances that are allowed to contain both goods and chores. They obtained a variety of complexity results for different variants of EFX in this setting, and concluded with a discussion on the special case of chores instances and left a conjecture regarding such instances.

Graphical chores instances can be used to model situations in which tasks are to be allocated to agents, but some tasks may not suitable to an agent because the agent is far away physically. As previous discussed, such a consideration is important for food delivery platforms, because orders should be assigned to workers who are closer to the task location rather than to workers who are far away. There are two variants of EFX that have been defined for chores. We recall their definitions below. Let $\pi$ be an allocation, $\pi_i$ denote the set of chores given to $i$, and $u_i$ denote the utility function of agent $i$. We say that $\pi$ is
\begin{itemize}
	\item EFX$_-$ if for each pair of agents $i \neq j$ such that $i$ envies $j$ and each $e \in \pi_i$ such that $u_i(e) < 0$, we have $u_i(\pi_i \setminus \{e\}) \geq u_i(\pi_j)$; and
	\item EFX$_0$ if for each pair of agents $i \neq j$ such that $i$ envies $j$ and each $e \in \pi_i$ such that $u_i(e) \leq 0$, we have $u_i(\pi_i \setminus \{e\}) \geq u_i(\pi_j)$.
\end{itemize}
The difference is that while EFX$_-$ only requires the envy be alleviated when ignoring chores with strictly negative marginal utility, EFX$_0$ requires that the envy be alleviated even if we ignore a chore with zero marginal utility. Clearly, EFX$_0$ is a strictly stronger condition than EFX$_-$.

\citet{zhou2024complete} found a polynomial-time algorithm that decides whether an EFX$_-$ orientation exists, and conjectured that deciding whether an EFX$_0$ orientation exists is NP-complete, similar to the case involving only goods. 

We resolve their conjecture negatively by giving a polynomial-time algorithm that decides the existence of an EFX$_0$ orientation of chores and outputs one if it exists (Theorem~\ref{thm:mainContributionThree}). We also give a polynomial-time algorithm for the EF1 orientation problem (Theorem~\ref{thm:mainContributionThreeB}).

\begin{restatable}{theorem}{mainContributionThree}\label{thm:mainContributionThree}
	There exists a $O((|V(G)|+|E(G)|)^2)$-time algorithm that decides whether a graph $G$ of chores has an EFX$_0$ orientation and outputs one if it exists.
\end{restatable}

\begin{restatable}{theorem}{mainContributionThreeB}\label{thm:mainContributionThreeB}
	There exists a $O(|V(G)|+|E(G)|)$-time algorithm that decides whether a graph $G$ of chores has an EF1 orientation and outputs one if it exists.
\end{restatable}

For multigraphs, we show that the same problem becomes NP-complete, even for very simple cases (Theorem~\ref{thm:main-multi}). Finally, we show that the EF1 orientation problem remains NP-hard for multigraphs without self-loops (Theorem~\ref{thm:main-no-self-loop}).

\begin{restatable}{theorem}{mainContributionThreeC}\label{thm:main-multi}
	Deciding whether a multigraph $G$ of chores has an EF1 (resp.\ EFX$_0$) orientation is NP-complete, even if $G$ has only two vertices and utility functions are additive and symmetric.
\end{restatable}

\begin{restatable}{theorem}{mainContributionThreeD}\label{thm:main-no-self-loop}
	Deciding whether a multigraph $G$ of chores has an EF1 orientation is NP-complete, even if $G$ has 3 vertices, no self-loops, and utility functions are additive and symmetric.
\end{restatable}

Remarkably, our results for chores stand in contrast with previous results for goods. While EF1 orientations of goods always exist \citep{deligkas2024ef1}, our results imply they can fail to exist for chores (see Example~\ref{ex:chores-no-ef1-graphical}). Second, while deciding whether EFX$_0$ orientations of goods exist is NP-complete \citep{christodoulou2023fair}, out results imply the same problem is polynomial-time decidable for chores.

\begin{example}[Chores Instances with no EF1 Orientations]\label{ex:chores-no-ef1-graphical}
	Consider the complete graph $K_4$ on four vertices, in which each edge represents a chore of utility $-1$ to both of its endpoints.
	$$
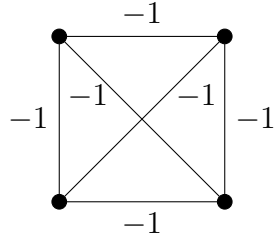
\begin{tikzpicture}
		\node[redvertex] (1) at (0,0) {};
		\node[redvertex] (2) at (0,2.2) {};
		\node[redvertex] (3) at (2.2,2.2) {};
		\node[redvertex] (4) at (2.2,0) {};
		\draw (1) to node[left] {$-1$} (2);
		\draw (2) to node[above] {$-1$} (3);
		\draw (3) to node[right] {$-1$} (4);
		\draw (4) to node[below] {$-1$} (1);
		\draw (1) to node[right,pos=0.65] {$-1$} (3);
		\draw (2) to node[left,pos=0.35] {$-1$} (4);
	\end{tikzpicture}$$
	Since $K_4$ contains 6 edges and 4 vertices, the pigeonhole principle implies that any orientation $\pi$ contains a vertex $i$ with two edges directed toward it. Let $j \neq i$ be a vertex such that the edge $e = (j, i)$ is directed toward $i$. Then, $u_i(\pi_i \setminus \{e\}) \leq -1 < 0 = u_i(\pi_j)$, so $\pi$ is not EF1.
\end{example}

\chapter{Existence of MMS Allocations of Mixed Manna}\label{chapter:MMS-Mixed-Manna}

In this chapter, we prove our first result, which we recall below.

\firstMainThm*

In the remainder of this chapter, we assume all utility functions are additive.

\section{Preliminaries}\label{sec:pre-1}

In this section, we develop the tools required for proving Theorem~\ref{thm:contribution-1}.

\subsection{Same-Order Preference Instances}

An important tool used when studying MMS allocations is the observation made by \citet{BL16} that the instances in which every agent ranks the items in the same order (i.e.\ SOP instances) are the hardest instances to find MMS allocations for. Any instance can be transformed into a corresponding SOP instance as follows.

\begin{definition}[Corresponding SOP Instances]
	Let $I$ be an instance and $M(I)$ be its matrix representation. Let $M(I^\text{SOP})$ denote the matrix obtained from $M(I)$ by sorting the entries of each row in non-increasing order. We define the {\em SOP instance $I^\text{SOP}$ corresponding to} $I$ to be the instance represented by the matrix $M(I^\text{SOP})$.
\end{definition}

By construction, $I^\text{SOP}$ has SOP. 

\begin{example}[Corresponding SOP Instance]
	The following two matrices $M(I)$ and $M(I^\text{SOP})$ represent an additive instance $I$ and its corresponding SOP instance $I^\text{SOP}$, respectively. The two matrices have the same entries in each row, but the rows in $M(I^\text{SOP})$ are sorted in non-increasing order.
	$$M(I) = \begin{bmatrix}
		1 & 0 & 6 & 5 \\
		-1 & 2 & 5 & 2 \\
		5 & 1 & 8 & -5
	\end{bmatrix} \xrightarrow[]{\text{Sort Each Row}} \begin{bmatrix}
		6 & 5 & 1 & 0 \\
		5 & 2 & 2 & -1 \\
		8 & 5 & 1 & -5
	\end{bmatrix} = M(I^\text{SOP})
	$$
\end{example}

\begin{proposition}[\citet{BL16}]\label{prop:sop}
	Let $I$ be an additive instance and $I^\text{SOP}$ be its corresponding SOP instance. If $I^\text{SOP}$ has an MMS allocation, then $I$ does as well. \qed
\end{proposition}

Observe that the MMS threshold of each agent is the same in $I$ as in $I^\text{SOP}$. This is because the MMS threshold of an agent depends only on the number of agents and the item utilities, and not on the order of the item utilities.

A particularly useful observation is that introducing {\em dummy items} with zero utility to a given instance does not affect whether or not an MMS allocation exists. This is because the MMS threshold of each agent is unaffected by the addition of such dummy items. This allows us to assume $m = n+5$ without loss of generality, by adding an appropriate number of dummy items if $m < n+5$. We will often make use of this assumption when working with instances for which $m \leq n+5$.

We further assume that all instances have SOP without loss of generality in light of Proposition~\ref{prop:sop} in the remainder of the chapter.


\subsection{Types of Agents}

We differentiate types of agents by their utility functions and MMS thresholds. An agent $i$ is said to be a {\em goods agent} if $u_i(o_j) \geq 0$ for each item $o_j \in M$, and a {\em chores agent} if $u_i(o_j) \leq 0$ for each item $o_j \in M$. It is possible for an agent to be neither a goods agent nor a chores agent. Such agents are called {\em mixed agents}. We also classify agents depending on the values of their MMS thresholds. An agent is said to be {\em positive}, {\em non-positive}, {\em negative}, or {\em non-negative} if its MMS threshold is positive, non-positive, negative, or non-negative, respectively.


\subsection{Reduction Rules}

Another important tool that is used commonly when studying MMS allocations is that of {\em reduction rules}. Essentially, a reduction rule is an operation that reduces a given instance of the fair division problem into a smaller sub-instance by allocating some of the items to some of the agents.

\begin{definition}[Reduction Rules]
	A {\em reduction rule} $R$ is a function that maps an instance $I = (N, M, U)$ to a pair $(I_R, \pi)$ where
	\begin{enumerate}
		\item $\pi$ is an allocation of a nonempty subset of items $M' \subseteq M$ to a nonempty subset of agents $N' \subseteq N$; and
		\item $I_R = (N_R, M_R, U_R)$ is the instance of the fair division problem obtained from $I$ by deleting the items in $M'$ and the agents in $N'$.
	\end{enumerate}
\end{definition}

To preserve consistency between the agent and item indices in $I$ and $I_R$, we assume that $R$ always deletes agents and items with the greatest indices, starting from agent $n$ and item $o_m$. This assumption can be made without loss of generality by permuting these indices.

If the agents eliminated by $R$ are satisfied with respect to their respective MMS thresholds, and if the MMS thresholds of the remaining agents do not decrease after applying $R$, then we can inductively consider the smaller sub-instance. Such reduction rules are central to finding MMS allocations. We formally define them below.

\begin{definition}[Properties of Reduction Rules]
	Let $R$ be a reduction rule that maps $I = (N, M, U)$ to $(I_R, \pi)$. We say that 
	\begin{enumerate}
		\item $R$ {\em preserves the MMS threshold} of an agent $i \in N_R$ if $u_i^\text{MMS}(I_R) \geq u_i^\text{MMS}(I)$; 
		\item $R$ {\em satisfies} an agent $i \in N \setminus N_R$ if $u_i(\pi_i) \geq u_i^\text{MMS}(I)$ for the bundle $\pi_i$ that $\pi$ allocates to agent $i$;
		\item $R$ is a {\em valid reduction rule} if $R$ preserves the MMS threshold of every agent $i \in N_R$ and satisfies every agent $i \in N \setminus N_R$.
	\end{enumerate}
\end{definition}

\begin{proposition}\label{prop:validreduction}
	Let $I$ be an instance and $R$ be a valid reduction rule. If the sub-instance $I_R$ obtained from $I$ by applying $R$ admits an MMS allocation, then $I$ admits an MMS allocation.
\end{proposition}
\begin{proof}
	Let $\pi$ be the allocation that is implicitly defined by $R$ and $\pi'$ be an MMS allocation for the instance $I_R$. Clearly, applying both $\pi$ and $\pi'$ to the instance $I$ is an MMS allocation for $I$.
\end{proof}

We present two simple reduction rules.

\begin{proposition}\label{prop:1reduction_preserve}
	Let $I$ be an instance and $\pi = (\pi_1, \pi_2, \dots, \pi_n)$ be an allocation that is MMS for an agent $i$. The following two reduction rules preserve the MMS threshold of agent $i$.
	\begin{enumerate}
		\item Allocating an item $o_a$ to an agent $j \neq i$ if $\pi$ contains a bundle $\pi_\ell$ with $u_i(o_a) \leq u_i(\pi_\ell)$.
		\item Allocating two items $o_a, o_b$ to an agent $j \neq i$ if $u_i(\{o_a, o_b\}) \leq u_i^\text{MMS}(I)$.
	\end{enumerate}
\end{proposition}
\begin{proof}
	(1): Without loss of generality, assume $o_a \in \pi_n$. If $\pi_\ell = \pi_n$, then we have $u_i(o_a) \leq u_i(\pi_n)$. Otherwise, we can relabel the bundles of $\pi$ so that $\pi_\ell = \pi_{n-1}$ and $u_i(o_a) \leq u_i(\pi_\ell) = u_i(\pi_{n-1})$. In either case, we have $u_i(o_a) \leq \max(u_i(\pi_{n-1}), u_i(\pi_n))$. We claim the allocation $\pi' = (\pi_1, \pi_2, \dots, \pi_{n-2}, \pi_{n-1} \cup \pi_n \setminus \{o_a\})$ certifies that $u_i^\text{MMS}(I_R) \geq u_i^\text{MMS}(I)$. Since $u_i$ is additive, we have
	\begin{align*}
		u_i(\pi_{n-1} \cup \pi_n \setminus \{o_a\}) &= u_i(\pi_{n-1}) + u_i(\pi_n) - u_i(o_a) \\
		&\geq u_i(\pi_{n-1}) + u_i(\pi_n) \\
		&- \max(u_i(\pi_{n-1}), u_i(\pi_n)) \\
		&= \min(u_i(\pi_{n-1}), u_i(\pi_n)) \\
		&\geq u_i^\text{MMS}(I)
	\end{align*}
	It follows that the utility of each bundle of $\pi'$ is at least $u_i^\text{MMS}(I)$, so $u_i^\text{MMS}(I_R) \geq u_i^\text{MMS}(I)$.
	
	(2): Without loss of generality, assume $o_a, o_b \in \pi_{n-1} \cup \pi_n$. We claim that $\pi' = (\pi_1, \pi_2, \dots, \pi_{n-2}, \pi_{n-1} \cup \pi_n \setminus \{o_a, o_b\})$ certifies that $u_i^\text{MMS}(I_R) \geq u_i^\text{MMS}(I)$. Since $u_i(\pi_n) \geq u_i^\text{MMS}(I) \geq u_i(\{o_a, o_b\})$, we have
	\begin{align*}
		u_i(\pi_{n-1} \cup \pi_n \setminus \{o_a, o_b\}) &= u_i(\pi_{n-1}) + u_i(\pi_n) - u_i(\{o_a, o_b\}) \\
		&\geq u_i(\pi_{n-1}) \\
		&\geq u_i^\text{MMS}(I)
	\end{align*}
	It follows that the utility of each bundle of $\pi'$ is at least $u_i^\text{MMS}(I)$, so $u_i^\text{MMS}(I_R) \geq u_i^\text{MMS}(I)$.
\end{proof}

Using Proposition~\ref{prop:1reduction_preserve}, we can derive two useful valid reduction rules when the allocations that are MMS for each agent have certain properties. The first valid reduction rule is as follows.

\begin{lemma}\label{lemma:1reduction}
	Let $I$ be an instance and suppose that for each agent $i$, there exists an allocation that is MMS for agent $i$ and contains a singleton bundle. Then, there exists a valid reduction rule that allocates a single item to a single agent.
\end{lemma}
\begin{proof}
	For each agent $i$, let $\pi^i$ denote an allocation that is MMS for agent $i$ and contains a singleton bundle. Let $o_a$ denote the item with maximum index $a$ such that $\{o_a\}$ is a bundle in some allocation $\pi^j$. We show that the reduction rule $R$ that allocates $o_a$ to agent $j$ is a valid reduction rule. Since $\{o_a\}$ is a bundle in $\pi^j$, we have $u_j(o_a) \geq u_j^\text{MMS}(I)$. So, $R$ satisfies agent $j$. On the other hand, by the choice of $o_a$, for any agent $i \neq j$, there exists a singleton bundle $\pi_\ell = \{o_b\}$ in $\pi^i$ such that $b \leq a$. The assumption that $I$ has SOP together with the inequality $b \leq a$ imply $u_i(o_a) \leq u_i(\pi_\ell)$. By Proposition~\ref{prop:1reduction_preserve}(1), $R$ preserves the MMS threshold of agent $i$. Thus, $R$ is a valid reduction rule.
\end{proof}

The second valid reduction rule that we will show requires the following technical result, which can be understood as an iterated version of Proposition~\ref{prop:1reduction_preserve}.

\begin{proposition}\label{prop:atmost2reduction}
	Let $I$ be an instance, $X$ be a set of $k$ agents, and $Y$ be a set of $k$ disjoint bundles each of size 1 or 2. Suppose that for some agent $j \notin X$, we have $u_j(\pi_s) \leq u_j^\text{MMS}(I)$ for each bundle $\pi_s$ in $Y$. Let $R$ be any reduction rule that allocates the bundles in $Y$ to agents in $X$ so that each agent in $X$ receives exactly one bundle in $Y$. Then, $R$ preserves the MMS threshold of agent $j$.
\end{proposition}
\begin{proof}
	We prove this by induction on $k$ by repeatedly applying Proposition~\ref{prop:1reduction_preserve}. If $k=0$, then there is nothing to show. Suppose $k=1$ and let $\pi_s$ be the unique bundle in $Y$. Let $\pi$ be an allocation that is MMS for agent $j$ and $\pi_\ell$ be any bundle in $\pi$. By the choice of $\pi_\ell$, we have $u_j(\pi_s) \leq u_j^\text{MMS}(I) \leq u_j(\pi_\ell)$. If $\pi_s$ has size 1, then allocating $\pi_s$ to the agent in $X$ preserves the MMS threshold of agent $j$ by Proposition~\ref{prop:1reduction_preserve}(1). Otherwise if $\pi_s$ has size 2, then allocating $\pi_s$ to the agent in $X$ preserves the MMS threshold of agent $j$ by Proposition~\ref{prop:1reduction_preserve}(2). Assume the proposition holds for all $k < k'$ for some $k'$. We proceed to show that it holds for $k = k'$ as well.
	
	Suppose $R$ allocates some bundle $\pi_s$ in $Y$ to some agent $i$ in $X$. We can view $R$ as being the composition of two reduction rules $R_1$ and $R_2$, where $R_1$ allocates $\pi_s$ to $i$ and $R_2$ allocates $Y \setminus \pi_s$ to $X \setminus \{i\}$ in the same way as $R$. Let $I_1$ be the instance obtained from $I$ by applying $R_1$, and $I_2$ be the instance obtained from $I_1$ by applying $R_2$. According to the induction hypothesis, $R_1$ preserves the MMS threshold of agent $j$, so $u_j^\text{MMS}(I_1) \geq u_j^\text{MMS}(I)$. Similarly, the induction hypothesis implies that $R_2$ preserves the MMS threshold of agent $j$ (when applied to the instance $I_1$), so $u_j^\text{MMS}(I_2) \geq u_j^\text{MMS}(I_1)$. Together, we have $u_j^\text{MMS}(I_2) \geq u_j^\text{MMS}(I)$, so applying $R_1$ followed by $R_2$ on the instance $I$ preserves the MMS threshold of agent $j$. The proposition follows from the fact that applying $R_1$ followed by $R_2$ on the instance $I$ is equivalent to applying $R$ on the instance $I$.
\end{proof}

Using the above proposition, we can derive the following valid reduction rule.

\begin{lemma}\label{lemma:atmost2reduction}
	Let $I$ be an instance and suppose that for some agent $i$, there exists an allocation $\pi$ that is MMS for $i$ and contains $n-1$ bundles each of size 1 or 2. Then, $I$ admits an MMS allocation or there exists a valid reduction rule that allocates to $k$ agents at least $k$ items for some $k>0$.
\end{lemma}
\begin{proof}
	We begin by defining a bipartite graph using the allocation $\pi$. Let $G(\pi) = (A \cup B, E)$ be the bipartite graph where $A$ is the set of agents in the instance $I$ and $B$ is the set of bundles in the allocation $\pi$. For each agent $j$ and each bundle $\pi_k$, the graph $G(\pi)$ contains the edge $\{j, \pi_k\}$ if and only if the bundle $\pi_k$ satisfies the agent $j$. If $G(\pi)$ contains a perfect matching, then allocating each bundle to the agent it is matched with yields an MMS allocation, so we are done.
	
	Assume $G(\pi)$ does not contain a perfect matching. By Hall's Theorem (Theorem~\ref{thm:hall}) there exists a set $Y \subseteq B$ such that $|N(Y)| < |Y|$, where $N(Y) \coloneqq \{j \in A \mid \{j, \pi_k\} \in E \text{ for some } \pi_k \in Y\}$ is the neighbourhood of $Y$. We choose such a set $Y$ with minimal size, that is, $|N(Y)| < |Y|$ and no smaller set $Y'$ has the property that $|N(Y')| < |Y'|$. Note that $|Y| \geq 2$ because every vertex of $Y$ is adjacent to at least one vertex of $X$, namely, the vertex representing agent $i$ for whom $\pi$ is MMS. Since $\pi$ is MMS for agent $i$, the vertex $i$ is adjacent to every vertex of $B$. In particular, $Y$ does not contain an isolated vertex, so we have $0 < |N(Y)| < |Y|$. Hence, $|Y| \geq 2$.
	
	Since $\pi$ contains $n-1$ bundles each of size 1 or 2, at most one bundle has size not equal to 1 or 2. In particular, at most one bundle of $Y$ has size not equal to 1 or 2, so it is possible to delete one bundle from $Y$ to obtain a subset $Y' \subset Y$ such that $|Y'| = |Y| - 1 \geq 1$ such that each bundle of $Y'$ has size 1 or 2. Moreover, by the minimality of $Y$ and Hall's theorem, the subgraph of $G(\pi)$ induced by $Y'$ and $N(Y')$ contains a perfect matching $M$. We claim the reduction rule $R$ that allocates bundles in $Y'$ to agents in $N(Y')$ according to the matching $M$ is a valid reduction rule that allocates to $k$ agents at least $k$ items for some $k > 0$.
	
	Since $|Y'| \geq 1$ and each bundle of $Y'$ is non-empty, the number of items $R$ allocates is at least the number of agents $R$ allocates items to. Hence, $R$ allocates to $k$ agents at least $k$ items for some $k > 0$. Moreover, $R$ satisfies the agents in $N(Y')$ because $M$ matches each of these agents to a bundle that satisfies that agent. It remains to show that $R$ preserves the MMS thresholds of the agents not in $N(Y')$. Let $j$ be an agent not in $N(Y')$. By construction, $Y'$ is a set of disjoint bundles each of size 1 or 2. Moreover, we have $u_j(\pi_k) < u_j^\text{MMS}(I)$ for each bundle $\pi_k$ of $Y'$ by the definition of $G(\pi)$ because $i$ is not in $N(Y')$. Thus, Proposition~\ref{prop:atmost2reduction} implies that $R$ preserves the MMS threshold of agent $j$, so $R$ is valid.
\end{proof}

We remark that the idea involving bundles of sizes 1 and 2 and Hall's Theorem has also been applied independently by Hummel in the proof of Lemma 23 of \citep{hummel2023lower}, in which a version of Lemma~\ref{lemma:atmost2reduction} for goods is implicitly obtained.

\section{Existence of MMS Allocations}

In this section, we prove Theorem~\ref{thm:contribution-1}. Each condition of Theorem~\ref{thm:contribution-1} requires a different technique, so we divide the proof into three subsections.

\subsection{$I$ contains $n \leq 3$ agents}

The proof of Theorem~\ref{thm:contribution-1} relies on applying valid reduction rules to reduce a given instance to a smaller instance. This requires the base cases for the induction to first be shown separately. The $n=2$ case has been settled by Theorem~\ref{thm:mms-2agents}(2). In this subsection, we additionally show the $n=3$ case.

The outline of the proof is as follows. For each agent $i \in [3]$ in an instance $I$, we fix an allocation $\pi^i$ that is MMS for $i$. For each pair of allocations $\pi^i, \pi^j$, we define a special bipartite graph $G(\pi^i, \pi^j)$. By analyzing this graph, we show that whenever some $\pi^i$ contains a bundle of size at most one or two 2-bundles (i.e.\ bundle of size 2), then $I$ admits an MMS allocation. This leaves the remaining case in which each of the three allocations $\pi^i$ contains one 2-bundle and two 3-bundles. In this case, we will reduce the instance to a 3-agent 9-item instance for which each agent $i$, there is an allocation that is MMS for $i$ and contains three 3-bundles. Finally, we will use a proposition due to \citet{feige2021tight} that implies that such an instance admits an MMS allocation.

We now introduce the special bipartite graph. For each agent $i \in [3]$, we fix an allocation $\pi^i$ that is MMS for $i$. For each pair of agents $i \neq j$, we define the bipartite graph $G(\pi^i, \pi^j) = (A \cup B, E)$ as follows. The vertex sets $A, B$ consist of the bundles of $\pi^i, \pi^j$, respectively. The edge set $E$ contains the edge $\{\pi^i_x, \pi^j_y\}$ if and only if the bundles $\pi^i_x, \pi^j_y$ are disjoint.

\begin{proposition}\label{prop:edge}
	Let $I$ be a 3-agent instance. For each agent $i$, let $\pi^i = (\pi^i_1, \pi^i_2, \pi^i_3)$ be an allocation that is MMS for agent $i$. Suppose that for some $i \neq j$, the graph $G(\pi^i, \pi^j)$ contains an edge $\{\pi^i_x, \pi^j_y\}$ and there exists a bundle $\pi^i_z \neq \pi^i_x$ such that $u_i(\pi^i_z) \geq u_i(\pi^j_y)$. Then, $I$ admits an MMS allocation.
\end{proposition}
\begin{proof}
	Without loss of generality, assume $G(\pi^1, \pi^2)$ contains the edge $\{\pi^1_1, \pi^2_1\}$ and that $u_1(\pi^1_2) \geq u_1(\pi^2_1)$. Consider the allocation $\pi = (\pi^1_1, \pi^2_1, M \setminus (\pi^1_1 \cup \pi^2_1))$. This allocation is well-defined because $\pi^1_1$ is disjoint with $\pi^2_1$. We have
	\begin{align*}
		u_1(M \setminus (\pi^1_1 \cup \pi^2_1)) &= u_1(M \setminus \pi^1_1) - u_1(\pi^2_1) \\
		&= u_1(\pi^1_2 \cup \pi^1_3) - u_1(\pi^2_1) \\
		&\geq u_1(\pi^1_2 \cup \pi^1_3) - u_1(\pi^1_2) \\
		&= u_1(\pi^1_3) \geq u_1^\text{MMS}(I)
	\end{align*}
	So, the first and third bundles of $\pi$ both satisfy agent 1. On the other hand, since the bundle $\pi_1^2$ of $\pi$ is also a bundle of $\pi^2$, the additivity of $u_2$ implies that one of the bundles $\pi_1^1$ and $M \setminus (\pi^1_1 \cup \pi^2_1)$ of $\pi$ satisfies agent 2. Without loss of generality, assume $\pi_1^1$ satisfies agent 2. Hence, the first and second bundles of $\pi$ both satisfy agent 2. Finally, at least one of the three bundles of $\pi$ satisfies agent 3. Let agent 3 pick a bundle in $\pi$ satisfying them first. Regardless of which bundle agent 3 picks, it is clearly possible to assign the remaining two bundles to agents 1 and 2 in a way that satisfies them both.
\end{proof}

\begin{corollary}\label{cor:2edges}
	Let $I$ be a 3-agent instance. For each agent $i$, let $\pi^i$ be an allocation that is MMS for agent $i$. If for a pair of agents $i \neq j$, the graph $G(\pi^i, \pi^j)$ contains two edges, then $I$ admits an MMS allocation. In particular, if any $\pi^i$ contains a bundle of size at most one, then $I$ admits an MMS allocation.
\end{corollary}
\begin{proof}
	Without loss of generality, assume $G(\pi^1, \pi^2)$ contains two edges.  Suppose the two edges share an endvertex. Again without loss of generality, assume $\{\pi^1_1, \pi^2_1\}$ and $\{\pi^1_1, \pi^2_2\}$ are edges. By Proposition~\ref{prop:edge}, if $u_1(\pi^1_2) \geq u_1(\pi^2_1)$ or if $u_1(\pi^1_2) \geq u_1(\pi^2_2)$, then $I$ admits an MMS allocation. Otherwise, $u_1(\pi^2_1) > u_1(\pi^1_2)$ and $u_1(\pi^2_2) > u_1(\pi^1_2)$, so the two bundles $\pi^2_1$ and $\pi^2_2$ of the allocation $\pi^2$ both satisfy agent 1. Using the picking order 3, 1, 2 on the $\pi^2$ (i.e.\ let agent 3 pick a bundle satisfying them first, then agent 1, then agent 2) results in an MMS allocation.
	
	Otherwise, the two edges do not share an endvertex. Without loss of generality, assume $\{\pi^1_1, \pi^2_1\}$ and $\{\pi^1_2, \pi^2_2\}$ are edges. Proposition~\ref{prop:edge} implies that if $u_1(\pi^1_2) \geq u_1(\pi^2_1)$ or if $u_1(\pi^1_1) \geq u_1(\pi^2_2)$, then $I$ admits an MMS allocation. Otherwise, we have $u_1(\pi^2_1) > u_1(\pi^1_2)$ and $u_1(\pi^2_2) > u_1(\pi^1_1)$. Again, $\pi^2_1$ and $\pi^2_2$ both satisfy agent 1. Using the picking order 3, 1, 2 on the $\pi^2$ results in an MMS allocation.
	
	Suppose some $\pi^i$ contains a bundle of size at most one, say $\pi^i_1$. For any $j \neq i$, the bundle $\pi^i_1$ can only intersect with at most one bundle of $\pi^j$ because it contains at most one item. Equivalently, the bundle $\pi^i_1$ is disjoint with at least two bundles of $\pi^j$. Hence, the graph $G(\pi^i, \pi^j)$ contains two edges, so $I$ admits an MMS allocation.
\end{proof}

We also make use of the following proposition from \citet{feige2021tight}. Although their work is concerned with goods division, the proof of this proposition depends only on the additivity of the utility functions and the fact that $I$ has SOP --- whether the items are goods or chores makes no difference to its validity. We include their proof below for the sake of completeness.

\begin{proposition}[Proposition 23 in \citep{feige2021tight}]\label{prop:39}
	Let $I$ be a 3-agent 9-item instance. For each agent $i$, let $\pi^i$ be an allocation that is MMS for agent $i$. If each $\pi^i$ contains three 3-bundles, then $I$ admits an MMS allocation.
\end{proposition}
\begin{proof}
	(This proof is due to \citet{feige2021tight}.) Assume $I$ has SOP. For each agent $i$, let $\pi^i = (\pi^i_1, \pi^i_2, \pi^i_3)$ be an allocation that is MMS for $i$.
	
	Suppose that for a pair of agents $i, j$, there exist two bundles $\pi^i_a$ and $\pi^j_b$ that are identical. Without loss of generality, assume $\pi^1_1 = \pi^2_1$. By the additivity of $u_1$, agent $1$ is satisfied with at least one of the other two bundles of $\pi^2$. Hence, agent $1$ is satisfied with two bundles of $\pi^2$ (including $\pi^2_1$), and agent $2$ is satisfied with three bundles of $\pi^2$ by definition. Thus, using the picking order $3, 1, 2$ results in an MMS allocation.
	
	Otherwise, suppose that for a pair of agents $i, j$, there exist two bundles $\pi^i_a$ and $\pi^j_b$ that share exactly two items, that is, they differ in exactly one item. Without loss of generality, assume this is true for bundles $\pi^1_1$ and $\pi^2_1$. Since $I$ has same-order preference, every agent agrees on which of these two bundles is at least as valuable as the other. Assume $u_\ell(\pi^2_1) \geq u_\ell(\pi^1_1)$ for every agent $\ell \in [3]$. Since $u_2(\pi^2_1) \geq u_2(\pi^1_1)$, the additivity of $u_2$ implies that $u_2(\pi^1_2 \cup \pi^1_3) \geq u_2(\pi^2_2 \cup \pi^2_3) \geq 2u_2^\text{MMS}(I)$. This allows us to assume $u_2(\pi^1_2) \geq u_2^\text{MMS}(I)$ without loss of generality.
	
	If $u_3(\pi^1_1) \geq u_3^\text{MMS}(I)$, giving $\pi^1_1$ to agent 3, $\pi^1_2$ to agent 2, and $\pi^1_3$ to agent 1 results in an MMS allocation. Otherwise, $u_3(\pi^1_1) < u_3^\text{MMS}(I)$. In this case, we first give $\pi^1_1$ to agent 1, satisfying them. Without loss of generality, assume the item $o_x \in \pi^1_1 \setminus \pi^2_1$ is in the bundle $\pi^2_2$ of $\pi^2$. Agent 2 is able to partition the remaining items into two bundles that both satisfy agent 2, by replacing the item $o_x$ in $\pi^2_2$ with the item $o_y \in \pi^2_1 \setminus \pi^1_1$ (which is at least as good as $o_x$). One of these two resulting bundles must satisfy agent 3 because $\frac{u_3(M) - u_3(\pi^1_1)}{2} \geq u_3^\text{MMS}(I)$. Thus, an MMS allocation can be obtained by letting agent 3 pick first from these two bundles.
	
	Finally, we consider the remaining case where the intersection between each pair of bundles $\pi^i_a, \pi^j_b$ where $i \neq j$ contains exactly one item. Without loss of generality, assume the least valuable item $o_9$ is in $\pi^i_1$ for each $i \in [3]$. Except for containing $o_9$, the three bundles $\pi^i_1$ are disjoint. Hence, there are two items $o_x$ and $o_y$ that are in none of these bundles, each having utility at least that of $o_9$ to every agent. Giving $\pi^3_1$ to agent 3, $(\pi^2_1 \setminus \{o_9\}) \cup \{o_x\}$ to agent 2, and $(\pi^1_1 \setminus \{o_9\}) \cup \{o_y\}$ to agent 1 results in an MMS allocation.
\end{proof}

\begin{theorem}\label{thm:38}
	Let $I$ be an instance of the fair division. Suppose $n=3$ and $m=8$. Then, $I$ admits an MMS allocation.
\end{theorem}
\begin{proof}
	For each agent $i$, let $\pi^i$ be an allocation that is MMS for $i$. If some $\pi^i$ contains a bundle of size at most one, then we are done by Corollary~\ref{cor:2edges}. Otherwise, each $\pi^i$ contains only bundles of size at least two. Note that because $m=8$, the number of 2-bundles each $\pi^i$ contains is at most two.
	
	We distinguish two cases. In the first case, some $\pi^i$ contains two 2-bundles. Clearly, $\pi^i$ contains two 2-bundles and one 4-bundle. For any $\pi^j$, each of the two 2-bundles of $\pi^i$ intersects with at most two bundles of $\pi^j$. In other words, each of them is disjoint with a bundle of $\pi^j$, so the graph $G(\pi^i, \pi^j)$ contains two edges. By Corollary~\ref{cor:2edges}, $I$ admits an MMS allocation.
	
	In the second case, no $\pi^i$ contains two 2-bundles, so each $\pi^i$ contains one 2-bundle and two 3-bundles. Let $I'$ be the instance obtained from $I$ by adding a dummy item with zero utility to every agent. For each $i$, let ${\pi^i}'$ be the allocation obtained from $\pi^i$ by adding the dummy item to the 2-bundle of $\pi^i$. Clearly, for each agent $i$, the allocation ${\pi^i}'$ is MMS for $i$ for the instance $I'$ and contains three 3-bundles. By Proposition~\ref{prop:39}, $I'$ admits an MMS allocation. Since the addition of a dummy item does not affect the existence of MMS allocations, $I$ admits an MMS allocation as well.
	
\end{proof}

\subsection{$I$ contains a non-negative agent}

In this subsection, we prove Theorem~\ref{thm:main1}, which constitutes the second part of Theorem~\ref{thm:contribution-1}. Our proof uses the following preprocessing operation. Let $I$ be an instance of the fair division problem that contains an agent $i$ for whom $u_i^\text{MMS}(I) \geq 0$. We define the {\em mimicked instance} $I^{(i)}$ {\em of $I$ with respect to agent $i$} as the instance obtained from $I$ by replacing the utility function $u_j$ with $u_i$ for each agent $j$ such that $u_j^\text{MMS}(I) < 0$. Note that every agent in the resulting instance $I^{(i)}$ has non-negative MMS threshold.

\begin{proposition}\label{prop:mimic}
	Let $I$ be an instance of the fair division problem and $i$ be an agent for whom $u_i^\text{MMS}(I) \geq 0$. If $I^{(i)}$ admits an MMS allocation, then $I$ does as well.
\end{proposition}
\begin{proof}
	For each agent $j$, we use $u_j$ and $u_j'$ to denote its utility function in $I$ and $I^{(i)}$, respectively. Let $\pi' = (\pi_1', \pi_2', \dots, \pi_n')$ be an MMS allocation for the instance $I^{(i)}$. We construct a new allocation from $\pi'$ as follows. For each agent $j \neq i$ such that $u_j(\pi_j') < u_j^\text{MMS}(I)$, we move all of the items in the bundle $\pi_j'$ to the bundle $\pi_i'$. We denote the allocation that results when no more items can be moved this way by $\pi = (\pi_1, \pi_2, \dots, \pi_n)$.
	
	We show that $\pi$ is MMS for the instance $I$ by showing that $u_j(\pi_j) \geq u_j^\text{MMS}(I)$ for each agent $j$. The first observation is that each agent $j \neq i$ with $u_j^\text{MMS}(I) \geq 0$ has the same utility function in $I$ and $I^{(i)}$, and hence has the same MMS threshold. Thus, $u_j(\pi_j') \geq u_j^\text{MMS}(I)$ and $\pi$ allocates $\pi_j'$ to agent $j$. Therefore, the only agents that remain to be considered are agents $j \neq i$ such that $u_j^\text{MMS}(I) < 0$ and agent $i$.
	
	First consider the agents $j \neq i$ such that $u_j^\text{MMS}(I) < 0$. In this case, we have $u_j' = u_i$. If $u_j(\pi_j') \geq u_j^\text{MMS}(I)$, then the items in the bundle $\pi_j'$ were not moved to the bundle $\pi_i'$ during the construction of $\pi$ from $\pi'$. Hence, $\pi$ allocates $\pi_j'$ to agent $j$ and satisfies agent $j$ because $u_j(\pi_j') \geq u_j^\text{MMS}(I)$. Otherwise, we have $u_j(\pi_j') < u_j^\text{MMS}(I)$ instead. In this case, the items in the bundle $\pi_j'$ were moved to the bundle $\pi_i'$ during the construction of $\pi$, so $\pi_j = \emptyset$ and we have $u_j(\pi_j) = 0 > u_j^\text{MMS}(I)$.
	
	It remains to consider agent $i$. In the allocation $\pi$, agent $i$ receives $\pi_i'$ together with every bundle $\pi_j'$ such that $u_j(\pi_j') < u_j^\text{MMS}(I)$. We claim that the inequality $u_i(\pi_j') \geq u_i^\text{MMS}(I)$ holds for any such bundle $\pi_j'$. Fix any such bundle $\pi_j'$. Since $\pi'$ is MMS for the instance $I^{(i)}$, we have $u_j'(\pi_j') \geq u_j^\text{MMS}(I^{(i)})$. This implies that $u_j \neq u_j'$ because otherwise, we would have $u_j^\text{MMS}(I) > u_j(\pi_j') = u_j'(\pi_j') \geq u_j^\text{MMS}(I^{(i)}) = u_j^\text{MMS}(I)$, a contradiction. Since $u_j \neq u_j'$, we have $u_j' = u_i$ by construction, so $u_j^\text{MMS}(I^{(i)}) = u_i^\text{MMS}(I^{(i)})$. Thus, $u_i(\pi_j') = u_j'(\pi_j') \geq u_j^\text{MMS}(I^{(i)}) = u_i^\text{MMS}(I^{(i)}) = u_i^\text{MMS}(I)$ as claimed. On the other hand, we also have $u_i(\pi_i') \geq u_i^\text{MMS}(I^{(i)}) = u_i^\text{MMS}(I)$ because $\pi'$ is MMS for $I^{(i)}$. Thus, $\pi_i$ is a union of bundles that each has utility at least $u_i^\text{MMS}(I)$ to agent $i$. Since $u_i^\text{MMS}(I) \geq 0$ and $u_i$ is additive, we have $u_i(\pi_i) \geq u_i^\text{MMS}(I)$.
\end{proof}

\begin{theorem}\label{thm:main1}
	Let $I$ be an instance of the fair division problem. Suppose $m \leq n+5$ and $I$ contains a non-negative agent. Then, $I$ admits an MMS allocation.
\end{theorem}
\begin{proof}
	Without loss of generality, assume $m = n+5$. Since $I$ contains a non-negative agent, there exists an agent $i$ for whom $u_i^\text{MMS}(I) \geq 0$. First suppose that every agent has non-positive MMS threshold. In particular, this implies that $u_i^\text{MMS}(I) = 0$. We show that the allocation $\pi = (\pi_1, \pi_2, \dots, \pi_n)$ which allocates every item to agent $i$ is an MMS allocation. Each agent $j \neq i$ receives no items, so $u_j(\pi_j) = 0 \geq u_j^\text{MMS}(I)$. On the other hand, let $\pi^i$ be an MMS allocation for agent $i$. By definition, we have $u_i(\pi_j^i) \geq u_i^\text{MMS}(I) = 0$ for each bundle $\pi_j^i$ of $\pi^i$. By the additivity of $u_i$, we have $u_i(M) = \sum_{\pi_j^i \in \pi^i} u_i(\pi_j^i) \geq 0 = u_i^\text{MMS}(I)$. Thus, $\pi$ is an MMS allocation for the instance $I$.

	Otherwise, $I$ contains an agent $i$ with positive MMS threshold. Let $I^{(i)}$ be the mimicked instance of $I$ with respect to agent $i$. Note that by construction, $I^{(i)}$ is an instance in which every agent has positive MMS threshold. Proposition~\ref{prop:mimic} implies that if $I^{(i)}$ admits an MMS allocation, then $I$ does as well, so it suffices to show that $I^{(i)}$ admits an MMS allocation.
	
	We show that $I^{(i)}$ admits an MMS allocation by induction on the number $n$ of agents. If $n \leq 3$, then $I^{(i)}$ admits an MMS allocation by Theorem~\ref{thm:mms-2agents}(2) or Theorem~\ref{thm:38}. Assume $n \geq 4$ and that for all $n' < n$, any $n'$-agent $m'$-item instance $I'$ of the fair division problem with $m' \leq n' + 5$ in which every agent has positive MMS threshold admits an MMS allocation.
	
	Suppose that for some agent $j$, there exists an allocation $\pi^j$ that is MMS for agent $j$ and contains no singletons. Since $u_j^\text{MMS}(I^{(i)}) > 0$, every bundle of $\pi^j$ has positive utility to agent $j$. In particular, $\pi^j$ contains no empty bundles, so every bundle of $\pi^j$ contains at least two items. Since $\pi^j$ contains exactly $n$ bundles, we have $2n \leq m \leq n+5$, so $(n, m)$ is one of $(4, 9)$ and $(5, 10)$. It follows that $\pi^j$ contains $n-1$ bundles of size 2. By Lemma~\ref{lemma:atmost2reduction}, $I$ admits an MMS allocation or there exists a valid reduction rule $R$ that allocates to $k$ agents at least $k$ items for some $k > 0$. In the former case, we are done. In the latter case, let $I_R$ be the $n_R$-agent $m_R$-agent instance obtained from $I^{(i)}$ by applying $R$. We claim that $I_R$ admits an MMS allocation. Since $R$ preserves the MMS thresholds of each agent it does not delete from $I^{(i)}$, every agent in the instance $I_R$ is positive. Clearly, $m_R \leq n_R + 5$ and $n_R < n$. Thus, $I_R$ admits an MMS allocation by the induction hypothesis. By Proposition~\ref{prop:validreduction}, the instance $I^{(i)}$ admits an MMS allocation.
	
	Otherwise, there is no agent $j$ for which there exists an allocation $\pi_j$ that is MMS for agent $j$ and contains no singletons. So, for each agent $i$, there exists an allocation that is MMS for agent $i$ and contains a singleton bundle. By Lemma~\ref{lemma:1reduction}, there exists a valid reduction rule $R$ that allocates a single item to a single agent. let $I_R$ be the $n_R$-agent $m_R$-agent instance obtained from $I^{(i)}$ by applying $R$. Clearly, we have $m_R \leq n_R+5$ and every agent in $I_R$ has positive MMS threshold because $R$ is a valid reduction rule. Thus, $I_R$ admits an MMS allocation by the induction hypothesis. By Proposition~\ref{prop:validreduction}, the instance $I^{(i)}$ admits an MMS allocation.
\end{proof}

\subsection{$I$ contains only chores agents}

In this subsection, we prove Theorem~\ref{thm:main2}, which constitutes the third part of Theorem~\ref{thm:contribution-1}. This proof relies on the observation that instances with chores agents always have valid reduction rules as long as $m \leq 2n+1$, regardless of what other types of agents it contains.

\begin{proposition}\label{prop:choresagent}
	Let $I$ be an instance with $m \leq n+5$. If there is a chores agent, then $I$ admits an MMS allocation or there exists a valid reduction rule that allocates to $k$ agents at least $k$ items for some $k > 0$.
\end{proposition}
\begin{proof}
	Without loss of generality, assume $m = n+5$ and agent 1 is a chores agent. If $I$ contains a non-negative agent or $n \leq 3$, then one of  Theorem~\ref{thm:main1}, Theorem~\ref{thm:mms-2agents}(2), and Theorem~\ref{thm:38} implies $I$ admits an MMS allocation and we are done. Thus, we assume $n \geq 4$ and every agent in $I$ is negative. Note that in this case, $I$ having SOP implies $o_m$ is a chore to every agent.
	
	For each agent $i$, let $\pi^i$ be an allocation that is MMS for $i$. Suppose there exists an agent $i$ such that every bundle in the allocation $\pi^i$ contains at least two items. Counting the number of items yields the inequality $m \geq 2n$. Hence, $2n \leq m = n+5$, so $4 \leq n \leq 5$. If $(n,m) = (4,9)$, then $\pi^i$ contains three bundles of size 2 and one bundle of size 3. Otherwise, $(n,m) = (5,10)$ and $\pi^i$ contains five bundles of size 2. In either case, $\pi^i$ contains $n-1$ bundles each of size 2, so Lemma~\ref{lemma:atmost2reduction} implies $I$ admits an MMS allocation or there exists a valid reduction rule that allocates to $k$ agents at least $k$ items for some $k > 0$ and we are done.
	
	Otherwise, for every agent $i$, the allocation $\pi^i$ contains an empty bundle or a singleton bundle. For each $i$, let $\pi^i_\ell$ be a bundle in $\pi^i$ that is either empty or a singleton. We proceed to show that the reduction rule $R$ that allocates $o_m$ to agent $1$ is a valid reduction rule. Assume $o_m$ belongs to the bundle $\pi^1_1$ of $\pi^1$. Since agent 1 considers every item to be a chore, we have $u_1(o_m) \geq u_1(\pi^1_1) \geq u_1^\text{MMS}(I)$, so $R$ satisfies agent 1. It remains to show that $R$ preserves the MMS threshold of every agent $i \neq 1$. Fix any agent $i \neq 1$. Suppose the bundle $\pi^i_\ell$ of $\pi^i$ is empty. Since $o_m$ is a chore to every agent, we have $u_i(o_m) \leq 0 = u_i(\pi^i_\ell)$. Otherwise, $\pi^i_\ell$ is a singleton. In this case, we again have $u_i(o_m) \leq u_i(\pi^i_\ell)$ because $I$ has SOP. By Proposition~\ref{prop:1reduction_preserve}(1), $R$ preserves the MMS threshold of agent $i$.
\end{proof}

\begin{theorem}\label{thm:main2}
	Let $I$ be an instance of the fair division problem. Suppose $m \leq n+5$ and $I$ contains only chores agents. Then, $I$ admits an MMS allocation.
\end{theorem}
\begin{proof}
	Let $I$ be an instance with $m \leq n+5$ and suppose $I$ contains only chores agents. If $n=2$ or $n=3$, then $I$ admits an MMS allocation by Theorem~\ref{thm:mms-2agents}(2) or Theorem~\ref{thm:38}. Assume $n \geq 4$ and that for all $n' < n$, any $n'$-agent $m'$-item instance $I'$ of the fair division problem with $m' \leq n'+5$ that contains only chores agents admits an MMS allocation.
	
	Since $n+5 \leq 2n+1$ for all $n \geq 4$, Proposition~\ref{prop:choresagent} applies to $I$ and implies that $I$ admits an MMS allocation or there exists valid reduction rule $R$ that allocates to $k$ agents at least $k$ items for some $k > 0$. In the first case, we are done. Otherwise, let $I_R$ be the $n_R$-agent $m_R$-item instance obtained from $I$ by applying $R$. Clearly, $m_R \leq n_R+5$. Moreover, since reduction rules do not change the utilities of the remaining items to the remaining agents, a chores agent of $I$ that is not deleted by $R$ remains a chores agent in the instance $I_R$. In particular, $I_R$ contains only chores agents. By the induction hypothesis, $I_R$ admits an MMS allocation. By Proposition~\ref{prop:validreduction}, the instance $I$ admits an MMS allocation.
\end{proof}

\section{Discussion}\label{sec:discussion-1}

In this chapter, we have studied the existence of MMS allocations for the mixed manna setting. Specifically, we showed that for an instance $I$ of the fair allocation problem, as long as $m \leq n+5$ and some auxiliary conditions hold, then an MMS allocation exists. We also showed that these auxiliary conditions can be dropped if $n \leq 3$. Ultimately, it would be interesting to show whether or not these auxiliary conditions can still be eliminated for $n \geq 4$. Showing that they can be eliminated would imply a full generalization of the $m \leq n+5$ bound obtained by \citet{feige2021tight} to the mixed manna setting.

\begin{problem}\label{prob1}
	Let $I$ be an additive instance of the fair allocation problem such that $m \leq n+5$. Is it true that $I$ admits an MMS allocation?
\end{problem}

As an immediate consequence of Theorem~\ref{thm:contribution-1} and Proposition~\ref{prop:choresagent}, the only remaining part of the puzzle is the instances with $n \geq 4$ that contain only negative mixed agents. We can assume they are mixed because negative non-mixed agents are chores agents, and we can use Proposition~\ref{prop:choresagent} to handle chores agents.

Handling negative mixed agents requires strategies of a different flavour than those used in the previous cases, and seems to be particularly difficult. A source of the difficulty is the fact that if an allocation $\pi^i$ is MMS for a negative mixed agent $i$, then it is possible that $\pi^i$ contains an empty bundle. As an example, consider an instance for which $n=4$, $m=9$, and $(u_1(o_1), u_1(o_2), \dots, u_1(o_9)) = (1, 1, 1, 1, 1, 1, -3, -3, -3)$. In this case, the only allocations that are MMS for agent 1 must have one empty bundle and three non-empty bundles, each of which contains two items of utility 1 and one item of utility $-3$. The presence of an empty bundle means that neither Lemma~\ref{lemma:1reduction} nor Lemma~\ref{lemma:atmost2reduction} applies. This is problematic because we can no longer use the valid reduction rules that we have obtained. It would be interesting to know if these remaining instances admit MMS allocations.

\chapter{EFX Orientations of Multigraphs}\label{chapter:EFX-Multigraphs}

In this chapter, we prove the following two results.

\restateContributionTwoA*

\restateContributionTwoB*

\section{Preliminaries}

\subsection{The Multigraph Model}

Recall that in the graphical model, an instance of the fair division problem is represented as a tuple $(G, U)$. Here, $G$ is a multigraph that possibly contains self-loops, in which each vertex represents an agent and each edge represents an item. In this chapter, we are only interested in {\em goods instances}, that is, each edge represents a {\em good} with non-negative marginal utility. 

We make two assumptions about $u_i$. First, each $u_i$ is {\em additive}, i.e.\ for each subset $S \subseteq E$, we have $u_i(S) = \sum_{e \in S} u_i(e)$. Second, if an edge $e$ is not incident to a vertex $i$, then $u_i(e) = 0$.

An instance $(G, U)$ is said to be {\em symmetric} if $u_i(e) = u_j(e)$ whenever $i, j$ are both endpoints of an edge $e$. In this case, we refer to $u_i(e)$ as the {\em weight} of the edge $e$. By abusing terminology, we also say that the multigraph $G$ is symmetric in this case. If there are only two possible edge weights $\alpha > \beta \geq 0$, we say $G$ is a {\em bi-valued symmetric multigraph}. We reserve the symbols $\alpha, \beta$ for edge weights and call edges of weight $\alpha$ {\em heavy edges} and edges of weight $\beta$ {\em light edges}.

A {\em multitree} is a multigraph that is a tree if each set of parallel edges is treated as a single edge. A multitree is {\em non-trivial} if it has more than one vertex and {\em odd} if every edge has odd multiplicity. We are particularly interested in heavy components of $G$ whose heavy edges induce a non-trivial odd multitree (NTOM).

\subsection{Private Envy-Freeness}

Let $G = (V, E)$ be a multigraph. We define the notion of {\em \abbrevs{private envy-freeness}{PEF}}. Let $i \neq j$ be two vertices and $E_{ij}$ be the set of edges between $i$ and $j$, not including self-loops. For a (possibly partial) orientation $\pi$ of $G$ and two vertices $i$ and $j$, we say $\pi$ is {\em privately envy-free (PEF)} between $i$ and $j$ if $u_i(\pi_i \cap E_{ij}) \geq u_i(\pi_j \cap E_{ij})$ and $u_j(\pi_i \cap E_{ij}) \geq u_j(\pi_j \cap E_{ij})$, i.e.\ $i$ and $j$ do not envy each other if we only consider the edges between them. Clearly, if $\pi$ is PEF between $i$ and $j$, then neither $i$ nor $j$ envies the other.

\subsection{Circuit Satisfiability}

A Boolean circuit is a directed acyclic graph with three types of vertices:
\begin{enumerate}
	\item {\em Input vertices} with zero in-degree that are labelled with either a unique Boolean variable or with \texttt{true};
	
	\item {\em Internal vertices} that are each labelled with AND, OR, or NOT;
	
	\item A designated internal vertex called the {\em output vertex}.
\end{enumerate}
Moreover, we require that each internal vertex labelled AND or OR to have exactly 2 in-neighbours, and each internal vertex labelled NOT to have a unique in-neighbour. Clearly, a Boolean circuit represents a Boolean formula in a natural way.

Given a Boolean circuit, \textsc{CircuitSAT} asks whether it is possible to assign Boolean truth values to the input vertices representing Boolean variables in such a way that the output vertex is true (i.e.\ the Boolean formula represented by the Boolean circuit is satisfied).

\section{Proof of Theorem~\ref{thm:contribution-2a}}\label{section:hardness}

In this section, we prove Theorem~\ref{thm:contribution-2a}. Similarly to \citet{christodoulou2023fair}, we also show a reduction from the NP-complete problem \textsc{CircuitSAT} \citep{cook1971complexity}.

There are two notable differences between the multigraphs that our reduction and their reduction produce. First, while their reduction can result in a non-bipartite multigraph $G$ as a consequence of their NOT and TRUE gadgets, we circumvent this by designing new gadgets that ensure $G$ is bipartite. This is important because while bipartite graphs always have EFX$_0$ orientations \citep{zeng2024structure}, our reduction implies that deciding whether a bipartite multigraph has an EFX$_0$ orientations is NP-hard. Second, by replacing the TRUE gadget due to \citet{christodoulou2023fair} with a novel TRUE gadget, we ensure that each vertex of $G$ is incident to at most only two light edges, rather than the three in their reduction. Doing so tightens the gap between $\alpha$ and $\beta$ required to guarantee NP-hardness from $\alpha >3\beta$ to $\alpha >2\beta$.

We first prove a lemma concerning the multigraphs $H$ shown in Figure~\ref{fig:subgadget}, which appears as an induced submultigraph in some of the gadgets we use in Theorem~\ref{thm:contribution-2a}.

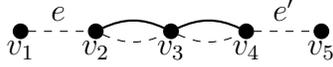
\begin{figure}
	\centering
	\begin{tikzpicture}
		\node[redvertex] (1) at (0,0) {};
		\node[redvertex] (2) at (1,0) {};
		\node[redvertex] (3) at (2,0) {};
		\node[redvertex] (4) at (3,0) {};
		\node[redvertex] (5) at (4,0) {};
		\node (l1) at (0, -0.25) {$v_1$};
		\node (l2) at (1, -0.25) {$v_2$};
		\node (l3) at (2, -0.25) {$v_3$};
		\node (l4) at (3, -0.25) {$v_4$};
		\node (l5) at (4, -0.25) {$v_5$};
		\draw[thick] (2) to[bend left=25] (3)
		(3) to[bend left=25] (4);
		\draw[dashed] (1) to node[above] {$e$} (2)
		(4) to node[above] {$e'$} (5)
		(2) to[bend right=25] (3)
		(3) to[bend right=25] (4);
	\end{tikzpicture}
	\caption{The multigraph $H$. Each solid (resp.\ dashed) edge represents a heavy (resp.\ light) edge.}
	\label{fig:subgadget}
\end{figure}

\begin{lemma}\label{lemma:subgadget}
	Let $G$ be a bi-valued symmetric multigraph that contains $H$(see Figure~\ref{fig:subgadget}) as an induced submultigraph. If $\alpha > 2\beta$ and none of $v_2, v_3, v_4$ are adjacent to a vertex in $G-H$, then $e$ is directed toward $v_1$ or $e'$ is directed toward $v_5$ in any EFX$_0$ orientation of $G$.
\end{lemma}
\begin{proof}
	Fix any EFX$_0$ orientation of $G$ and suppose for contradiction that $e, e'$ are directed toward $v_2, v_4$, respectively. Consider the two heavy edges of $H$. Suppose neither are directed toward $v_3$. Because $v_3$ is not adjacent to a vertex in $G-H$, the maximum utility it can receive is $2\beta$ if it receives all its incident light edges. Since $\alpha > 2\beta$, the vertex $v_3$ envies $v_2$. Since both $e$ and the heavy edge between $v_2$ and $v_3$ are directed toward $v_2$, the vertex $v_3$ still envies $v_2$ even if we disregard $e$, that is, $v_3$ strongly envies $v_2$, a contradiction.
	
	Suppose instead the two heavy edges are both directed in the same direction with respect to Figure~\ref{fig:subgadget}. Without loss of generality, assume they are both directed rightward. In this case, $v_2$ envies $v_3$ regardless of whether it receives the light edge between $v_2$ and $v_3$. This is because $v_2$ is incident to exactly one heavy edge (directed toward $v_3$) and two light edges, as $v_2$ is not adjacent to any vertex in $G-H$. Since $v_2$ does not strongly envy $v_3$ and $v_3$ receives the heavy edge between $v_2$ and $v_3$, all of the light edges incident to $v_3$ are directed away from $v_3$, causing $v_3$ to envy $v_4$. However, because $e'$ is directed toward $v_4$, the vertex $v_3$ strongly envies $v_4$, a contradiction.
	
	Otherwise, both heavy edges are directed toward $v_3$, and we have a contradiction because $v_2$ strongly envies $v_3$.
\end{proof}

\begin{figure}[tb]
	\begin{subfigure}{0.45\textwidth}
		\centering
		\begin{tikzpicture}
			\node[blkvertex] (w) at (0, 0) {};
			\node[redvertex] (v) at (-1, 0) {};
			\node[blkvertex] (a') at (-2, 0) {};
			\node[redvertex] (u) at (1, 0) {};
			\node[blkvertex] (b') at (2, 0) {};
			\node[redvertex] (a) at (-2, 1) {};
			\node[redvertex] (b) at (2, 1) {};
			\node[redvertex] (c) at (0, -1) {};
			\node[blkvertex] (v') at (-1, -1) {};
			\node[blkvertex] (u') at (1, -1) {};
			\node[blkvertex] (c') at (0, -2) {};
			\draw[thick]	(a) to node [right] {$x$} (a') 
			(b') to node [left] {$y$} (b)
			(v) to (w)
			(w) to (u)
			(v) to (v')
			(u) to (u')
			(c) to node [right] {$x \vee y$} (c');
			\draw[dashed]	(a') to (v)
			(w) to (c)
			(u) to (b')
			(a) to[bend right=70] (c')
			(b) to[bend left=70] (c');
	\end{tikzpicture}
	\caption{OR gadget}
	\end{subfigure}
	\begin{subfigure}{0.45\textwidth}
		\centering
		\begin{tikzpicture}
			\node[redvertex] (a) at (0, 0) {};
			\node[blkvertex] (u) at (1, 0) {};
			\node[redvertex] (v) at (2, 0) {};
			\node[blkvertex] (w) at (3, 0) {};
			\node[redvertex] (b) at (4, 0) {};
			
			\node[blkvertex] (a') at (0, -1) {};
			\node[redvertex] (u') at (1, -1) {};
			\node[blkvertex] (v') at (2, -1) {};
			\node[redvertex] (w') at (3, -1) {};
			\node[blkvertex] (b') at (4, -1) {};
			
			\node (la) at (0, 0.25) {$v_1$};
			\node (lu) at (1, 0.25) {$v_2$};
			\node (lv) at (2, 0.25) {$v_3$};
			\node (lw) at (3, 0.25) {$v_4$};
			\node (lb) at (4, 0.25) {$v_5$};
			
			\node (la') at (0, -1.25) {$w_5$};
			\node (lu') at (1, -1.25) {$w_4$};
			\node (lv') at (2, -1.25) {$w_3$};
			\node (lw') at (3, -1.25) {$w_2$};
			\node (lb') at (4, -1.25) {$w_1$};
			
			\draw[thick]	(a) to node[left] {$x$} (a')
			(b) to node[right] {$\neg x$} (b')
			(u) to[bend left=25] (v)
			(v) to[bend left=25] (w)
			(u') to[bend left=25] (v')
			(v') to[bend left=25] (w');
			\draw[dashed]	(a) to (u)
			(w) to (b)
			(a') to (u')
			(w') to (b')
			(u) to[bend right=25] (v)
			(v) to[bend right=25]  (w)
			(u') to[bend right=25] (v')
			(v') to[bend right=25] (w');
			
	\end{tikzpicture}
	\caption{NOT gadget of multiplicity $q \geq 2$}
	\end{subfigure}
	
	\begin{subfigure}{0.45\textwidth}
		\centering
		\begin{tikzpicture}
			\node[redvertex] (a) at (0, 1) {};
			\node[blkvertex] (b) at (0, 0) {};
			\node[redvertex] (c) at (1, 1) {};
			\node[blkvertex] (d) at (1, 0) {};
			\draw[thick]	(a) to node[left] {$x$} (b)
			(c) to node[right] {$x$}(d);
			\draw[dashed]	(a) to (d)
			(b) to (c);
	\end{tikzpicture}
	\caption{Duplication gadget}
	\end{subfigure}
	\begin{subfigure}{0.45\textwidth}
		\centering
		\begin{tikzpicture}
			\node[blkvertex] (1) at (1, 0) {};
			\node[redvertex] (2) at (0.5, 0.866) {};
			\node[blkvertex] (3) at (-0.5, 0.866) {};
			\node[redvertex] (4) at (-1, 0) {};
			\node[blkvertex] (5) at (-0.5, -0.866) {};
			\node[redvertex] (6) at (0.5, -0.866) {};
			\node[redvertex] (7) at (2, 0) {};
			\node[blkvertex] (8) at (3, 0) {};
			\node[redvertex] (9) at (4, 0) {};
			\node (l1) at (1.25, -0.25) {$v_1$};
			\node (l2) at (0.75, 1.066) {$v_2$};
			\node (l3) at (-0.75, 1.066) {$v_3$};
			\node (l4) at (-1.3, 0) {$v_4$};
			\node (l5) at (-0.75, -1.066) {$v_5$};
			\node (l6) at (0.75, -1.066) {$v_6$};
			\node (l7) at (2, -0.25) {$v_7$};
			\node (l7) at (3, -0.25) {$v_8$};
			\node (l7) at (4, -0.25) {$v_9$};
			\draw[thick] (8) to node[above] {\texttt{true}} (9)
			(2) to[bend left=25] (3)
			(3) to[bend left=25] (4)
			(5) to (6)
			(1) to (7);
			\draw[dashed] (1) to (2)
			(4) to (5)
			(6) to node[left] {$*$} (1)
			(7) to (8)
			(2) to[bend right=25] (3)
			(3) to[bend right=25] (4);
	\end{tikzpicture}
	\caption{TRUE gadget}
	\end{subfigure}
	\caption{The gadgets used in Theorem~\ref{thm:contribution-2a}. Each solid (resp.\ dashed) edge represents a heavy (resp.\ light) edge, except for the dashed edge marked with * in (d) which represent $q$ light edges.}
	\label{fig:gadgets}
\end{figure}
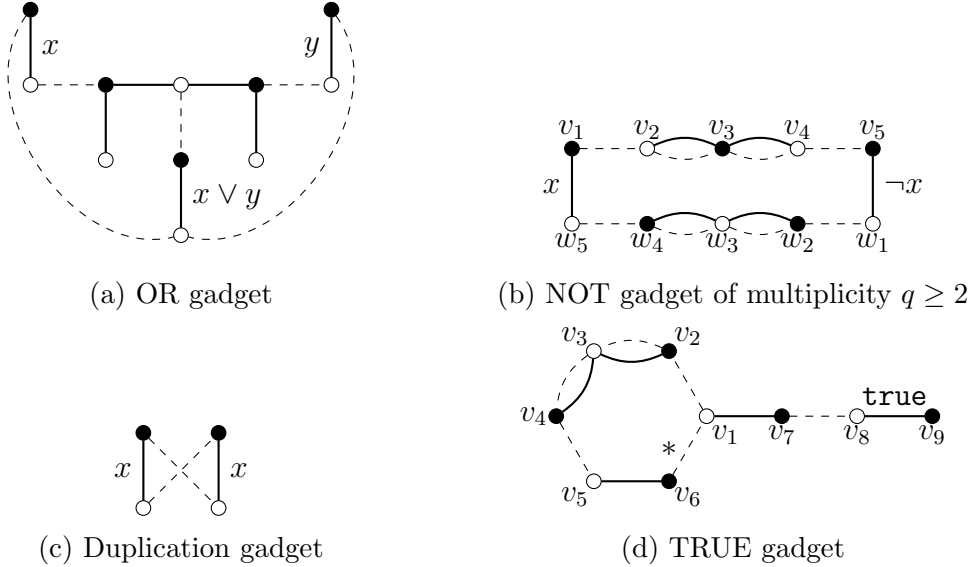

We are now ready to prove Theorem~\ref{thm:contribution-2a}.

\restateContributionTwoA*
\begin{proof}
	To see that this problem is in NP, suppose one is given an orientation $\pi$ of a bi-valued symmetric multigraph $G$. To verify that $\pi$ is EFX$_0$, one can check the EFX$_0$ condition between every pair of agents $i$ and $j$. Determining $u_i(\pi_i)$ and $u_i(\pi_j)$ and determining whether $u_i(\pi_i) \geq u_i(\pi_j \setminus \{e\})$ for each edge $e \in \pi_j$ takes polynomial-time, because it can be done by looping over the set of edges of $G$ a constant number of times.
	
	Before giving our reduction, we first make two assumptions about instances $C$ of \textsc{CircuitSAT}. The first assumption is that an instance uses the TRUE gate at least once --- if it does not, we can replace the Boolean formula $\phi$ represented by $C$ with $\phi \land \texttt{true}$ without affecting whether a satisfying assignment exists. The second assumption is that an instance contains no AND gates, because AND gates can be simulated using a combination of NOT and OR gates. Both of these assumptions can be made without loss of generality.
	
	Fix any $q \geq 2$ and two edge weights such that $\alpha > q\beta$. Given an instance $C$ of \textsc{CircuitSAT}, we represent each Boolean variable $x$ with a heavy edge $e_x$ between a black and a white vertex. We use the duplication gadget in Figure~\ref{fig:gadgets}(c) to duplicate the edge $e_x$ as many times as the literals $x$ and $\overline{x}$ appear in $C$. The other gadgets shown in Figure~\ref{fig:gadgets} show how to combine Boolean expressions together according to $C$ to construct a multigraph $G$.
	
	We show the properties given in the theorem hold for $G$. First, any time we combine the gadgets in Figure~\ref{fig:gadgets}, we do not introduce edges between vertices of the same colour, so $G$ is bipartite. On the other hand, we have $\alpha > q\beta$ by construction. Moreover, by inspection of the gadgets in Figure~\ref{fig:gadgets}, each heavy component of $G$ is an NTOM.
	
	We claim that $C$ has a satisfying assignment if and only if $G$ has an EFX$_0$ orientation. The truth value assigned to a variable corresponds to the orientation of the heavy edge representing that variable. Specifically, an assignment of \texttt{true} (resp.\ \texttt{false}) corresponds to the heavy edge being directed toward its black (resp.\ white) endpoint. For convenience, we will directly refer to a heavy edge as being {\em true} or {\em false} depending on its orientation. We proceed to show that each of the gadgets correctly corresponds to the stated Boolean functions. The OR and duplication gadgets are those used in \citep{christodoulou2023fair}, so we only consider the NOT and TRUE gadgets.
	
	Consider the NOT gadget. Fix any EFX$_0$ orientation. We show that the edge $\neg x$ is false if and only if the edge $x$ is true. By symmetry, we only need to show one direction. Suppose the edge $x$ is true, i.e.\ it is directed toward $v_1$. Then, $w_5$ envies $v_1$. Since $w_5$ does not strongly envy $v_1$, the light edge between $v_1$ and $v_2$ must be directed toward $v_2$. The submultigraph induced by the vertices $v_i$ is exactly $H$ in Figure~\ref{fig:subgadget}. Since none of $v_2, v_3, v_4$ are adjacent to a vertex in $G-H$, Lemma~\ref{lemma:subgadget} implies the light edge between $v_4$ and $v_5$ is oriented toward $v_5$. If the edge $\neg x$ is oriented toward $v_5$, then $w_1$ would strongly envy $v_5$. Hence, the edge $\neg x$ is oriented toward $w_1$, i.e.\ it is false.
	
	Moreover, we must show that regardless of the orientation of $x$, the NOT gadget has an EFX$_0$ orientation. By symmetry, assume $x$ is directed toward $v_1$ without loss of generality. It is straightforward to check that the orientation in which every edge is oriented in a clockwise manner, except for the light edge between $v_2, v_3$ and the light edge between $w_2, w_3$ which are oriented counterclockwise, is EFX$_0$.
	
	Now consider the TRUE gadget. Fix any EFX$_0$ orientation. We show that the edge between $v_8$ and $v_9$ is directed toward $v_9$, i.e.\ it is true. Suppose instead that it is directed toward $v_8$. Because $v_9$ does not strongly envy $v_8$, the light edge between $v_7$ and $v_8$ is directed toward $v_7$. Because $v_1$ does not strongly envy $v_7$, the heavy edge between $v_7$ and $v_1$ is directed toward $v_1$. Since $v_7$ does not strongly envy $v_1$, all $q+1$ light edges incident to $v_1$ are directed away from $v_1$. The submultigraph induced by $v_1, v_2, \dots, v_5$ is $H$ and none of $v_2, v_3, v_4$ are adjacent to any vertex in $G-H$. Hence, Lemma~\ref{lemma:subgadget} implies the light edge between $v_4$ and $v_5$ is directed toward $v_5$. Since the light edges between $v_1$ and $v_6$ are directed toward $v_6$, regardless how the heavy edge between $v_5$ and $v_6$ is directed, one of $v_5, v_6$ strongly envies the other, a contradiction. Thus, the edge between $v_8$ and $v_9$ is directed toward $v_9$.
	
	It remains to exhibit an EFX$_0$ orientation of the TRUE gadget. Orient all of the edges on the path between $v_1$ and $v_9$ rightward in the direction of $v_9$.  The remaining edges all belong to a cycle of length 6. Orient all of them counterclockwise along the cycle, except for the light edge between $v_2$ and $v_3$, which is oriented in a clockwise manner. It is straightforward to verify that this is an EFX$_0$ orientation.
\end{proof}

\section{EFX$_0$ Orientations}\label{section:EFX-orientations}

In this section, we prove Theorem~\ref{thm:contribution-2b}, which together with Observation~\ref{obs:one-tree-example} imply that an NTOM is a problematic structure preventing an EFX$_0$ orientation from existing --- as long as a symmetric bi-valued multigraph does not contain an NTOM, then it has an EFX$_0$ orientation. If it does contain an NTOM, then it is possible that no EFX$_0$ orientations exist.

We assume multigraphs are connected. Otherwise, we can apply Theorem~\ref{thm:contribution-2b} to each component.

\subsection{Technical Lemmas}\label{subsec:technical-lemmas}

Our approach to proving Theorem~\ref{thm:contribution-2b} consists of three steps. First, we find a (possibly partial) EF orientation of a small subset of the edges of $G$ that satisfies certain properties (Lemmas~\ref{lemma:non-odd-multitree} and \ref{lemma:not-multitree}). Second, we apply a technical lemma (Lemma~\ref{lemma:extension-lemma}) to finish orienting almost all of the remaining edges while maintaining envy-freeness, while leaving a matching of light edges unoriented. Finally, we finish orienting the matching of light edges (Lemma~\ref{lemma:light_matching}).

As $G$ does not have a heavy component whose heavy edges induce an NTOM, each heavy component $K$ is one of 3 types:
\begin{enumerate}
	\item $K$ is non-trivial and the heavy edges of $K$ induce a non-odd multitree;
	\item $K$ is non-trivial and the heavy edges of $K$ do not induce a multitree;
	\item $K$ is trivial.
\end{enumerate}

First, we handle the simple case in which every heavy component of $G$ is of type 3 (i.e.\ trivial). Observe that any vertex with a self-loop must be unenvied in any EFX$_0$ orientation. So, increasing the weight of a self-loop in an EFX$_0$ orientation preserves EFX$_0$. These observations together allow us to assume without loss of generality that all self-loops of $G$ have weight $\beta$. Thus, all edges (including self-loops) of $G$ have weight $\beta \geq 0$.

In the case that $\beta = 0$, any orientation is EFX$_0$. Suppose $\beta > 0$. Because EFX$_0$ is scale-invariant (i.e. scaling any agent's utility function by a multiplicative factor does not affect whether an EFX$_0$ allocation or orientation exists), we may assume $\beta = 1$ without loss of generality. So, for each vertex $i$ and each edge $e$, we have $u_i(e) \in \{0,1\}$. It is known that for such instances (called {\em binary instances}), a MNW allocation is EFX$_0$+PO (see Corollary~\ref{cor:efx-additive-2-valued}), and can be found in polynomial time \citep{amanatidis2021maximum}. PO allocations of a symmetric multigraph whose edges have positive weight are orientations\footnote{For a more thorough discussion on the relationship between PO and orientations, see Chapter~\ref{chapter:po-and-orientations}.}, so this can be stated as:

\begin{theorem}\citep{darmann2015maximizing,barman2018greedy}\label{thm:only-trivial}
	Let $G$ be a bi-valued symmetric multigraph. If every heavy component of $G$ is trivial, then $G$ has an EFX$_0$ orientation that can be found in polynomial time. \qed
\end{theorem}

The remainder of this section handles the case in which $G$ contains a heavy component of type 1 or 2. The following two lemmas show the first step in orienting such components. A {\em (partial) orientation} of a heavy component $K$ of a multigraph $G$ is a (partial) orientation of $G$ that only orients edges of $K$.

\begin{lemma}\label{lemma:non-odd-multitree}
	If $K$ is a type 1 heavy component, then there exists a (possibly partial) EF orientation $\pi^K$ of $K$ and two vertices $v \neq w$ in $K$ such that
	\begin{enumerate}
		\item for each vertex $i \in K$, we have $u_i(\pi^K_i) \geq \alpha$;
		\item all edges between $v$ and $w$ are oriented except possibly one light edge;
		\item $\pi^K$ is PEF between $v$ and $w$; and
		\item at most one edge is oriented between each pair of vertices, except for the pair $v, w$.
	\end{enumerate}
	Moreover, $\pi^K$ can be found in polynomial time.
\end{lemma}
\begin{proof}
	Since $K$ is type 1, it is non-trivial and the heavy edges of $K$ induce a non-odd multitree. Let $v \neq w$ be a pair of vertices of $K$ with an even number $h \geq 2$ of heavy edges in between and $\ell$ denote the number of light edges between $v, w$. We construct $\pi^K$ in two steps. First, orient $h/2$ heavy edges and $\lfloor \ell/2 \rfloor$ light edges between $v, w$ to each of $v, w$. Then, fix a spanning tree $T$ of $K$ consisting of only heavy edges. $T$ contains exactly one heavy edge $e$ between $v, w$, which is already oriented. Orient the tree edges of $T$ (except for $e$) in the direction away from $v, w$.
	
	Clearly, $\pi^K$ satisfies (1)--(4) and can be found in polynomial time. To see that $\pi^K$ is EF, observe that neither of $v, w$ envies the other by (3), and for any other pair of vertices $i \neq j$, (1) and (4) ensure neither envies the other.
\end{proof}

\begin{lemma}\label{lemma:not-multitree}
	If $K$ is a type 2 heavy component, then there exists a (possibly partial) EF orientation $\pi^K$ of $K$ such that
	\begin{enumerate}
		\item for each vertex $i \in K$, we have $u_i(\pi^K_i) \geq \alpha$; and
		\item at most one edge is oriented between each pair of vertices.
	\end{enumerate}
	Moreover, $\pi^K$ can be found in polynomial time.
\end{lemma}
\begin{proof}
	Since $K$ is type 2, it is non-trivial and the heavy edges of $K$ do not induce a multitree. Let $T$ be a spanning tree of $K$ that consists of only heavy edges. Since the heavy edges of $K$ do not induce a multitree, there is a heavy self-loop at some vertex $v$ or there are two vertices $v \neq w$ that are joined by a heavy edge and not joined by a tree edge of $T$. Orient all tree edges of $T$ in the direction away from $v$. Then, orient either a self-loop at $v$ toward $v$ in the first case, or exactly one heavy edge between $v, w$ toward $v$ in the second case.
	
	In either case, each vertex receives exactly one incident heavy edge, so (1) holds. The construction implies (2). Moreover, (1)--(2) imply $\pi^K$ is EF. On the other hand, it is clear that $\pi^K$ can be constructed in polynomial time.
\end{proof}

Before proceeding, we introduce a technical result.

\begin{theorem}\citep{plaut2020almost}\label{thm:2vertices}
	For two vertices $i \neq j$ with additive utility functions, there exists an EFX$_0$ allocation in which $i$ does not envy $j$. Moreover, such an allocation can be found in polynomial time.\qed
\end{theorem}

Their idea is to compute an EFX$_0$ allocation $\pi$ from the perspective of $j$, that is, an allocation of the items into two bundles such that $j$ would not envy $i$ regardless of which bundle $j$ receives. Letting $i$ pick a bundle of $\pi$ before $j$ yields an EFX$_0$ allocation in which $i$ does not envy $j$.

We can now show Lemma~\ref{lemma:extension-lemma}, which allows us to extend the partial orientations resulting from Lemmas~\ref{lemma:non-odd-multitree} and \ref{lemma:not-multitree} while preserving envy-freeness. 

\begin{lemma}\label{lemma:extension-lemma}
	Let $G$ be a bi-valued symmetric multigraph of multiplicity $q \geq 1$ and $\pi' = (\pi'_1, \pi'_2, \dots, \pi'_n)$ be a (possibly partial) EF orientation of $G$. If $\pi'$ orients at most one edge between a pair of vertices $i \neq j$ and
	\begin{enumerate}
		\item $u_i(\pi'_i) \geq \alpha$ and $u_j(\pi'_j) \geq \alpha$; or
		\item $u_i(\pi'_i) \geq \beta$ and $u_j(\pi'_j) \geq \beta$ and there are only light edges between $i, j$.
	\end{enumerate}
	or if $\pi'$ orients no edges between $i \neq j$ and
	\begin{enumerate}
		\item[3.] $u_i(\pi'_i) \geq \alpha$; or
		\item[4.] $u_i(\pi'_i) \geq \beta$ and there are only light edges between $i, j$.
	\end{enumerate}
	then there exists an EF extension $\pi$ of $\pi'$ such that all edges between $i, j$ and self-loops at $i, j$ are oriented. Moreover, $\pi$ can be found in polynomial time.
\end{lemma}
\begin{proof}
	If $\pi'$ orients an edge between $i$ and $j$, then we assume without loss of generality that it orients it toward $i$. We construct the extension $\pi$ of $\pi'$ in two steps. First, we finish orienting all the edges between $i, j$ in a way that ensures neither $i$ nor $j$ envies the other. Second, we orient all self-loops at $i, j$ toward $i, j$, respectively. Note that this construction does not depend on whether $\pi'$ orients an edge between $i$ and $j$.
	
	Set $\pi = \pi'$ and let $U$ denote the set of edges between $i, j$ that are not oriented by $\pi'$. Theorem~\ref{thm:2vertices} implies the existence of an EFX$_0$ allocation $X = (X_i, X_j)$ of $U$ to vertices $i, j$ in which $j$ does not envy $i$. Note that $X$ is an orientation because it allocates edges between $i, j$ to $i, j$. Let $\pi$ orient the edges in $U$ according to $X$, and $\pi_i, \pi_j$ denote the resulting bundles of vertices $i, j$, respectively.
	
	We claim that $i$ and $j$ do not envy each other. By the construction of $X$, we have $u_j(X_j) \geq u_j(X_i)$. If $\pi'$ does not orient an edge between $i$ and $j$, then $u_j(\pi_j) \geq u_j(X_j) \geq u_j(X_i) = u_j(\pi_i)$, so $j$ does not envy $i$. Otherwise, $\pi$' orients some edge $e$ between $i, j$ toward $i$ and one of (1)--(2) holds. If (1) holds, then $u_j(\pi'_j) \geq \alpha \geq u_j(e)$. If (2) holds, then $u_j(\pi'_j) \geq \beta \geq u_j(e)$. In either case, $u_j(\pi_j') \geq u_j(e)$, so
	\begin{align*}
		u_j(\pi_j) &= u_j(\pi'_j) + u_j(X_j) \\
		&\geq u_j(\pi'_j) + u_j(X_i) \\
		&\geq u_j(e) + u_j(X_i) = u_j(\pi_i)
	\end{align*}
	On the other hand, if $X_j = \emptyset$, then $\pi$ orients no edges between $i, j$ toward $j$, so $i$ does not envy $j$. Assume $X_j$ is nonempty and fix any edge $e \in X_j$. If (1) or (3) holds, then $u_i(\pi'_i) \geq \alpha \geq u_i(e)$. If (2) or (4) hold, then there are only light edges between $i, j$, so $u_i(\pi'_i) \geq \beta \geq u_i(e)$. In either case, we have $u_i(\pi'_i) \geq u_i(e)$, so
	\begin{align*}
		u_i(\pi_i) &= u_i(\pi'_i) + u_i(X_i) \\
		&\geq u_i(\pi'_i) + u_i(X_j \setminus \{e\}) \\
		&= u_i(\pi'_i) + u_i(X_j) - u_i(e) \\
		&\geq u_i(X_j) = u_i(\pi_j)
	\end{align*}
	Thus, $i$ and $j$ do not envy each other. It follows that $\pi$ can additionally orient any self-loops at $i, j$ toward $i, j$, respectively, without causing envy.
	
	Our construction relies on Theorem~\ref{thm:2vertices} to find an EFX$_0$ allocation between two vertices in polynomial time, so $\pi$ can be found in polynomial time.
\end{proof}

Using the above, we can construct an EF orientation $\pi^G$ of $G$ in which all heavy edges are oriented and the set of unoriented edges is a matching in $G$. Recall that a {\em matching} in $G$ is a set of edges $M$ such that each vertex of $G$ appears as an endpoint of an edge of $M$ at most once.

\begin{lemma}\label{lemma:all_but_a_matching}
	Let $G$ be a bi-valued symmetric multigraph containing a non-trivial heavy component. If $G$ does not contain a heavy component whose heavy edges induce an NTOM, then there is a (possibly partial) EF orientation $\pi^G$ of $G$ such that
	\begin{enumerate}
		\item all heavy edges are oriented;
		\item the set of unoriented edges is a matching $M$ in $G$; and
		\item $\pi^G$ is PEF between all pairs of vertices with an unoriented edge in between.
	\end{enumerate}
	Moreover, $\pi^G$ can be found in polynomial time.
\end{lemma}
\begin{proof}
	At least one heavy component of $G$ is of type 1 or 2. For each type 1 heavy component $K$, Lemma~\ref{lemma:non-odd-multitree} implies there exists a (possibly partial) EF orientation $\pi^K$ of $K$ and two vertices $v \neq w$ in $K$ such that (1) for each vertex $i \in K$, we have $u_i(\pi^K_i) \geq \alpha$; (2) all edges between $v$ and $w$ are oriented except possibly one light edge; (3) $\pi^K$ is PEF between $v$ and $w$; and (4) at most one edge is oriented between each pair of vertices, other than the pair $v, w$.
	
	For each type 2 heavy component $K$, Lemma~\ref{lemma:not-multitree} implies there exists a (possibly partial) EF orientation $\pi^K$ of $K$ such that (1) for each vertex $i \in K$, we have $u_i(\pi^K_i) \geq \alpha$; and (2) at most one edge is oriented between each pair of vertices. 
	
	Let $\pi$ be the (possibly partial) orientation of $G$ that orients an edge $e$ if and only if some $\pi^K$ orients $e$, in which case $\pi$ orients $e$ in the same direction as $\pi^K$.
	
	
	We mark every vertex in a type 1 or type 2 heavy component as "processed" and every vertex in a trivial heavy component as "unprocessed". Observe that $u_i(\pi_i) \geq \alpha \geq \beta$ for each processed vertex $i$ and $u_j(\pi_j) = 0$ for each unprocessed vertex $j$. In the following, we "process" the unprocessed vertices while maintaining two invariant properties: (P1) $\pi$ is EF and (P2) every processed vertex receives at least $\beta$ utility.
	
	For an unprocessed vertex $j$ that has a processed neighbour $i$, let $\pi$ orient exactly one light edge between $i, j$ toward $j$, and mark $j$ as "processed". Doing so does not cause $i$ or $j$ to envy the other because (P2) holds for $i$, and maintains (P1)--(P2) because (P2) now holds for $j$. Since $G$ is connected, we can repeat this procedure until all vertices are "processed".
	
	We claim that for all pairs of vertices $i \neq j$ in $G$ except for pairs of special vertices $v \neq w$ in type 1 heavy components, one of the conditions in Lemma~\ref{lemma:extension-lemma} holds. If $i$ and $j$ belong to different heavy components, then $\pi$ orients at most one edge between them, and (P1)--(P2) ensure $u_i(\pi_i) \geq \beta$ and $u_j(\pi_j) \geq \beta$ and there are only light edges between $i, j$, which is condition (2) of Lemma~\ref{lemma:extension-lemma}. Otherwise, $i \neq j$ belong to the same heavy component (except for the special pairs $v \neq w$) so condition (1) of Lemma~\ref{lemma:extension-lemma} holds.
	
	We now apply Lemma~\ref{lemma:extension-lemma} to all pairs of vertices $i \neq j$ except the special pairs $v \neq w$ belonging to a type 1 heavy component. This yields an EF extension of $\pi^G$ such that all edges between $i, j$ and self-loops at $i, j$ are oriented. Since $\pi^G$ is EF, we can orient any remaining unoriented self-loops toward their respective endpoints without causing envy.
	
	Clearly, $\pi^G$ orients all heavy edges, and all light edges except for exactly one light edge between a special pair of vertices $v \neq w$ in each type 1 heavy component. Moreover, Lemma~\ref{lemma:non-odd-multitree}(3) guarantees $\pi^G$ to be PEF between every such pair $v \neq w$. Thus, (1)--(3) hold for $\pi^G$.
	
	We show $\pi^G$ can be found in polynomial time. Finding the heavy components of $\pi^G$ takes polynomial time using breadth-first search (BFS) \citep{moore1959shortest}. Orienting each heavy component of type 1 and 2 using Lemmas~\ref{lemma:non-odd-multitree} and \ref{lemma:not-multitree} takes polynomial time. Processing the vertices takes polynomial time using BFS. Finally, applying Lemma~\ref{lemma:extension-lemma} at most once to each pair of vertices takes polynomial time.
\end{proof}

We now orient the light edges left using Lemma~\ref{lemma:all_but_a_matching}.

\begin{lemma}\label{lemma:light_matching}
	Let $G$ be a bi-valued symmetric multigraph $G$ containing a non-trivial heavy component. If $\pi^G$ is a partial EF orientation satisfying (1)--(3) of Lemma~\ref{lemma:all_but_a_matching}, then there exists an extension $\pi$ of $\pi^G$ that is an EFX$_0$ orientation of $G$ that can be found in polynomial time.
\end{lemma}
\begin{proof}
	Let $M$ be the set of unoriented light edges of $\pi^G$ and $v_i, w_i$ denote the endpoints of each $e_i \in M$. By possibly renaming $v_i$ and $w_i$, we can assume for each $i$, if $\pi^G$ orients an edge not between $v_i, w_i$ toward one of $v_i, w_i$, it orients one such edge toward $v_i$. If $\pi^G$ orients such edges to both $v_i$ and $w_i$, then our assumption holds without renaming $v_i$ and $w_i$.
	
	We claim that the extension $\pi$ of $\pi^G$ that orients each edge $e_i$ toward $w_i$ is an EFX$_0$ orientation of $G$. Since $\pi^G$ is EF, any envy experienced by a vertex in $\pi$ is caused by an edge that $\pi$ orients in addition to $\pi^G$. Hence, any envy in $\pi$ is between $v_i$ and $w_i$ for some $i$. Fix such a pair $v_i, w_i$. The orientation $\pi$ orients $e_i$ toward $w_i$, so $w_i$ does not envy $v_i$. It remains to show that $v_i$ does not strongly envy $w_i$.
	
	Suppose first $\pi^G$ does not orient an edge not between $v_i, w_i$ toward one of $v_i, w_i$. Since $v_i$ does not envy $w_i$ in $\pi^G$ and $\pi$ only orients a single additional light edge between $v_i$ and $w_i$ to $w_i$, the envy that $v_i$ experiences toward $w_i$ can be alleviated by ignoring any of the light edges between $v_i$ and $w_i$. Because all edges oriented to $w_i$ by $\pi$ are between $v_i, w_i$, the vertex $v_i$ does not strongly envy $w_i$.
	
	Otherwise, $\pi^G$ orients an edge not between $v_i, w_i$ toward one of $v_i, w_i$. By our assumption, $\pi^G$ orients one such edge $e$ toward $v_i$. Since $\pi^G$ is PEF between $v_i$ and $w_i$, the vertices $v_i$ and $w_i$ receive an equal number of heavy (resp.\ light) edges that are between $v_i$ and $w_i$ in $\pi^G$. On the other hand, because $\pi^G$ orients $e$ toward $v_i$, we have $u_{v_i}(\pi^G_{v_i}) \geq u_{v_i}(\pi^G_{w_i}) + u_{v_i}(e)$. Thus, after $\pi$ orients the additional light edge $e_i$ between $v_i$ and $w_i$ toward $w_i$, we have
	\begin{align*}
		u_{v_i}(\pi_{v_i}) &= u_{v_i}(\pi^G_{v_i}) \\
		&\geq u_{v_i}(\pi^G_{w_i}) + u_{v_i}(e) \\
		&= u_{v_i}(\pi_{w_i}) - u_{v_i}(e_i) + u_{v_i}(e) \\
		&\geq u_{v_i}(\pi_{w_i})
	\end{align*}
	where the last inequality follows is because $u_{v_i}(e) \geq u_{v_i}(e_i)$ as $e_i$ is a light edge. So, $v_i$ does not envy $w_i$.
	
	Clearly, $\pi$ can be found in polynomial time given $\pi^G$.
\end{proof}

\subsection{Proof of Theorem~\ref{thm:contribution-2b}}

\restateContributionTwoB*
\begin{proof}
	If every heavy component of $G$ is trivial, then $G$ has an EFX$_0$ orientation that can be found in polynomial time by Theorem~\ref{thm:only-trivial}. Otherwise, $G$ contains a non-trivial heavy component. By Lemma~\ref{lemma:all_but_a_matching}, there exists a (possibly partial) EF orientation $\pi^G$ of $G$ such that all heavy edges are oriented, the set of unoriented edges is a matching $M$ in $G$, and $\pi^G$ is PEF between all pairs of vertices with an unoriented edge in between. By Lemma~\ref{lemma:light_matching}, there exists an extension $\pi$ of $\pi^G$ that is an EFX$_0$ orientation of $G$ that can be found in polynomial time.
\end{proof}

\section{Discussion}\label{section:discuss}

In this chapter, we showed that deciding whether a bi-valued symmetric multigraph $G$ of multiplicity $q$ has an EFX$_0$ orientation is NP-complete, even in a highly restrictive setting. We also showed that as long as $G$ does not contain a heavy component whose heavy edges induce an NTOM, then $G$ has an EFX$_0$ orientation that can be found in polynomial time. It is important to note that our results do not preclude such multigraphs from having EFX {\em allocations}. Indeed, \citet{kaviani2024almost} showed that all symmetric multigraphs have EFX allocations.

One might ask whether there are additional settings in which deciding whether EFX$_0$ orientations exist is possible in polynomial time. Theorem~\ref{thm:contribution-2a} suggests that it may be fruitful to consider settings in which $\alpha$ and $\beta$ are not so far apart.

\begin{problem}
	Is it possible to decide in polynomial time whether bi-valued symmetric multigraphs for which $\alpha \leq q\beta$ have EFX$_0$ orientations?
\end{problem}

In a slightly different vein, even if some desired type of fair allocations or orientations cannot be shown to exist, researchers have attempted to remedy the situation via a variety of means. One approach is to restrict the type of instances under consideration. However, because the setting of bi-valued symmetric multigraphs is already rather restrictive, further restrictions may be too limited in generality.

Another possibility is to leave a small subset of goods unallocated, and donate them to an external "charity". An example of this is Theorem~\ref{thm:monotone-n-2}, which states that all monotone goods instances have a partial EFX$_0$ allocation that leaves at most $n-2$ goods unallocated. In the case that a bi-valued symmetric multigraph $G$ fails to have an EFX$_0$ orientation, it is always possible to find a partial EFX$_0$ orientation by donating heavy edges until $G$ no longer contains a heavy component whose heavy edges induce an NTOM. A trivial upper bound on the number of heavy edges that need to be donated is $n-1$ in the case that each vertex belongs to a heavy component whose edges induce an NTOM. Thus, we obtain the following.

\begin{theorem}
	Let $G$ be a bi-valued symmetric multigraph of multiplicity $q \geq 1$. Then, $G$ has a partial EFX$_0$ orientation that leaves at most $n-1$ heavy edges unallocated. \qed
\end{theorem}

The following problem immediately arises.

\begin{problem}
	Do bi-valued symmetric multigraphs have partial EFX$_0$ orientations that leave fewer than $n-1$ heavy edges unallocated?
\end{problem}

A related approach is to offer a small amount of monetary subsidy to envious agents in order to relieve their envy. The subsidy can be thought of as a divisible good. The goal in this setting is to minimize the amount of subsidy required. By possibly scaling the utility functions of the agents, assume that $u_i(o_j) \in [0, 1]$ for each agent $i$ and each good $o_j$. \citet{halpern2019fair} showed that the minimum subsidy required to ensure an EF allocation for an additive instance is at most $(n-1)m$. \citet{brustle2020one} improved this by showing that in fact, a subsidy of $n-1$ is sufficient to guarantee the existence of an EF allocation.

As for the graphical model, \citet{li2025subsidy} showed that when the agents have monotone utility functions, the minimum subsidy required to guarantee any multigraph $G$ to have an EF orientation is  between $n-2$ and $n-1$. They also considered two special cases: (1) when the agents have additive utility functions and (2) when $G$ is a simple graph. For these two cases, they successfully determined the exact minimum subsidy required to be $n/2$ and $n-2$, respectively. In particular, their findings imply that a subsidy of $n/2$ is necessary and sufficient to guarantee the existence of EF orientations of bi-valued symmetric multigraphs.

\chapter{Polynomial-Time Algorithms for Fair Orientations of Chores}\label{chapter:Fair-Orientations-Chores}

In this chapter, we present our work on EF1 and EFX$_0$ orientations of graphical chores instances. The work presented in this chapter is joint work with Valerie King.

Our first two theorems pertain to graphs that can contain self-loops. In each of them, we make minimal assumptions on the utility functions of the agents, the details of which are given in Subsection~\ref{subsec:assumptions}. Moreover, these two results hold even if the agents have additive utility functions (see discussion in Section~\ref{section:additive}).

\mainContributionThree*

\mainContributionThreeB*

Our second two theorems pertain to multigraphs.

\mainContributionThreeC*

\mainContributionThreeD*

\section{Preliminaries}\label{section:prelim}

\subsection{The Fair Orientation Problem}\label{subsec:assumptions}

Recall that a graphical instance of the fair division problem is represented by a graph $G$ in which each vertex represents an agent and each edge represents an item. In graphical instances, we assume an item $e$ has zero marginal utility to an agent $i$ if it is $e$ is not incident to $i$.

In this chapter, we are particularly interested in chores instances, in which every edge represents a chore. In addition, we make the following two minimal assumptions on each utility function $u_i$:
\begin{itemize}
	\item $u_i$ is {\em monotone}, i.e., $u_i(S) \leq u_i(T)$ whenever $T \subseteq S$.
	\item If $u_i(e) = 0$ for each $e \in S$, then $u_i(S) = 0$.
\end{itemize}

Given $(G, U)$, \textsc{EFX$_0$-Orientation} asks whether $G$ has an EFX$_0$ orientation. We are particularly interested in instances in which every edge $e = \{i, j\}$ is {\em objective}, i.e.\ $u_i(e) = 0$ if and only if $u_j(e) = 0$. Equivalently, either $e$ has zero marginal utility for both endpoints or $e$ has negative marginal utility for both. If every edge is objective, we call $(G, U)$ an {\em objective} instance of \textsc{EFX$_0$-Orientation}. We use \textsc{EFX$_0$-Orientation-Objective} to denote the \textsc{EFX$_0$-Orientation} problem restricted to objective instances.

When considering objective instances, we refer to edges with zero (resp.\ negative) marginal utility for both endpoints as {\em dummy edges} (resp.\ {\em negative edges}). Moreover, we define a {\em negative component} of $G$ to be a maximal vertex-induced subgraph $K$ such that for any two vertices $i, j$ of $K$, there is a path in $K$ between $i$ and $j$ that only contains negative edges. The negative components of $G$ form a partition of $V(G)$. See Figure~\ref{fig:negative-component} for an example.

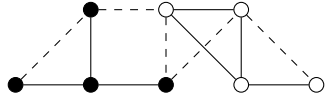
\begin{figure}
	\centering
	\begin{tikzpicture}
		\node[circlevertex] (a1) at (0,0) {};
		\node[circlevertex] (a2) at (1,0) {};
		\node[circlevertex] (a3) at (2,0) {};
		\node[circlevertex] (a4) at (1,1) {};
		\node[emptyvertex] (b1) at (2,1) {};
		\node[emptyvertex] (b2) at (3,1) {};
		\node[emptyvertex] (b3) at (3,0) {};
		\node[emptyvertex] (b4) at (4,0) {};
		\draw (a1) to (a2) (a2) to (a3) (a2) to (a4)
		(b1) to (b2) (b2) to (b3) (b3) to (b1) (b3) to (b4);
		
		\draw[dashed] (a1) to (a4) (a3) to (b2) (b2) to (b4) (a4) to (b1) (b1) to (a3);
	\end{tikzpicture}
	\caption{The graph of an objective instance. Solid (resp.\ dashed) edges denote negative (resp.\ dummy) edges. The set of black (resp.\ white) vertices induces one negative component.}
	\label{fig:negative-component}
\end{figure}

\subsection{Related Decision Problems}

Let $H$ be a graph, $P = \{P_1, P_2, \dots, P_k\}$ be a set of pairwise disjoint subsets of $V(H)$, and $D \subseteq V(H)$ be a subset of vertices. The set $P$ is not necessarily a partition of $V(H)$ because not every vertex needs to be in $P$. We say a subset $C \subseteq V(H)$ is a {\em $(P, D)$-vertex cover} of $H$ if $C$ is a vertex cover of $H$ that contains at most one vertex from each $P_\ell \in P$ and no vertex in $D$, i.e.\ the following three conditions hold:
\begin{itemize}
	\item Each edge of $H$ is incident to at least one vertex in $C$;
	\item $|C \cap P_\ell| \leq 1$ for each $P_\ell \in P$; and
	\item $C \cap D = \emptyset$.
\end{itemize}

\begin{definition}[\textsc{PD-Vertex-Cover}]
	Given a tuple $(H, P, D)$ where $H$ is a graph, $P = \{P_1, P_2, \dots, P_k\}$ is a set of pairwise disjoint subsets of $V(H)$, and $D \subseteq V(H)$, \textsc{PD-Vertex-Cover} asks whether there exists a $(P,D)$-vertex cover $C$ of $H$.
\end{definition}

\begin{definition}[\textsc{2SAT}]
	The \textsc{2SAT} problem is a special case of \textsc{SAT} in which each clause consists of one or two literals.
\end{definition}

\begin{definition}[\textsc{Partition} \citep{karp1972reducibility}]
	Given a set $S = \{s_1, s_2, \dots, s_k\}$ of natural numbers, \textsc{Partition} asks whether there exists a partition $S = S_a \cup S_b$ such that $\sum_{x \in S_a} x = \sum_{x \in S_b} x$. (Such a partition is called an {\em equipartition}.)
\end{definition}

\section{EF1 Orientations of Graphs}\label{section:EF1}

We give our $O(|V(G)| + |E(G)|)$-time algorithm for finding an EF1 orientation of a graph $G$ if one exists. Our algorithm follows immediately from the following two propositions. The second follows from the first and is stated without proof.

\begin{proposition}\label{prop:EF1-graphs}
	An orientation $\pi$ of a graph $G$ is EF1 if and only if each vertex $i$ receives at most one edge of negative utility to it.
\end{proposition}
\begin{proof}
	Suppose $i$ receives two edges $e = \{i, j\}$ and $e' \neq e$ of negative utility to it. Assume $e'$ is a worst chore in $\pi_i$. Since $e \in \pi_i \setminus \{e'\}$, we have $u_i(\pi_i \setminus \{e'\}) \leq u_i(e) < 0 = u_i(\pi_{j})$, so $i$ envies $j$ even after ignoring $e'$ and $\pi$ is not EF1. Otherwise, $i$ receives at most one edge of negative utility to it. Ignoring such an edge alleviates any envy $i$ has, so $\pi$ is EF1.
\end{proof}

\begin{proposition}\label{prop:EF1-augment}
	If $\pi$ is an EF1 orientation of a graph $G$, then introducing a new edge $e = \{i, j\}$ such that $u_i(e) = 0$ and orienting it toward $i$ results in another EF1 orientation. \qed
\end{proposition}

It is straightforward to decide whether $G$ has an EF1 orientation and compute one if it exists using these two propositions.

\mainContributionThreeB*
\begin{proof}
	First, partition the edge set $E(G)$ of $G$ into $E_0 \cup E_1$, where $E_0$ is the set of edges that have zero utility to at least one endpoint, and $E_1$ is the set of edges that have strictly negative utility to both endpoints. Proposition~\ref{prop:EF1-augment} allows us to ignore the edges in $E_0$ without loss of generality. Hence, every edge has strictly negative utility to both endpoints. In this case,  Proposition~\ref{prop:EF1-graphs} implies that an orientation $\pi$ of $G$ is EF1 if and only if each vertex $i$ receives at most one edge. Clearly, this is the case if and only if each component $K$ of $G$ contains at most $|V(K)|$ edges.
	
	To construct an EF1 orientation in the case that it exists, orient each edge $e \in E_0$ toward an endpoint $i$ such that $u_i(e) = 0$. Then, using BFS, find the components of the subgraph of $G$ induced by the edges of $E_1$, and orient the edges of $E_1$ in a way such that each vertex of $G$ receives at most one.
	
	Ignoring the set $E_0$ of edges that have zero utility to at least one endpoint takes $O(|E(G)|)$ time by checking every edge. Verifying the number of edges in each component of $G$ takes $O(|V(G)| + |E(G)|)$ time using BFS. Constructing the EF1 orientation also takes $O(|V(G)| + |E(G)|)$ time because it can be done using BFS.
\end{proof}

\section{EFX$_0$ Orientations of Graphs}\label{section:graphs}

In this section, we give our $O(|V(G)|^4)$-time algorithm for \textsc{EFX$_0$-Orientation}. Our algorithm is essentially a reduction from \textsc{EFX$_0$-Orientation} to \textsc{2SAT}. It comprises three nested routines, each functioning as a reduction from one decision problem to another. The main algorithm \textsc{FindEFXOrientation} (Algorithm~\ref{alg:FindEFXOrientation}) corresponds to a reduction from \textsc{EFX$_0$-Orientation} to \textsc{EFX$_0$-Orientation-Objective}. \textsc{FindEFXOrientation} calls \textsc{FindEFXOrientObj} as a subroutine, which corresponds to a reduction from \textsc{EFX$_0$-Orientation-Objective} to \textsc{PD-Vertex-Cover}. Finally, \textsc{FindEFXOrientObj} calls \textsc{FindPDVertexCover} as a subroutine, which corresponds to a reduction from \textsc{PD-Vertex-Cover} to \textsc{2SAT}. See Figure~\ref{fig:schematic} for a visual overview.

We present the proofs of correctness of the algorithms in their nested order, starting with innermost one, \textsc{FindPDVertexCover}.

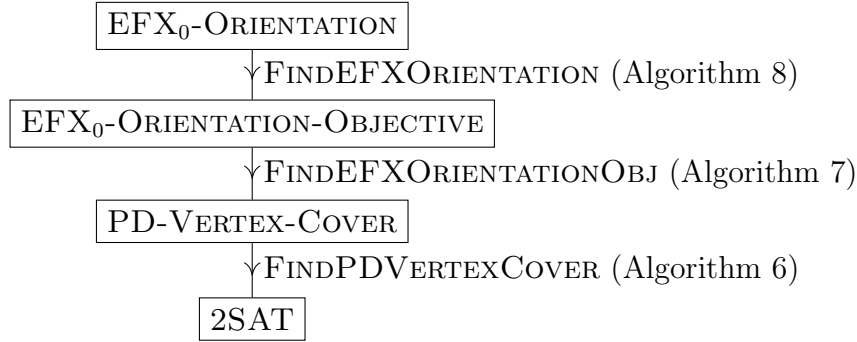
\begin{figure}
	\centering
	\begin{tikzpicture}
		\node[sqvertex] (P1) at (0,3*1.3) {\textsc{EFX$_0$-Orientation}};
		\node[sqvertex] (P2) at (0,2*1.3) {\textsc{EFX$_0$-Orientation-Objective}};
		\node[sqvertex] (P3) at (0,1*1.3) {\textsc{PD-Vertex-Cover}};
		\node[sqvertex] (P4) at (0,0) {\textsc{2SAT}};
		
		\draw[arc] (P1) to node[right] {\textsc{FindEFXOrientation} (Algorithm~\ref{alg:FindEFXOrientation})} (P2);
		
		\draw[arc] (P2) to node[right] {\textsc{FindEFXOrientationObj} (Algorithm~\ref{alg:FindEFXOrientObj})} (P3);
		
		\draw[arc] (P3) to node[right] {\textsc{FindPDVertexCover} (Algorithm~\ref{alg:FindPDVertexCover})} (P4);
	\end{tikzpicture}
	\caption{The reductions our algorithms represent. Each node on the left is a decision problem. A directed edge between nodes is a reduction and is labelled with the relevant algorithm. \vspace{\baselineskip}}
	\label{fig:schematic}
\end{figure}

\subsection{\textsc{FindPDVertexCover} (Algorithm~\ref{alg:FindPDVertexCover})}

\textsc{FindPDVertexCover} (Algorithm~\ref{alg:FindPDVertexCover}) accepts an instance $(H, P, D)$ of \textsc{PD-Vertex-Cover} as input, and outputs a $(P, D)$-vertex cover of $H$ if one exists and \texttt{false} otherwise.

Recall that in an instance $(H, P, D)$ of \textsc{PD-Vertex-Cover}, $H$ is a graph, $P = \{P_1, P_2, \dots, P_k\}$ is a set of pairwise disjoint subsets of $V(H)$, and $D \subseteq V(H)$. The algorithm starts by constructing an instance $\phi$ of \textsc{2SAT} on Line~\ref{line:1.create-sat}, which contains a Boolean variable $x_i$ for each vertex $i$ of $H$ and three types of clauses:
\begin{itemize}
	\item (type 1) the clause $\{x_i, x_j\}$ for each edge $e = \{i, j\}$ of $H$ (possibly $i=j$),
	\item (type 2) the clause $\{\overline{x_i}, \overline{x_j}\}$ for each set $P_\ell \in P$ and each pair of distinct vertices $i, j \in P_\ell$, and
	\item (type 3) the clause $\{\overline{x_i}\}$ for each vertex $i \in D$.
\end{itemize}
Then, the algorithm determines whether $\phi$ has a satisfying assignment $f$ on Line~\ref{line:1.call-sat}. If so, the algorithm outputs the $(P,D)$-vertex cover $C \coloneqq \{i \in V(H) \mid f(x_i) = \texttt{true}\}$. Otherwise, the algorithm outputs \texttt{false}.

\begin{algorithm}
	\caption{$\textsc{FindPDVertexCover}(H,P,D)$}
	\label{alg:FindPDVertexCover}
	
	\hspace*{\algorithmicindent} \textbf{Input:} An instance $(H, P, D)$ of \textsc{PD-Vertex-Cover}.
	
	\hspace*{\algorithmicindent} \textbf{Output:} A $(P,D)$-vertex cover $C$ of $H$ if one exists and \texttt{false} otherwise.
	\begin{algorithmic}[1] 
		
		
		
		
		
		
		
		\State $\phi \gets$ the 2SAT instance defined above \label{line:1.create-sat}
		\State $f \gets \textsc{2SAT}(\phi)$\label{line:1.call-sat}
		
		\If{$f = \texttt{false}$} \Comment{$\phi$ has no satisfying assignment}
		\State \Return \texttt{false}\label{line:1.return-false}
		\EndIf

		\State \Return $C \coloneqq \{i \in V(H) \mid f(x_i) = \texttt{true}\}$\label{line:1.return-C}
	\end{algorithmic}
\end{algorithm}

\begin{lemma}\label{lemma:reduc:FindPDVertexCover}
	The following statements hold for \textsc{FindPDVertexCover} (Algorithm~\ref{alg:FindPDVertexCover}):
	\begin{enumerate}
		\item If $H$ has a $(P,D)$-vertex cover, then $\phi$ has a satisfying assignment.
		
		\item If $\phi$ has a satisfying assignment $f$, then $C \coloneqq \{i \in H \mid f(x_i) = \normalfont{\texttt{true}}\}$ is a $(P,D)$-vertex cover of $H$. Otherwise, \textsc{FindPDVertexCover} outputs \texttt{false}.
		
		\item \textsc{FindPDVertexCover} runs in $O(|V(H)|^2)$ time.
	\end{enumerate}
\end{lemma}
\begin{proof}
	(1): Suppose $H$ has a $(P, D)$-vertex cover $C$. We claim the truth assignment $f$ in which $f(x_i) = \texttt{true}$ if and only if $i \in C$ to be a satisfying assignment for $\phi$. Since $C$ is a $(P, D)$-vertex cover of $H$, each edge $e$ of $H$ is incident to at least one vertex in $C$. In other words, for each edge $e = \{i, j\}$ of $H$, we have $i \in C$ or $j \in C$. Hence, at least one of $x_i, x_j$ is true under $f$, so the type 1 clause $\{x_i, x_j\}$ is satisfied.
	
	On the other hand, the definition of $(P, D)$-vertex covers ensures $|C \cap P_\ell| \leq 1$ for each $P_\ell \in P$. Hence, each $P_\ell$ contains at most one vertex $i$ such that $f(x_i) = \texttt{true}$, so all type 2 clauses are satisfied.
	
	Finally, because $C$ is a $(P, D)$-vertex cover, we have $C \cap D = \emptyset$. Hence, for each $i \in D$, we have $i \notin C$, so $f(x_i) = \texttt{false}$, ensuring all type 3 clauses are satisfied.
	
	(2): Suppose $\phi$ has a satisfying assignment $f$. We claim that $C \coloneqq \{i \in V(H) \mid f(x_i) = \texttt{true}\}$ is a $(P, D)$-vertex cover of $H$. Since the type 1 clause $\{x_i, x_j\}$ is satisfied under $f$, the edge $e = \{i, j\}$ of $H$ is incident to at least one vertex of $C$, namely, $i$ or $j$. Since the type 2 clauses are satisfied, for each fixed $P_\ell \in P$, there is at most one vertex $i \in P_\ell$ such that $x_i$ is true under $f$. Hence, $|C \cap P_\ell| \leq 1$. Finally, the type 3 clauses ensure that $f(x_i) = \texttt{false}$ for each vertex $i \in D$, so $i \notin C$. Thus, $C \cap D = \emptyset$.
	
	Otherwise, $\phi$ has no satisfying assignment and \textsc{FindPDVertexCover} outputs \texttt{false} on Line~\ref{line:1.return-false}.
	
	(3): Constructing $\phi$ on Line~\ref{line:1.create-sat} takes $O(|V(H)|^2)$ time because $\phi$ contains $|V(H)|$ variables and $O(|V(H)|^2)$ clauses. Verifying whether $\phi$ has a satisfying assignment on Line~\ref{line:1.call-sat} takes time linear in the number of variables and clauses using the algorithm due to \citet{aspvall1979linear}, which is $O(|V(H)|^2)$. Finally, constructing $C$ takes $O(|V(H)|)$ time by checking each vertex of $H$. Thus, \textsc{FindPDVertexCover} runs in $O(|V(H)|^2)$ time.
\end{proof}

\subsection{\textsc{FindEFXOrientObj} (Algorithm~\ref{alg:FindEFXOrientObj})}

\textsc{FindEFXOrientObj} (Algorithm~\ref{alg:FindEFXOrientObj}) takes an instance $(G^o, u^o)$ of \textsc{EFX$_0$-Orientation-Objective} as input and outputs an EFX$_0$ orientation of $G^o$ if one exists and \texttt{false} otherwise. Recall that in an objective instance, all edges are objective, and we distinguish dummy edges that provide zero marginal utility to both endpoints and negative edges that provide negative marginal utility to both endpoints. Recall also that a negative component of $G^o$ is a maximal vertex-induced subgraph $K$ such that for any two vertices $i, j$ of $K$, there exists a path in $K$ between $i$ and $j$ that only contains negative edges.

We make the following crucial observation.

\begin{proposition}\label{prop:one-neg-or-all-dummies}
	Let $(G^o, u^o)$ be an instance of \textsc{EFX$_0$-Orientation-Objective}. An orientation $\pi^o$ of $G^o$ is EFX$_0$ if and only for each vertex $i$, the orientation $\pi^o$ contains a unique edge directed toward $i$ or every edge directed toward $i$ is a dummy edge.
\end{proposition}
\begin{proof}
	
	Let $\pi^o = (\pi_1^o, \pi_2^o, \dots, \pi_{|V(G^o)|}^o)$ be an orientation of $G^o$ and suppose that for each vertex $i$, the orientation $\pi^o$ contains a unique edge directed toward $i$ or every edge directed toward $i$ is a dummy edge. If $\pi^o$ contains a unique edge $e$ directed toward $i$, then we have $u_i^o(\pi^o_i \setminus \{e\}) = u_i^o(\emptyset) = 0 \geq u_i^o(\pi_j^o)$ for any agent $j$. Otherwise, every edge directed toward $i$ is a dummy edge, so $u_i^o(\pi_i^o) = 0 \geq u_i^o(\pi_j^o)$ for any $j$. Hence, no agent strongly envies another and $\pi^o$ is EFX$_0$.
	
	Conversely, suppose the orientation $\pi^o$ contains two edges $e, e'$ directed toward some $i$ and not every edge directed toward $i$ is a dummy edge. Without loss of generality, assume $u_i^o(e') < 0$. Let $j, j'$ denote the other endpoints of the edges $e, e'$, respectively. Note that $j \neq j'$ because $G$ does not contain parallel edges. Also, $u_i^o(\pi_j^o) = 0$ because $i$ receives the edge $e$ between $i$ and $j$. Since $u_i^o(e') < 0$ and $\{e'\} \subseteq \pi_i^o \setminus \{e\}$, the monotonicity of $u_i^o$ implies $u_i^o(\pi_i^o \setminus \{e\}) \leq u_i^o(e') < 0 = u_i^o(\pi_j^o)$. So, $\pi^o$ is not EFX$_0$.
\end{proof}

Proposition~\ref{prop:one-neg-or-all-dummies} gives a graph theoretical condition that is equivalent to the EFX$_0$ condition, which \textsc{FindEFXOrientObj} exploits. The algorithm first finds the set $\mathcal{K}$ of negative components of $G^o$ on Line~\ref{line:2.find-neg-comps}. If some negative component $K$ contains more than $|V(K)|$ negative edges, the algorithm concludes that $G^o$ has no EFX$_0$ orientation and outputs \texttt{false}. Otherwise, every negative component $K$ has at most $|V(K)|$ negative edges, and the algorithm constructs the instance $(H, P, D)$ of \textsc{PD-Vertex-Cover} on Line~\ref{line:2.construct-HPD} defined as:
\begin{itemize}
	\item $H$ is the graph with the same vertex set as $G^o$ and the set of dummy edges of $G^o$;
	\item $P = \{K \in \mathcal{K} \mid K \text{ has exactly } |V(K)|-1 \text{ negative edges}\}$;
	\item $D$ is the set of vertices that belong to a negative component $K \in \mathcal{K}$ that has exactly $|V(K)|$ negative edges.
\end{itemize}

Using \textsc{FindPDVertexCover} (Algorithm~\ref{alg:FindPDVertexCover}) on Line~\ref{line:2.subroutine}, the algorithm tries to find a $(P,D)$-vertex cover $C$ of $H$. If successful, it uses $C$ to construct the orientation $\pi^o$ of $G^o$  defined below and outputs it. Otherwise, it outputs \texttt{false}.

We give the construction of $\pi^o$. Because $C$ is a $(P, D)$-vertex cover of $H$, each edge $e$ of $H$ (equivalently, each dummy edge of $G^o$) is incident to a vertex $i \in C$. Orient $e$ toward $i$, with arbitrary tie-breaking when two such $i$ exist.

We now orient the negative edges of $G^o$ in such a way that ensures each vertex receives at most one negative edge. This can be done as follows for each negative component $K \in \mathcal{K}$. If $K$ has exactly $|V(K)|$ negative edges, then the negative edges of $K$ have the form of a spanning tree of $K$ with an extra edge, so there is a unique cycle $O$ of negative edges of $K$. By orienting $O$ cyclically, any remaining unoriented negative edges of $K$ form trees branching from $O$, so orienting them in the direction away from $O$ ensures each vertex of $K$ receives exactly one negative edge.

Otherwise, $K$ has exactly $|V(K)|-1$ negative edges. Because $C$ is a $(P, D)$-vertex cover of $H$, it contains at most one vertex from each $K \in \mathcal{K}$, so at most one vertex of $K$ received a dummy edge in the previous step. Choose $r$ to be such a vertex if it exists, or an arbitrary vertex of $K$ otherwise. Orient the negative edges of $K$ in the direction away from $r$, so that they form a directed spanning tree of $K$ rooted at $r$, ensuring that each vertex of $K$ receives at most one negative edge.

\begin{algorithm}
	\caption{$\textsc{FindEFXOrientObj}(G^o,u^o)$}
	\label{alg:FindEFXOrientObj}
	
	\hspace*{\algorithmicindent} \textbf{Input:} An instance $(G^o, u^o)$ of \textsc{EFX$_0$-Orientation-Objective}.
	
	\hspace*{\algorithmicindent} \textbf{Output:}
	An EFX$_0$ orientation $\pi^o$ of $G^o$ if one exists and \texttt{false} otherwise.
	\begin{algorithmic}[1] 
		\State $\mathcal{K} \gets \{K \mid K \text{ is a negative component of } G^o \}$  \label{line:2.find-neg-comps}
		
		
		
		
		\If{$\exists K \in \mathcal{K}: K$ contains $>|V(K)|$ negative edges}
		\State \Return \texttt{false} \label{line:2.return-false-1}
		\EndIf
		
		\State $(H, P, D) \gets $ the instance of \textsc{PD-Vertex-Cover} defined above \label{line:2.construct-HPD}
		
		
		
		\State $C \gets \textsc{FindPDVertexCover}(H,P,D)$ \label{line:2.subroutine} \Comment{Algorithm~\ref{alg:FindPDVertexCover}}
		
		\If{$C = \texttt{false}$}
		\State \Return \texttt{false} \label{line:2.return-false-2}
		\EndIf
		
		\State \Return the orientation $\pi^o$ of $G^o$ defined above \label{line:2.construct-pio} \label{line:2.return}
	\end{algorithmic}
\end{algorithm}

\begin{lemma}\label{lemma:reduc:FindEFXOrientObj}
	The following statements hold for \textsc{FindEFXOrientObj} (Algorithm~\ref{alg:FindEFXOrientObj}):
	\begin{enumerate}
		\item If $G^o$ has an EFX$_0$ orientation, then no negative component $K$ of $G^o$ contains $>|V(K)|$ negative edges and $H$ has a $(P,D)$-vertex cover.
		\item If $C$ is a $(P,D)$-vertex cover of $H$, then $\pi^o$ is an EFX$_0$ orientation of $G^o$.
		\item \textsc{FindEFXOrientObj} runs in $O(|V(G^o)|^2)$ time.
	\end{enumerate}
\end{lemma}
\begin{proof}
	(1): Suppose $G^o$ has an EFX$_0$ orientation $\pi^o$. First, we show that no negative component $K$ of $G^o$ contains $>|V(K)|$ negative edges. Fix a negative component $K$. By Proposition~\ref{prop:one-neg-or-all-dummies}, for each vertex $i$, the orientation $\pi^o$ contains a unique edge directed toward $i$ or every edge directed toward $i$ is a dummy edge. In particular, this is true for each vertex in $K$, so $\pi^o$ contains at most $|V(K)|$ negative edges directed toward a vertex of $K$. On the other hand, every negative edge of $K$ is directed toward a vertex of $K$ because $\pi^o$ is an orientation. It follows that $K$ contains at most $|V(K)|$ negative edges.
	
	We claim that the set $C$ of vertices $i$ of $G^o$ such that $\pi^o$ contains a dummy edge directed toward $i$ is a $(P, D)$-vertex cover of $H$. First we show that each edge of $H$ is incident to at least one vertex in $C$. Let $e$ be an edge of $H$. By the definition of $H$, the edge $e$ is a dummy edge of $G^o$. Let $i$ be the endpoint toward which $e$ is directed in the orientation $\pi^o$ of $G^o$. The definition of $C$ implies $i \in C$, so $e$ is incident to at least one vertex in $C$.
	
	Next, we show that $|C \cap K_\ell| \leq 1$ for each $K_\ell \in P$. Fix any negative component $K_\ell$ of $G^o$. Since $K_\ell$ is a negative component, it contains at least $|K_\ell|-1$ negative edges. Since $\pi^o$ is EFX$_0$, Proposition~\ref{prop:one-neg-or-all-dummies} implies that for each vertex $i$, the orientation $\pi^o$ contains a unique edge directed toward $i$ or every edge directed toward $i$ is a dummy edge. Since $K_\ell$ contains at least $|K_\ell|-1$ negative edges and each vertex of $K$ can have at most one negative edge directed towards it, the pigeonhole principle implies that $K_\ell$ contains at most one vertex to which dummy edges are directed in $\pi^o$. In other words, at most one vertex of $K_\ell$ is in $C$, i.e., $|C \cap K_\ell| \leq 1$.
	
	Finally, we show that $C \cap D = \emptyset$. Suppose $i \in D$. Because $i \in D$, the negative component $K$ to which $i$ belongs contains $|V(K)|$ negative edges. Proposition~\ref{prop:one-neg-or-all-dummies} implies that for each vertex $j$, the orientation $\pi^o$ contains a unique edge directed toward $j$ or every edge directed toward $j$ is a dummy edge. In particular, each vertex of $K$ has at most one negative edge directed toward it. By the pigeonhole principle, each vertex of $K$ has exactly one negative edge of $K$ directed toward it. In particular, this is true for $i$, so Proposition~\ref{prop:one-neg-or-all-dummies} implies no dummy edge is directed toward $i$. Thus, $i \notin C$. It follows that $C \cap D = \emptyset$ and $H$ indeed has a $(P,D)$-vertex cover.
	
	(2): Suppose the algorithm finds a $(P, D)$-vertex cover $C$ of $H$ on Line~\ref{line:2.subroutine}. First, we show that $\pi^o$ indeed orients every edge of $G^o$. Recall that $H$ is the graph with the same vertex set as $G^o$ and the set of dummy edges of $G^o$. Since $C$ is a $(P, D)$-vertex cover of $H$, each edge of $H$ (thereby each dummy edge of $G^o$) is incident to at least one vertex in $C$. It follows that every dummy edge of $G^o$ is oriented by $\pi^o$. On the other hand, let $e$ be a negative edge of $G^o$ and $K$ be the negative component containing $e$. Since the algorithm successfully finds a $(P, D)$-vertex cover $C$ of $H$ on Line~\ref{line:2.subroutine}, it does not terminate on Line~\ref{line:2.return-false-1}. So, $K$ contains exactly $|V(K)|$ or $|V(K)|-1$ negative edges, and $\pi^o$ orients all of them. It follows that $\pi^o$ indeed  orients every edge of $G^o$.
	
	It remains to show that $\pi^o$ is EFX$_0$. We show that the condition in Proposition~\ref{prop:one-neg-or-all-dummies} holds for every vertex $i$. Let $C' \subseteq C$ be the set of vertices with a dummy edge directed toward it in $\pi^o$ and consider any $i \in C'$. Since $C$ is a $(P, D)$-vertex cover of $H$, we have $C \cap D = \emptyset$, so $i \notin D$. Because $D$ is the set of vertices that belong to a negative component $K$ of $G^o$ that has exactly $|V(K)|$ negative edges, $i$ belongs to a negative component $K$ of $G^o$ containing exactly $|V(K)|-1$ negative edges. As discussed while defining $\pi^o$, at most one vertex of such a negative component $K$ could have received a dummy edge, which clearly is the vertex $i$. So, $\pi^o$ orients the negative edges of $K$ in the direction away from $i$, so that they form a spanning tree of $K$ rooted at $i$. Hence, every edge directed toward $i$ is a dummy edge. it follows that the condition in Proposition~\ref{prop:one-neg-or-all-dummies} holds for every $i \in C'$.
	
	On the other hand, the definition of $\pi^o$ ensures each vertex receives at most one negative edge. In particular, each vertex $i \notin C'$ receives no dummy edges and at most one negative edge, so the condition in Proposition~\ref{prop:one-neg-or-all-dummies} holds for every $i \notin C'$ as well. By Proposition~\ref{prop:one-neg-or-all-dummies}, $\pi^o$ is EFX$_0$.
	
	(3): Finding the set $\mathcal{K}$ of negative components of $G^o$ on Line~\ref{line:2.find-neg-comps} can be done using a BFS using $O(|V(G^o)| + |E(G^o)|)$ time. Constructing $(H,P,D)$ on Lines~\ref{line:2.construct-HPD} requires counting the number of negative edges in each negative component $K \in \mathcal{K}$, which can be done in $O(|E(G^o)|)$ time. By Lemma~\ref{lemma:reduc:FindPDVertexCover}(3), \textsc{FindPDVertexCover} on Line~\ref{line:2.subroutine} runs in $O(|V(H)|^2)$ time.
	
	Now consider the construction of $\pi^o$ on Line~\ref{line:2.construct-pio}. Orienting the dummy edges of $G^o$ can be done by iterating through the vertices of $C \subseteq V(H)$ and checking their incident dummy edges, taking $O(|V(H)|^2)$ time.
	
	We show how to orient the negative edges of $G^o$ using at most two BFSs for each negative component $K \in \mathcal{K}$. If $K$ has $|V(K)|$ negative edges, then the negative edges of $K$ take the form of a spanning tree of $K$ together with an extra edge, i.e., they form a unique cycle $O$ and the set of negative edges not on $O$ induce a forest. In order to orient them in a way so that each vertex of $K$ receives exactly one negative edge, we first use BFS to find the cycle $O$ and orient it cyclically. Then, by identifying the vertices of $O$ together as a single {\em supervertex} and doing a BFS starting at it, we can orient the remaining negative edges of $K$ away from the cycle $O$. This ensures each vertex of $K$ on $O$ receives exactly one negative edge of $K$, and each vertex of $K$ not on $O$ also receives exactly one negative edge. Together, this takes two BFSs on $K$.
	
	Otherwise $K$ has only $|V(K)|-1$ negative edges. Verifying if $K$ contains a vertex $r$ that has received a dummy edge in the previous step takes $O(|V(K)|)$ time. Orienting the the negative edges of $K$ as a directed rooted tree can be done using a single BFS on $K$, proving our claim that orienting the negative edges of $G^o$ can be done using at most two BFSs for each negative component $K \in \mathcal{K}$.
	
	Thus, orienting the negative edges of $G^o$ takes $O(|V(G^o)| + |E(G^o)|)$ time. By summing the above and observing that $|V(H)| = |V(G^o)|$, we conclude that \textsc{FindEFXOrientObj} runs in $O(|V(G^o)|^2)$ time.
\end{proof}

\begin{lemma}\label{lemma:FindEFXOrientObj}
	Given an instance $(G^o, u^o)$ of \textsc{EFX$_0$-Orientation-Objective}, \textsc{FindEFXOrientObj} outputs an EFX$_0$ orientation $\pi^o$ of $G^o$ if one exists and \normalfont{\texttt{false}} otherwise in $O(|V(G^o)|^2)$ time.
\end{lemma}
\begin{proof}
	If $G^o$ has an EFX$_0$ orientation, then no negative component $K$ of $G^o$ contains $>|V(K)|$ negative edges and $H$ has a $(P,D)$-vertex cover by Lemma~\ref{lemma:reduc:FindEFXOrientObj}(1). So, Line~\ref{line:2.subroutine} finds a $(P, D)$-vertex cover $C$ of $H$ by Lemma~\ref{lemma:reduc:FindPDVertexCover} and \textsc{FindEFXOrientObj} outputs the EFX$_0$ orientation $\pi^o$ by Lemma~\ref{lemma:reduc:FindEFXOrientObj}(2).
	
	Otherwise, $G^o$ does not have an EFX$_0$ orientation. If some negative component $K$ of $G^o$ contains $>|V(K)|$ negative edges, then \textsc{FindEFXOrientObj} outputs \texttt{false} on Line~\ref{line:2.return-false-1}, Assume every negative component $K$ of $G^o$ contains at most $|V(K)|$ negative edges, so \textsc{FindEFXOrientObj} constructs the instance $(H, P, D)$ of \textsc{PD-Vertex-Cover} on Line~\ref{line:2.construct-HPD} and attempts to find a $(P, D)$-vertex cover $C$ of $H$ on Line~\ref{line:2.subroutine}. By the contrapositive of Lemma~\ref{lemma:reduc:FindEFXOrientObj}(2), such a $C$ does not exist, so \textsc{FindEFXOrientObj} outputs \texttt{false} on Line~\ref{line:2.return-false-2}.
\end{proof}

\subsection{\textsc{FindEFXOrientation} (Algorithm~\ref{alg:FindEFXOrientation})}

\textsc{FindEFXOrientation} (Algorithm~\ref{alg:FindEFXOrientation}) accepts an instance $(G,U)$ of \textsc{EFX$_0$-Orientation} as input and outputs an EFX$_0$ orientation of $G$ if one exists and \texttt{false} otherwise.

It first constructs the instance $(G^o, u^o)$ of \textsc{EFX$_0$-Orientation-Objective} by subdividing each non-objective edge of $G$. Specifically, for each non-objective edge $e_{ij} = \{i, j\}$, first assume $u_i(e_{ij}) = 0$ without loss of generality by possibly exchanging the labels of $i$ and $j$, then replace $e_{ij}$ by a new vertex $k$ and two new edges $e_{ik} = \{i, k\}, e_{jk} = \{j, k\}$. The utilities of the new edges are defined as shown in Figure~\ref{fig:subdivide}, in which $\beta \coloneqq u_j(e_{ij})$. The objective edges and their utilities remain unchanged.

It then uses \textsc{FindEFXOrientObj} (Algorithm~\ref{alg:FindEFXOrientObj}) on Line~\ref{line:3.subroutine} to try to find an EFX$_0$ orientation of $G^o$. If none exists, it outputs \texttt{false} on Line~\ref{line:3.return-false}. Otherwise, an EFX$_0$ orientation $\pi^o$ of $G^o$ is successfully found. Using $\pi^o$, the algorithm constructs and outputs the orientation $\pi$ of $G$ defined as follows.

Let $e_{ij}$ be an edge of $G$. If $e_{ij}$ was not subdivided in the construction of $G^o$, then it is also an edge in $G^o$. In this case, $\pi$ orients $e_{ij}$ in $G$ in the same way as $\pi^o$ orients $e_{ij}$ in $G^o$.

Otherwise, $e_{ij}$ was subdivided in the construction of $G^o$, and corresponds to two edges $e_{ik}, e_{jk}$ of $G^o$, where $u_i^o(e_{ik}) = u_i(e_{ij}) = 0$. In this case, $\pi$ orients $e_{ij}$ in $G$ in the same direction as $\pi^o$ orients $e_{ik}$ in $G^o$ (with respect to Figure~\ref{fig:subdivide}). That is, if $\pi^o$ orients $e_{ik}$ toward $i$ in $G^o$, then $\pi$ orients $e_{ij}$ toward $i$ in $G$. Otherwise, $\pi^o$ orients $e_{ik}$ toward $k$ in $G^o$, and $\pi$ orients $e_{ij}$ toward $j$ in $G$.

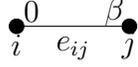
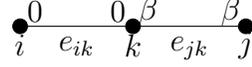
\begin{figure}
	\centering
	\begin{subfigure}{0.45\textwidth}
		\centering
		\begin{tikzpicture}
			\node[redvertex] (i) at (0, 0) {};
			\node[redvertex] (j) at (1.5, 0) {};
			
			\node (li) at (0, -0.25) {$i$};
			\node (lj) at (1.5, -0.25) {$j$};
			
			\node (l1) at (0.2, 0.2) {$0$};
			\node (l2) at (1.3, 0.2) {$\beta$};
			
			\draw (i) to node [below] {$e_{ij}$} (j);
	\end{tikzpicture}
	\caption{[the edge $e_{ij}$ in $G$}
	\end{subfigure}
	\begin{subfigure}{0.45\textwidth}
		\centering
		\begin{tikzpicture}
			\node[redvertex] (i) at (0, 0) {};
			\node[redvertex] (j) at (3, 0) {};
			\node[redvertex] (k) at (1.5, 0) {};
			
			\node (li) at (0, -0.25) {$i$};
			\node (lk) at (1.5, -0.25) {$k$};
			\node (lj) at (3, -0.25) {$j$};
			
			\node (l1) at (0.2, 0.2) {$0$};
			\node (l2) at (1.3, 0.2) {$0$};
			\node (l3) at (1.7, 0.2) {$\beta$};
			\node (l4) at (2.8, 0.2) {$\beta$};
			
			\draw (i) to node [below] {$e_{ik}$} (k);
			\draw (k) to node [below] {$e_{jk}$} (j);
	\end{tikzpicture}
	\caption{after subdivision in $G^o$}
	\end{subfigure}
	
	\caption{Subdivision of $e_{ij}$ during the construction of $G^o$ by \textsc{FindEFXOrientation}. The labels above an edge indicate the utility the edge has to its two endpoints. We write $\beta \coloneqq u_j(e_{ij})$ for clarity.}
	\label{fig:subdivide}
\end{figure}

\begin{algorithm}
	\caption{$\textsc{FindEFXOrientation}(G,U)$}
	\label{alg:FindEFXOrientation}
	
	\hspace*{\algorithmicindent} \textbf{Input:} An instance $(G, U)$ of \textsc{EFX$_0$-Orientation}.
	
	\hspace*{\algorithmicindent} \textbf{Output:} An EFX$_0$ orientation $\pi$ of $G$ if one exists and \texttt{false} otherwise.
	
	\begin{algorithmic}[1] 
		\State $(G^o, u^o) \gets$ the instance of \textsc{EFX$_0$-Orientation-Objective} defined above

		\State $\pi^o \gets \textsc{FindEFXOrientObj}(G^o, u^o)$\label{line:3.subroutine} \Comment{Algorithm~\ref{alg:FindEFXOrientObj}}
		
		\If{$\pi^o = \texttt{false}$}
		\State \Return \texttt{false}\label{line:3.return-false}
		\EndIf
		
		\State \Return $\pi \coloneqq$ the orientation of $G$ defined above \label{line:3.return-pi}
		
	\end{algorithmic}
\end{algorithm}

\begin{lemma}\label{lemma:reduc:FindEFXOrientation}
	The following statements hold for \textsc{FindEFXOrientation} (Algorithm~\ref{alg:FindEFXOrientation}):
	\begin{enumerate}
		\item If $G$ has an EFX$_0$ orientation, then so does $G^o$.
		\item If $\pi^o$ is an EFX$_0$ orientation of $G^o$, then $\pi$ is an EFX$_0$ orientation of $G$.
	\end{enumerate}
\end{lemma}
\begin{proof}
	(1): Suppose $\pi$ is an EFX$_0$ orientation of $G$. We construct the orientation $\pi^o$ of $G^o$ as follows. For each edge $e$ of $G$ that was not subdivided during the construction of $G^o$, we let $\pi^o$ orient it in $G^o$ in the same direction that $\pi$ orients it in $G$. On the other hand, suppose $e_{ik}$ and $e_{jk}$ are a pair of edges of $G^o$ that resulted from the subdivision of the edge $e_{ij}$ of $G$. If $\pi$ orients $e_{ij}$ leftward in terms of Figure~\ref{fig:subdivide} (resp.\ rightward), then $\pi^o$ orients both $e_{ik}$ and $e_{jk}$ leftward (resp.\ rightward).
	
	We claim $\pi^o$ is EFX$_0$. We refer to a vertex that is shared between both $G$ and $G^o$ as a {\em real} vertex, and a vertex that exists only in $G^o$ as a {\em fake} vertex. Fake vertices result from the subdivisions during the construction of $G^o$. Clearly, from the local perspective of each real vertex, the orientations $\pi$ and $\pi^o$ appear the same. Since no real vertex strongly envies a neighbour in $\pi$, no real vertex strongly envies a neighbour in $\pi^o$. On the other hand, the construction of $\pi^o$ ensures each fake vertex receives exactly one incident edge, so no fake vertex strongly envies a neighbour in $\pi^o$. Thus, $\pi^o$ is EFX$_0$.
	
	(2): Suppose $\pi^o$ is an EFX$_0$ orientation of $G^o$. Proposition~\ref{prop:one-neg-or-all-dummies} implies for each vertex $i$, the orientation $\pi^o$ contains a unique edge directed toward $i$ or every edge directed toward $i$ is a dummy edge. In particular, for any fake vertex $k$ in $G^o$ and the two edges $e_{ik}, e_{jk}$ incident to it, $e_{ik}$ is a dummy edge and $e_{jk}$ is a negative edge, so at most one of them is directed toward $k$ in $\pi^o$. Hence, there are only three possibilities for how pairs of fake edges are oriented in $\pi^o$ (see Figure~\ref{fig:construct-pi}).
	
	Consider the first two cases depicted in Figure~\ref{fig:construct-pi}. In Figure~\ref{fig:construct-pi}(a), both $e_{ik}$ and $e_{jk}$ are directed rightward in $\pi^o$, and $\pi$ orients the corresponding edge $e_{ij}$ rightward. In Figure~\ref{fig:construct-pi}(b), both $e_{ik}$ and $e_{jk}$ are directed leftward in $\pi^o$, and $\pi$ orients $\pi_{ij}$ leftward. In either case, the real vertices $i$ and $j$ see no difference between $\pi$ and $\pi^o$ locally in terms of the utilities they derive from the edges $e_{ij}, e_{ik}, e_{jk}$ between $i$ and $j$ that they receive, because $u_i(e_{ij}) = u_i(e_{ik})$ and $u_j(e_{ij}) = u_j(e_{jk})$.
	
	In the case depicted by Figure~\ref{fig:construct-pi}(c), the orientation $\pi$ directs the edge $e_{ij}$ toward $i$. So, $i$ sees no difference between $\pi$ and $\pi^o$ locally between $i$ and $j$, and $j$ thinks that it has given an edge of utility $u_j(e_{jk})$ to $i$ by going from $\pi^o$ to $\pi$.
	
	It follows that for each pair of adjacent vertices $i, j$ of $G$, the vertex $i$, having possibly given some of its edges in $\pi^o$ to $j$ going from $\pi^o$ to $\pi$, does not strongly envy $j$ in $\pi$ because it does not strongly envy $j$ in $\pi^o$. Thus, no vertex of $G$ strongly envies another in $\pi$, so $\pi$ is EFX$_0$.
\end{proof}

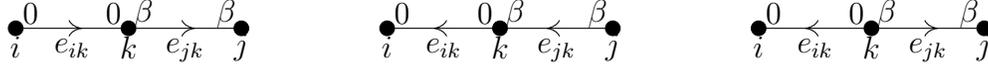
\begin{figure}
	\centering
	\begin{subfigure}{0.3\textwidth}
		\centering
		\begin{tikzpicture}
			\node[redvertex] (i) at (0, 0) {};
			\node[redvertex] (j) at (3, 0) {};
			\node[redvertex] (k) at (1.5, 0) {};
			
			\node (li) at (0, -0.25) {$i$};
			\node (lk) at (1.5, -0.25) {$k$};
			\node (lj) at (3, -0.25) {$j$};
			
			\node (l1) at (0.2, 0.2) {$0$};
			\node (l2) at (1.3, 0.2) {$0$};
			\node (l3) at (1.7, 0.2) {$\beta$};
			\node (l4) at (2.8, 0.2) {$\beta$};
			
			\draw[arc] (i) to node [below] {$e_{ik}$} (k);
			\draw[arc] (k) to node [below] {$e_{jk}$} (j);
	\end{tikzpicture}
	\end{subfigure}
	\begin{subfigure}{0.3\textwidth}
		\centering
		\begin{tikzpicture}
			\node[redvertex] (i) at (0, 0) {};
			\node[redvertex] (j) at (3, 0) {};
			\node[redvertex] (k) at (1.5, 0) {};
			
			\node (li) at (0, -0.25) {$i$};
			\node (lk) at (1.5, -0.25) {$k$};
			\node (lj) at (3, -0.25) {$j$};
			
			\node (l1) at (0.2, 0.2) {$0$};
			\node (l2) at (1.3, 0.2) {$0$};
			\node (l3) at (1.7, 0.2) {$\beta$};
			\node (l4) at (2.8, 0.2) {$\beta$};
			
			\draw[arc] (k) to node [below] {$e_{ik}$} (i);
			\draw[arc] (j) to node [below] {$e_{jk}$} (k);
	\end{tikzpicture}
	\end{subfigure}
	\begin{subfigure}{0.3\textwidth}
		\centering
		\begin{tikzpicture}
			\node[redvertex] (i) at (0, 0) {};
			\node[redvertex] (j) at (3, 0) {};
			\node[redvertex] (k) at (1.5, 0) {};
			
			\node (li) at (0, -0.25) {$i$};
			\node (lk) at (1.5, -0.25) {$k$};
			\node (lj) at (3, -0.25) {$j$};
			
			\node (l1) at (0.2, 0.2) {$0$};
			\node (l2) at (1.3, 0.2) {$0$};
			\node (l3) at (1.7, 0.2) {$\beta$};
			\node (l4) at (2.8, 0.2) {$\beta$};
			
			\draw[arc] (k) to node [below] {$e_{ik}$} (i);
			\draw[arc] (k) to node [below] {$e_{jk}$} (j);
	\end{tikzpicture}
	\end{subfigure}
	
	\caption{The three possibilities involving the fake edges $e_{ik}, e_{jk}$ in the EFX$_0$ orientation $\pi^o$ of $G^o$.}
	\label{fig:construct-pi}
\end{figure}

\mainContributionThree*
\begin{proof}
	We show \textsc{FindEFXOrientation} to be such an algorithm. If $G$ has an EFX$_0$ orientation, then $G^o$ does as well by Lemma~\ref{lemma:reduc:FindEFXOrientation}(1), so the subroutine \textsc{FindEFXOrientObj} produces such an orientation $\pi^o$ of $G^o$ by Lemma~\ref{lemma:FindEFXOrientObj}. In this case, \textsc{FindEFXOrientation} outputs $\pi$, which is an EFX$_0$ orientation of $G$ by Lemma~\ref{lemma:reduc:FindEFXOrientation}(2). Otherwise, $G$ does not have an EFX$_0$ orientation, so so $G^o$ has no EFX$_0$ orientation by Lemma~\ref{lemma:reduc:FindEFXOrientation}(2), causing the subroutine \textsc{FindEFXOrientObj} on Line~\ref{line:3.subroutine} to output \texttt{false} and \textsc{FindEFXOrientation} to also output \texttt{false}.
	
	We now analyze the running time. Constructing $(G^o, u^o)$ takes $O(|E(G)|)$ time because we subdivide each edge at most once. \textsc{FindEFXOrientObj} on Line~\ref{line:3.subroutine} takes $O(|V(G^o)|^2)$ time by Lemma~\ref{lemma:reduc:FindEFXOrientObj}(3). Since $G^o$ is constructed by subdividing each edge of $G$ at most once, and each subdivision introduces a new vertex, we have $|V(G^o)| \leq |V(G)| + |E(G)|$. Hence, $O(|V(G^o)|^2) = O((|V(G)|+|E(G)|)^2)$. Finally, constructing $\pi$ using $\pi^o$ takes $O(|E(G)|)$ time because it requires orienting each edge of $G$. Thus, \textsc{FindEFXOrientation} runs in $O((|V(G)|+|E(G)|)^2)$ time.
\end{proof}

\section{EF1 and EFX$_0$ Orientations of Multigraphs}\label{section:multigraphs}

We turn our attention to multigraphs and consider the problem of deciding if a multigraph has EF1 or EFX$_0$ orientations. We show that both of these problems are NP-complete using a reduction from the NP-complete problem \textsc{Partition} \citep{karp1972reducibility}.

We present two different reductions. The first reduction (Theorem~\ref{thm:main-multi}) is relatively simple but results in a multigraph with two self-loops. The second reduction (Theorem~\ref{thm:main-no-self-loop}) is more complex but has the advantage of not using self-loops.

\mainContributionThreeC*

\begin{proof}
	To see that both of these problems are in NP, suppose one is given an orientation $\pi$ of a multigraph $G$ of chores. To verify whether $\pi$ is EF1 or EFX$_0$, one can simply verify the EF1 and EFX$_0$ conditions between each pair of vertices. This can clearly be done by looping through the set of edges a constant number of times for each pair of vertices, so these problems are in NP.
	
	We reduce an instance $S = \{s_1, s_2, \dots, s_k\}$ of \textsc{Partition} to each of the EF1 and EFX$_0$ orientation problems. We give the EF1 reduction and show how to adapt it for EFX$_0$.
	
	Let $\alpha < -\max_{i} s_i < 0$. Construct a multigraph $G$ on two vertices $a, b$ as follows. For each $s_i \in S$, create an edge $e_i$ between $a$ and $b$ and set $u_a(e_i) = u_b(e_i) = -s_i$. Create two self-loops $e_a, e_b$ at each of $a, b$, respectively, and set $u_a(e_a) = u_b(e_b) = \alpha$. Assume the utility functions are additive.
	
	Partitions $S_a \cup S_b$ of $S$ correspond to orientations of the edges between $a$ and $b$ in a natural way. Moreover, because both $a$ and $b$ receive their respective self-loop which has $\alpha < -\max_{i} s_i$ utility, the EF1 criteria requires that their envy for each other be alleviated when they ignore their respective self-loops. Thus,  $\sum_{x \in S_a} x = \sum_{x \in S_b} x$ if and only if $G$ has an EF1 orientation.
	
	To adapt the reduction for EFX$_0$, let $\alpha = 0$ instead. This results is a correct reduction because the two self-loops have minimum marginal utility if $\alpha = 0$.
\end{proof}

We now show that deciding whether EF1 orientations exist remain NP-complete even if the multigraph $G$ contains no self-loops. To do this, we again rely on a reduction from \textsc{Partition}. For any instance $S = \{s_1, s_2, \dots, s_k\}$ of \textsc{Partition}, we define a multigraph $G$ on three vertices $a, b, c$ as follows. All of the edges $G$ contains has equal utility to both endpoints (called its {\em weight}). For each $s_i \in S$, create an edge between $a$ and $b$ of weight $-s_i$. Then, between $c$ and each of $a, b$, create two edges of weight $-T \coloneqq \sum_{i \in [k]} s_i$. (See Figure~\ref{fig:multi-no-loop} for an illustration.)

Before proving the correctness of this reduction, we first prove a lemma.

\begin{lemma}\label{lemma:multigraph-reduction}
	In any EF1 orientation $\pi = (\pi_a, \pi_b, \pi_c)$ of $G$, at least one edge between $c$ and $a$ is oriented toward $a$ (and similarly between $c$ and $b$ by symmetry).
\end{lemma}
\begin{proof}
	If both edges between $c$ and $a$ are oriented toward $c$, then for either one of them (call it $e$), we have $u_c(\pi_c \setminus \{e\}) \leq -T < 0 = u_c(\pi_a)$, so $\pi$ is not EF1.
\end{proof}

\mainContributionThreeD*

\begin{proof}
	We have already seen in the proof of Theorem~\ref{thm:main-no-self-loop} that this problem is in NP. We now show that $S$ has an equipartition if and only if $G$ has an EF1 orientation. Suppose $S = S_1 \cup S_2$ is an equipartition. Let $\pi$ be the orientation that orients the edges between $a$ and $b$ in the way that naturally corresponds to $S_1 \cup S_2$, the two edges between $c$ and $a$ in opposite directions (one toward $c$ and the other toward $a$), and the two edges between $c$ and $b$ in opposite directions. We claim $\pi$ is EF1.
	
	First consider the envy that $c$ experiences. It suffices to consider the envy from $c$ to $a$ by symmetry. By construction, $u_c(\pi_c) = -2T < -T = u_c(\pi_a)$. Because both edges directed toward $c$ have utility $-T$, disregarding either of them alleviates the envy from $c$ to $a$.
	
	Next, consider the envy that $a$ and $b$ experience. Again by symmetry, it suffices to consider only $a$. The edge $e$ between $a$ and $c$ directed toward $a$ has the greatest negative utility among the edges in $\pi_a$ by the definition of $T$, so the definition of EF1 requires the envy $a$ experiences to be alleviated if $e$ is disregarded. Indeed, we have $u_a(\pi_a \setminus \{e\}) = (-1/2)T = u_a(\pi_b)$ and $u_a(\pi_a \setminus \{e\}) = (-1/2)T < -T = u_a(\pi_c)$. So, $\pi$ is EF1.
	
	Conversely, suppose $S$ has no equipartition. Let $\pi$ be any orientation of $G$ and $S_1 \cup S_2$ be the partition of $S$ that corresponds to the way $\pi$ orients the edges between $a$ and $b$. More specifically, let $s_i \in S_1$ if and only if the edge between $a$ and $b$ corresponding to $s_i$ is directed toward $a$, and let $S_2 = S \setminus S_1$.
	
	If both edges between $c$ and $a$ are directed toward $c$ or if both edges between $c$ and $b$ are directed toward $b$, then $\pi$ is not EF1 by Lemma~\ref{lemma:multigraph-reduction}, so we are done. It remains to consider the case that an edge $e_{ac}$ between $c$ and $a$ is directed toward $a$ and an edge $e_{bc}$ between $c$ and $b$ is directed toward $b$. The definition of EF1 requires $-\sum_{x \in S_1} x = u_a(\pi_a \setminus \{e_{ac}\}) \leq u_a(\pi_b) = -\sum_{x \in S_2} x$. Similarly, $-\sum_{x \in S_2} x = u_b(\pi_b \setminus \{e_{bc}\}) \leq u_b(\pi_a) = -\sum_{x \in S_1} x$. Thus, $\sum_{x \in S_1} = \sum_{x \in S_1} x$, contradicting the fact that $S$ has no equipartition.
\end{proof}

\begin{figure}
	\centering
	\begin{tikzpicture}
		\node[redvertex] (a) at (-2, 0) {};
		\node[redvertex] (b) at (2, 0) {};
		\node[redvertex] (c) at (0, 3) {};
		
		\node (la) at (-2.3, 0) {$a$};
		\node (lb) at (2.3, 0) {$b$};
		\node (lc) at (0, 3.3) {$c$};

		\draw (a) to[bend left=35] node [left] {$-T$} (c);
		\draw (a) to[bend right=15] node [left] {$-T$} (c);
		
		\draw (b) to[bend left=15] node [right] {$-T$} (c);
		\draw (b) to[bend right=35] node [right] {$-T$} (c);
		
		\draw (a) to[bend left=25] node [above] {$-s_1$} (b);
		\draw (a) to node [above] {$-s_2$} (b);
		\node at (0, -0.2) {$\vdots$};
		\draw (a) to[bend right=25] node [below] {$-s_k$} (b);
	\end{tikzpicture}
	\caption{The multigraph $G$ in Theorem~\ref{thm:main-no-self-loop}. Here, $T \coloneqq \sum_{i \in [k]} s_i$.}
	\label{fig:multi-no-loop}
\end{figure}

\section{On the Additive Case}\label{section:additive}

All of our results also apply to the additive case in which every agent has an additive utility function, without any modification to the theorem statements and proofs. More generally, our results hold even if we allow each agent to have either a monotone utility function as previously described in Section~\ref{section:prelim} or an additive utility function, independently of other agents. We explain the reason below.

Our results pertaining to EF1 orientations stem from Proposition~\ref{prop:EF1-graphs}, which translates the EF1 orientation condition for each agent into a logically equivalent graph theoretical condition. Thus, any role that assumptions made on utility functions play lie entire within the proof of Proposition~\ref{prop:EF1-graphs}. It is easy to verify the correctness of this proof for additive utility functions as well.

As for EFX$_0$ orientations, Proposition~\ref{prop:one-neg-or-all-dummies} translates the EFX$_0$ orientation condition for each agent in an objective instances into an equivalent graph theoretical condition. So, we need to verify Proposition~\ref{prop:one-neg-or-all-dummies} and also the reduction of \textsc{EFX$_0$-Orientation} to \textsc{EFX$_0$-Orientation-Objective} for additive utility functions (i.e.\ Lemma~\ref{lemma:reduc:FindEFXOrientation}). Again, it is straightforward to verify the proofs of both of these results for additive utility functions.

\section{Discussion}\label{section:conclusion}

In this chapter, we determined the complexity of finding EF1 and EFX$_0$ orientations of chores in both the graph and multigraph settings. Our results resolve a conjecture due to \citet{zhou2024complete} and show a surprising contrast between goods and chores cases.

\chapter{Pareto Optimality and Orientations}\label{chapter:po-and-orientations}

We conclude this dissertation with a discussion on the relationship between Pareto optimality and orientations. This chapter was inspired by the questions raised by Nisarg Shah during the oral examination.

Throughout this chapter, we assume that agents have additive utility functions. Recall that in the additive setting, an item $o_j$ is {\em relevant} to an agent $i$ if $u_i(o_j) \neq 0$. Recall also that an {\em orientation} is an allocation in which each item is allocated to an agent to whom the item is relevant. Until now, we have only considered orientations in the context of the graphical model, which is characterized by the assumption that each item is relevant to at most two agents. However, orientations are also meaningful in the general model in which each item can be relevant to arbitrarily many agents. We consider orientations in the general model throughout this chapter, except when otherwise specified.

As briefly discussed in Subsection~\ref{subsec:technical-lemmas} in the paragraph preceding Theorem~\ref{thm:only-trivial}, PO allocations of {\em non-dummy} goods (i.e.\ goods relevant to at least one agent) are orientations. This is because if a good is not allocated to an agent to whom it is relevant, then reallocating it to such an agent results in a Pareto improvement. We state this formally below.

\begin{proposition}\label{prop:po-implies-orientation}
	Any PO allocation of non-dummy goods is an orientation. \qed
\end{proposition}

This fact is useful as it guarantees that EF1+PO allocations and EFX$_0$+PO allocations are orientations. However, the same cannot be said for instances containing dummy goods, even for small instances containing only 4 agents. More specifically, Example~\ref{ex:po-alloc-not-orientation} shows that it is possible for a goods instance to have an EFX$_0$+PO allocation that is not an orientation.

\begin{example}[A Goods Instance with an EFX$_0$+PO Allocation but no EFX$_0$ Orientations]\label{ex:po-alloc-not-orientation}
	We construct such a graphical instance represented by a bi-valued symmetric graph $G$. Let $G$ be the complete graph $K_4$ with $V(G) = \{1, 2, 3, 4\}$. Let the edge between 1 and 2 and the edge between 3 and 4 be heavy edges of weight 1, and the remaining edges be light edges of weight 0. Assume that $u_i(e) = 0$ whenever $e$ is not incident to $i$. (See Figure~\ref{fig:no-efx-orientation-k4} for an illustration.)
	
	\begin{figure}
		\centering
		\begin{tikzpicture}
			\node[redvertex] (1) at (0,0) {};
			\node[redvertex] (2) at (0,1.5) {};
			\node[redvertex] (3) at (1.5,1.5) {};
			\node[redvertex] (4) at (1.5,0) {};
			\node (l1) at (-0.2, -0.2) {1};
			\node (l2) at (-0.2, 1.7) {2};
			\node (l3) at (1.7, 1.7) {3};
			\node (l4) at (1.7, -0.2) {4};
			
			\draw[thick] (1) to (2) (3) to (4);
			\draw[dashed] (1) to (3) (1) to (4)
				(2) to (3) (2) to (4);
		\end{tikzpicture}

		\caption{An example of a goods instances with an EFX$_0$+PO allocation but no EFX$_0$ orientation. Each solid (resp.\ dashed) edge represents a heavy (resp.\ light) edge.}
		\label{fig:no-efx-orientation-k4}
	\end{figure}
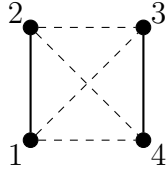
	
	To see that $G$ has an EFX$_0$+PO allocation, observe that it represents an additive 2-valued instance (i.e.\ there exist $a \geq b \geq 0$ such that for each agent $i$ and each good $o_j$, we have $u_i(o_j) \in \{a, b\}$). By Corollary~\ref{cor:efx-additive-2-valued}, MNW implies EFX$_0$+PO for such instances. Since MNW allocations always exist by definition, $G$ has an EFX$_0$+PO allocation.
	
	For an explicit example of an EFX$_0$+PO allocation, consider the allocation $\pi$ in which the two heavy edges are allocated to their upper endpoints (i.e.\ vertices 2 and 3), and all of the light edges are allocated to vertex 1. None of 2, 3, and 4 are strongly envied by any vertex because each of them is allocated at most one edge. Vertex 1 is also not strongly envied because $\pi_i(\pi_1) = 0$ for every vertex $i$. Thus, $\pi$ is EFX$_0$. To see that $\pi$ is PO, it suffices to observe that $\pi$ orients both heavy edges, and that reallocating light edges with zero weight cannot result in a Pareto improvement.
	
	We now show that $G$ has no EFX$_0$ orientation. Let $\pi$ be any orientation of $G$. By the symmetry of $G$, we can assume that the two heavy edges are directed upward toward 2 and 3 in $\pi$, respectively. Again by the symmetry of $G$, we can additionally assume that the light edge $e$ between 2 and 3 is directed leftward toward 2. We now have $u_1(\pi_1) = 0 < 1 = u_1(\pi_2 \setminus \{e\})$ for the edge $e \in \pi_2$, so $\pi$ is not EFX$_0$.
\end{example}

As for chores, the analogue of Proposition~\ref{prop:po-implies-orientation} does not hold. Specifically, PO allocation of chores are not necessarily orientations. In fact, any PO allocation necessarily allocates each chore to an agent to whom the chore has zero utility whenever such an agent exists.

Consider the goal of finding fair allocations for situations in which some chores are unsuitable for some agents. One approach is to let chores be relevant only to the agents to whom they are suitable, and to require allocations to be orientations. Doing so ensures that all {\em non-dummy} chores (i.e.\ chores relevant to at least one agent) are allocated to an agent to whom it is suitable. This is the approach that we studied in Chapter~\ref{chapter:Fair-Orientations-Chores}. Unfortunately, our results in Section~\ref{section:EF1} imply that EF1 orientations of chores do not always exist. Thus, this approach is limited in the sense that it can fail to produce a fair allocation (i.e.\ an EF1 orientation).

Another natural approach to ensure that an agent $i$ is not allocated an unsuitable chore $o_j$ is to set $u_i(o_j)$ to a sufficiently large negative value and requiring allocations to be PO. \citet{mahara2025existence} recently showed that EF1+PO allocations of chores always exist\footnote{In fact, \citet{mahara2025existence} showed the stronger fact that EF1+fPO allocations of chores always exist.}. So, in a sense, this approach is an improvement over the previous because it always produces a fair allocation (i.e.\ an EF1+PO allocation).

A natural question that arises is what types of allocations the second approach produces, in situations where the first approach fails to produce a fair allocation\footnote{This question was raised by Nisarg Shah during the oral examination.}. More precisely, let $I$ be a chores instances that does not have an EF1 orientation when unsuitable chores have zero utility. Suppose we modify $I$ by changing the utilities of the unsuitable chores to be an arbitrarily large negative value (say $-\infty$). What can be said about EF1+PO allocations of the resulting instance?

As we will see, there exists such an instance in which an EF1+PO allocation is not an orientation (Example~\ref{ex:ef1-po-not-orientation}), and another such instance in which all EF1+PO allocations are orientations (Example~\ref{ex:ef1-po-orientation}).

\begin{example}[An EF1+PO Allocation that is not an Orientation]\label{ex:ef1-po-not-orientation}
	Consider the instance represented by the following graph $G$. Assume the edge $o_5$ has $-5$ utility to both of its endpoints, and that each other edge has $-1$ utility to its endpoints. Furthermore, assume that whenever an edge is not incident to a vertex, it is {\em unsuitable} for that vertex.
	$$
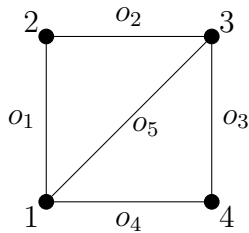
\begin{tikzpicture}
		\node[redvertex] (1) at (0,0) {};
		\node[redvertex] (2) at (0,2.2) {};
		\node[redvertex] (3) at (2.2,2.2) {};
		\node[redvertex] (4) at (2.2,0) {};
		\node (l1) at (-0.2, -0.2) {1};
		\node (l2) at (-0.2, 2.4) {2};
		\node (l3) at (2.4, 2.4) {3};
		\node (l4) at (2.4, -0.2) {4};
		
		\draw (1) to node[left] {$o_1$} (2);
		\draw (2) to node[above] {$o_2$} (3);
		\draw (3) to node[right] {$o_3$} (4);
		\draw (4) to node[below] {$o_4$} (1);
		\draw (1) to node[right,pos=0.45] {$o_5$} (3);
	\end{tikzpicture}$$
	
	In the first approach, we assume that whenever an edge $o$ is unsuitable for a vertex $i$, we have $u_i(o)=0$. By Proposition~\ref{prop:EF1-graphs}, an orientation $\pi$ of a graph $G$ is EF1 if and only if each vertex $i$ receives at most one edge of negative utility to it. However, such an orientation does not exist by the pigeonhole principle because there are more edges than vertices. In particular, $G$ has no EF1 orientation if we assume that unsuitable edges have zero utility.
	
	In the second approach, we assume that whenever an edge $o$ is unsuitable for a vertex $i$, we have $u_i(o) = -\infty$. We exhibit an EF1+PO allocation that is not an orientation. Let $\pi$ be the allocation that orients the edges on the 4-cycle in a counterclockwise way, and allocates the edge $o_5$ to vertex 4. In other words, $\pi$ is the allocation indicated by the bold underlined entries in the following matrix.
	$$\begin{blockarray}{cccccc}
		& o_1 & o_2 & o_3 & o_4 & o_5 \\
		\begin{block}{c[ccccc]}
			1 & \underline{\mathbf{-1}} & -\infty & -\infty&-1&-5\\
			2 & -1 & \underline{\mathbf{-1}} & -\infty&-\infty&-\infty\\
			3 & -\infty  & -1 & \underline{\mathbf{-1}}&-\infty&-5\\
			4 & -\infty & -\infty & -1 &\underline{\mathbf{-1}}&\underline{\mathbf{-\infty}}\\
		\end{block}
	\end{blockarray}$$
	
	It is straightforward to verify that $\pi$ is EF1. To see that $\pi$ is PO, observe that reassigning $o_5$ to a different vertex than 4 does not result in a Pareto improvement. This is because doing so necessarily decreases the utility of the bundle belonging to some vertex. Thus, $\pi$ is an EF1+PO allocation that is not an orientation.
\end{example}

\begin{example}[Every EF1+PO Allocation is an Orientation]\label{ex:ef1-po-orientation}
Consider the graphical instance represented by the following graph $G$. Assume that each edge has $-1$ utility to both of its endpoints. Furthermore, assume that whenever an edge is not incident to a vertex, it is {\em unsuitable} for that vertex.
	$$
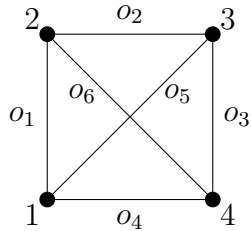
\begin{tikzpicture}
		\node[redvertex] (1) at (0,0) {};
		\node[redvertex] (2) at (0,2.2) {};
		\node[redvertex] (3) at (2.2,2.2) {};
		\node[redvertex] (4) at (2.2,0) {};
		\node (l1) at (-0.2, -0.2) {1};
		\node (l2) at (-0.2, 2.4) {2};
		\node (l3) at (2.4, 2.4) {3};
		\node (l4) at (2.4, -0.2) {4};
		
		\draw (1) to node[left] {$o_1$} (2);
		\draw (2) to node[above] {$o_2$} (3);
		\draw (3) to node[right] {$o_3$} (4);
		\draw (4) to node[below] {$o_4$} (1);
		\draw (1) to node[right,pos=0.65] {$o_5$} (3);
		\draw (2) to node[left,pos=0.35] {$o_6$} (4);
	\end{tikzpicture}$$
	
	In the first approach, we assume that whenever an edge $o$ is unsuitable for a vertex $i$, we have $u_i(o)=0$. By Proposition~\ref{prop:EF1-graphs}, an orientation $\pi$ of a graph $G$ is EF1 if and only if each vertex $i$ receives at most one edge of negative utility to it. By the pigeonhole principle, $G$ has no EF1 orientation because there are more edges than vertices.
	
	On the other hand, in the second approach, we assume that whenever an edge $o$ is unsuitable for a vertex $i$, we have $u_i(o) = -\infty$. Thus, the matrix representing this instance is as follows.
	$$\begin{blockarray}{ccccccc}
		& o_1 & o_2 & o_3 & o_4 & o_5 & o_6 \\
		\begin{block}{c[cccccc]}
			1 & -1 &  -\infty & -\infty   & -1 & -1 &  -\infty \\
			2 & -1 & -1 &  -\infty & -\infty &  -\infty & -1 \\
			3 & -\infty & -1 & -1 & -\infty & -1 &  -\infty \\
			4 & -\infty &  -\infty  & -1 & -1 &  -\infty  & -1 \\
		\end{block}
	\end{blockarray}$$
	
	We prove that all EF1+PO allocations are orientations by showing that in any PO allocation, no vertex receives an edge of $-\infty$ utility. 
	
	By the symmetry of $K_4$, it suffices to show that vertex 1 does not receive the edge $o_2$ in any PO allocation. Suppose for contradiction that some PO allocation $\pi$ allocates $o_2$ to 1. If $o_1 \in \pi_2 \cup \pi_3 \cup \pi_4$, then exchanging the edges $o_2 \in \pi_1$ and $o_1 \in \pi_2 \cup \pi_3 \cup \pi_4$ results in a Pareto improvement over $\pi$, which is not possible because $\pi$ is PO. So, $o_1 \in \pi_1$. For similar reasons, $o_4, o_5 \in \pi_1$. Thus, vertex 1 receives at least 4 edges, each of negative utility to 1. Since there are only 6 edges, the pigeonhole principle implies that some vertex $j \neq 1$ receives no edges. Clearly, for any edge $o \in \pi_1$, we have $u_1(\pi_1 \setminus \{o\}) < 0 = u_1(\pi_j)$, so $\pi$ is not EF1. This is a contradiction as required.
	
	For an explicit example of an EF1+PO orientation, consider any orientation in which at most 2 edges is directed toward each vertex. Such an orientation is EF1 because the EF1 condition holds for every pair $i \neq j$ because $|\pi_i| \leq 2$ and $|\pi_j| \geq 1$. Moreover, such an orientation is clearly PO.
\end{example}



\backmatter 

	\addtoToC{Bibliography}
	\bibliographystyle{abbrvnat} 
	\bibliography{bibliography}

\begin{thebibliography}{80}
\providecommand{\natexlab}[1]{#1}
\providecommand{\url}[1]{\texttt{#1}}
\expandafter\ifx\csname urlstyle\endcsname\relax
  \providecommand{\doi}[1]{doi: #1}\else
  \providecommand{\doi}{doi: \begingroup \urlstyle{rm}\Url}\fi

\bibitem[Afshinmehr et~al.(2024)Afshinmehr, Danaei, Kazemi, Mehlhorn, and
  Rathi]{afshinmehr2024efx}
M.~Afshinmehr, A.~Danaei, M.~Kazemi, K.~Mehlhorn, and N.~Rathi.
\newblock {EFX} allocations and orientations on bipartite multi-graphs: A
  complete picture.
\newblock \emph{arXiv preprint arXiv:2410.17002}, 2024.

\bibitem[Amanatidis et~al.(2021)Amanatidis, Birmpas, Filos-Ratsikas, Hollender,
  and Voudouris]{amanatidis2021maximum}
G.~Amanatidis, G.~Birmpas, A.~Filos-Ratsikas, A.~Hollender, and A.~A.
  Voudouris.
\newblock Maximum {N}ash welfare and other stories about {EFX}.
\newblock \emph{Theoretical Computer Science}, 863:\penalty0 69--85, 2021.

\bibitem[Amanatidis et~al.(2022)Amanatidis, Birmpas, Filos-Ratsikas, Voudouris,
  et~al.]{amanatidis2022fair}
G.~Amanatidis, G.~Birmpas, A.~Filos-Ratsikas, A.~Voudouris, et~al.
\newblock Fair division of indivisible goods: A survey.
\newblock In \emph{IJCAI International Joint Conference on Artificial
  Intelligence}, pages 5385--5393. International Joint Conferences on
  Artificial Intelligence, 2022.

\bibitem[Amanatidis et~al.(2024)Amanatidis, Filos-Ratsikas, and
  Sgouritsa]{amanatidis2024pushing}
G.~Amanatidis, A.~Filos-Ratsikas, and A.~Sgouritsa.
\newblock Pushing the frontier on approximate {EFX} allocations.
\newblock In \emph{Proceedings of the 25th ACM Conference on Economics and
  Computation}, pages 1268--1286, 2024.

\bibitem[Aspvall et~al.(1979)Aspvall, Plass, and Tarjan]{aspvall1979linear}
B.~Aspvall, M.~F. Plass, and R.~E. Tarjan.
\newblock A linear-time algorithm for testing the truth of certain quantified
  boolean formulas.
\newblock \emph{Information processing letters}, 8\penalty0 (3):\penalty0
  121--123, 1979.

\bibitem[Aziz and Mackenzie(2016)]{aziz2016discrete}
H.~Aziz and S.~Mackenzie.
\newblock A discrete and bounded envy-free cake cutting protocol for any number
  of agents.
\newblock In \emph{2016 IEEE 57th Annual Symposium on Foundations of Computer
  Science (FOCS)}, pages 416--427. IEEE, 2016.

\bibitem[Aziz et~al.(2022)Aziz, Caragiannis, Igarashi, and Walsh]{aziz2022fair}
H.~Aziz, I.~Caragiannis, A.~Igarashi, and T.~Walsh.
\newblock Fair allocation of indivisible goods and chores.
\newblock \emph{Autonomous Agents and Multi-Agent Systems}, 36:\penalty0 1--21,
  2022.

\bibitem[Barman et~al.(2018{\natexlab{a}})Barman, Krishnamurthy, and
  Vaish]{barman2018finding}
S.~Barman, S.~K. Krishnamurthy, and R.~Vaish.
\newblock Finding fair and efficient allocations.
\newblock In \emph{Proceedings of the 2018 ACM Conference on Economics and
  Computation}, pages 557--574, 2018{\natexlab{a}}.

\bibitem[Barman et~al.(2018{\natexlab{b}})Barman, Krishnamurthy, and
  Vaish]{barman2018greedy}
S.~Barman, S.~K. Krishnamurthy, and R.~Vaish.
\newblock Greedy algorithms for maximizing {N}ash social welfare.
\newblock In \emph{Proceedings of the 17th International Conference on
  Autonomous Agents and MultiAgent Systems}, pages 7--13, 2018{\natexlab{b}}.

\bibitem[B{\'e}rczi et~al.(2024)B{\'e}rczi, B{\'e}rczi-Kov{\'a}cs, Boros,
  Gedefa, Kamiyama, Kavitha, Kobayashi, and Makino]{berczi2024envy}
K.~B{\'e}rczi, E.~R. B{\'e}rczi-Kov{\'a}cs, E.~Boros, F.~T. Gedefa,
  N.~Kamiyama, T.~Kavitha, Y.~Kobayashi, and K.~Makino.
\newblock Envy-free relaxations for goods, chores, and mixed items.
\newblock \emph{Theoretical Computer Science}, 1002:\penalty0 114596, 2024.

\bibitem[Berger et~al.(2022)Berger, Cohen, Feldman, and Fiat]{berger2022almost}
B.~Berger, A.~Cohen, M.~Feldman, and A.~Fiat.
\newblock Almost full {EFX} exists for four agents.
\newblock In \emph{Proceedings of the AAAI Conference on Artificial
  Intelligence}, volume~36, pages 4826--4833, 2022.

\bibitem[Bhaskar and Pandit(2024)]{bhaskar2024efx}
U.~Bhaskar and Y.~Pandit.
\newblock {EFX} allocations on some multi-graph classes.
\newblock \emph{arXiv preprint arXiv:2412.06513}, 2024.

\bibitem[Bhaskar et~al.(2021)Bhaskar, Sricharan, and
  Vaish]{bhaskar2021approximate}
U.~Bhaskar, A.~Sricharan, and R.~Vaish.
\newblock On approximate envy-freeness for indivisible chores and mixed
  resources.
\newblock \emph{Approximation, Randomization, and Combinatorial Optimization.
  Algorithms and Techniques}, 2021.

\bibitem[Bouveret and Lema{\^\i}tre(2016)]{BL16}
S.~Bouveret and M.~Lema{\^\i}tre.
\newblock Characterizing conflicts in fair division of indivisible goods using
  a scale of criteria.
\newblock \emph{Autonomous Agents and Multi-Agent Systems}, 30\penalty0
  (2):\penalty0 259--290, 2016.

\bibitem[Brustle et~al.(2020)Brustle, Dippel, Narayan, Suzuki, and
  Vetta]{brustle2020one}
J.~Brustle, J.~Dippel, V.~V. Narayan, M.~Suzuki, and A.~Vetta.
\newblock One dollar each eliminates envy.
\newblock In \emph{Proceedings of the 21st ACM Conference on Economics and
  Computation}, pages 23--39, 2020.

\bibitem[Budish(2011)]{budish2011combinatorial}
E.~Budish.
\newblock The combinatorial assignment problem: Approximate competitive
  equilibrium from equal incomes.
\newblock \emph{Journal of Political Economy}, 119\penalty0 (6):\penalty0
  1061--1103, 2011.

\bibitem[Camacho et~al.(2023)Camacho, Fonseca-Delgado, P{\'e}rez, and
  Tapia]{camacho2023generalized}
F.~Camacho, R.~Fonseca-Delgado, R.~P. P{\'e}rez, and G.~Tapia.
\newblock Generalized binary utility functions and fair allocations.
\newblock \emph{Mathematical Social Sciences}, 121:\penalty0 50--60, 2023.

\bibitem[Caragiannis et~al.(2019{\natexlab{a}})Caragiannis, Gravin, and
  Huang]{caragiannis2019envy}
I.~Caragiannis, N.~Gravin, and X.~Huang.
\newblock Envy-freeness up to any item with high {N}ash welfare: The virtue of
  donating items.
\newblock In \emph{Proceedings of the 2019 ACM Conference on Economics and
  Computation}, pages 527--545, 2019{\natexlab{a}}.

\bibitem[Caragiannis et~al.(2019{\natexlab{b}})Caragiannis, Kurokawa, Moulin,
  Procaccia, Shah, and Wang]{caragiannis2019unreasonable}
I.~Caragiannis, D.~Kurokawa, H.~Moulin, A.~D. Procaccia, N.~Shah, and J.~Wang.
\newblock The unreasonable fairness of maximum {N}ash welfare.
\newblock \emph{ACM Transactions on Economics and Computation (TEAC)},
  7\penalty0 (3):\penalty0 1--32, 2019{\natexlab{b}}.

\bibitem[Caragiannis et~al.(2022)Caragiannis, Micha, and
  Shah]{caragiannis2022little}
I.~Caragiannis, E.~Micha, and N.~Shah.
\newblock A little charity guarantees fair connected graph partitioning.
\newblock In \emph{Proceedings of the AAAI Conference on Artificial
  Intelligence}, volume~36, pages 4908--4916, 2022.

\bibitem[Chartrand et~al.(2024)Chartrand, Jordon, Vatter, and
  Zhang]{chartrand2024graphs}
G.~Chartrand, H.~Jordon, V.~Vatter, and P.~Zhang.
\newblock \emph{Graphs \& digraphs}.
\newblock Chapman and Hall/crc, 2024.

\bibitem[Chaudhury et~al.(2020)Chaudhury, Garg, and Mehlhorn]{chaudhury2020efx}
B.~R. Chaudhury, J.~Garg, and K.~Mehlhorn.
\newblock {EFX} exists for three agents.
\newblock In \emph{Proceedings of the 21st ACM Conference on Economics and
  Computation}, pages 1--19, 2020.

\bibitem[Chaudhury et~al.(2021{\natexlab{a}})Chaudhury, Garg, Mehlhorn, Mehta,
  and Misra]{chaudhury2021improving}
B.~R. Chaudhury, J.~Garg, K.~Mehlhorn, R.~Mehta, and P.~Misra.
\newblock Improving {EFX} guarantees through rainbow cycle number.
\newblock In \emph{Proceedings of the 22nd ACM Conference on Economics and
  Computation}, pages 310--311, 2021{\natexlab{a}}.

\bibitem[Chaudhury et~al.(2021{\natexlab{b}})Chaudhury, Kavitha, Mehlhorn, and
  Sgouritsa]{chaudhury2021little}
B.~R. Chaudhury, T.~Kavitha, K.~Mehlhorn, and A.~Sgouritsa.
\newblock A little charity guarantees almost envy-freeness.
\newblock \emph{SIAM Journal on Computing}, 50\penalty0 (4):\penalty0
  1336--1358, 2021{\natexlab{b}}.

\bibitem[Christodoulou and Mastrakoulis(2025)]{christodoulou2025exact}
G.~Christodoulou and S.~Mastrakoulis.
\newblock Exact and approximate maximin share allocations in multi-graphs.
\newblock \emph{arXiv preprint arXiv:2506.20317}, 2025.

\bibitem[Christodoulou et~al.(2023)Christodoulou, Fiat, Koutsoupias, and
  Sgouritsa]{christodoulou2023fair}
G.~Christodoulou, A.~Fiat, E.~Koutsoupias, and A.~Sgouritsa.
\newblock Fair allocation in graphs.
\newblock In \emph{Proceedings of the 24th ACM Conference on Economics and
  Computation}, pages 473--488, 2023.

\bibitem[Christoforidis and Santorinaios(2024)]{christoforidis2024pursuit}
V.~Christoforidis and C.~Santorinaios.
\newblock On the pursuit of {EFX} for chores: non-existence and approximations.
\newblock In \emph{Proceedings of the Thirty-Third International Joint
  Conference on Artificial Intelligence}, pages 2713--2721, 2024.

\bibitem[Cole and Gkatzelis(2015)]{cole2015approximating}
R.~Cole and V.~Gkatzelis.
\newblock Approximating the {N}ash social welfare with indivisible items.
\newblock In \emph{Proceedings of the forty-seventh annual ACM symposium on
  Theory of computing}, pages 371--380, 2015.

\bibitem[Cook(1971)]{cook1971complexity}
S.~A. Cook.
\newblock The complexity of theorem-proving procedures.
\newblock In \emph{Proceedings of the third annual ACM symposium on Theory of
  computing}, pages 151--158, 1971.

\bibitem[Darmann and Schauer(2015)]{darmann2015maximizing}
A.~Darmann and J.~Schauer.
\newblock Maximizing {N}ash product social welfare in allocating indivisible
  goods.
\newblock \emph{European Journal of Operational Research}, 247\penalty0
  (2):\penalty0 548--559, 2015.

\bibitem[Deligkas et~al.(2024)Deligkas, Eiben, Goldsmith, and
  Korchemna]{deligkas2024ef1}
A.~Deligkas, E.~Eiben, T.-L. Goldsmith, and V.~Korchemna.
\newblock {EF1} and {EFX} orientations.
\newblock \emph{arXiv preprint arXiv:2409.13616}, 2024.

\bibitem[Dubins and Spanier(1961)]{dubins1961cut}
L.~E. Dubins and E.~H. Spanier.
\newblock How to cut a cake fairly.
\newblock \emph{The American Mathematical Monthly}, 68\penalty0 (1P1):\penalty0
  1--17, 1961.

\bibitem[Ebadian et~al.(2022)Ebadian, Peters, and Shah]{ebadian2022fairly}
S.~Ebadian, D.~Peters, and N.~Shah.
\newblock How to fairly allocate easy and difficult chores.
\newblock In \emph{21st International Conference on Autonomous Agents and
  Multiagent Systems}, 2022.

\bibitem[Edmonds and Pruhs(2011)]{edmonds2011cake}
J.~Edmonds and K.~Pruhs.
\newblock Cake cutting really is not a piece of cake.
\newblock \emph{ACM Transactions on Algorithms (TALG)}, 7\penalty0
  (4):\penalty0 1--12, 2011.

\bibitem[Edward~Su(1999)]{edward1999rental}
F.~Edward~Su.
\newblock Rental harmony: Sperner's lemma in fair division.
\newblock \emph{The American mathematical monthly}, 106\penalty0 (10):\penalty0
  930--942, 1999.

\bibitem[Even and Paz(1984)]{even1984note}
S.~Even and A.~Paz.
\newblock A note on cake cutting.
\newblock \emph{Discrete Applied Mathematics}, 7\penalty0 (3):\penalty0
  285--296, 1984.

\bibitem[Feige(2022)]{feige2022maximin}
U.~Feige.
\newblock Maximin fair allocations with two item values.
\newblock \emph{Unpublished Manuscript}, 2022.

\bibitem[Feige et~al.(2021)Feige, Sapir, and Tauber]{feige2021tight}
U.~Feige, A.~Sapir, and L.~Tauber.
\newblock A tight negative example for {MMS} fair allocations.
\newblock In \emph{International Conference on Web and Internet Economics},
  pages 355--372. Springer, 2021.

\bibitem[Garg and Murhekar(2023)]{garg2023computing}
J.~Garg and A.~Murhekar.
\newblock Computing fair and efficient allocations with few utility values.
\newblock \emph{Theoretical Computer Science}, 962:\penalty0 113932, 2023.

\bibitem[Garg et~al.(2018)Garg, Hoefer, and Mehlhorn]{garg2018approximating}
J.~Garg, M.~Hoefer, and K.~Mehlhorn.
\newblock Approximating the {N}ash social welfare with budget-additive
  valuations.
\newblock In \emph{Proceedings of the Twenty-Ninth Annual ACM-SIAM Symposium on
  Discrete Algorithms}, pages 2326--2340. SIAM, 2018.

\bibitem[Garg et~al.(2022{\natexlab{a}})Garg, Husic, Murhekar,
  et~al.]{garg2022tractable}
J.~Garg, E.~Husic, A.~Murhekar, et~al.
\newblock Tractable fragments of the maximum {N}ash welfare problem.
\newblock \emph{Web and Internet Economics (WINE)}, 2022{\natexlab{a}}.

\bibitem[Garg et~al.(2022{\natexlab{b}})Garg, Murhekar, and Qin]{garg2022fair}
J.~Garg, A.~Murhekar, and J.~Qin.
\newblock Fair and efficient allocations of chores under bivalued preferences.
\newblock In \emph{Proceedings of the AAAI Conference on Artificial
  Intelligence}, volume~36, pages 5043--5050, 2022{\natexlab{b}}.

\bibitem[Garg et~al.(2023)Garg, Murhekar, and Qin]{garg2023new}
J.~Garg, A.~Murhekar, and J.~Qin.
\newblock New algorithms for the fair and efficient allocation of indivisible
  chores.
\newblock In \emph{IJCAI}, 2023.

\bibitem[Gourv{\`e}s et~al.(2014)Gourv{\`e}s, Monnot, and
  Tlilane]{gourves2014near}
L.~Gourv{\`e}s, J.~Monnot, and L.~Tlilane.
\newblock Near fairness in matroids.
\newblock In \emph{ECAI}, volume~14, pages 393--398, 2014.

\bibitem[Hall(1935)]{hall1935representatives}
P.~Hall.
\newblock On representatives of subsets.
\newblock \emph{Journal of the London Mathematical Society}, 1\penalty0
  (1):\penalty0 26--30, 1935.

\bibitem[Halpern and Shah(2019)]{halpern2019fair}
D.~Halpern and N.~Shah.
\newblock Fair division with subsidy.
\newblock In \emph{International Symposium on Algorithmic Game Theory}, pages
  374--389. Springer, 2019.

\bibitem[Hosseini et~al.(2023)Hosseini, Sikdar, Vaish, and
  Xia]{hosseini2023fairly}
H.~Hosseini, S.~Sikdar, R.~Vaish, and L.~Xia.
\newblock Fairly dividing mixtures of goods and chores under lexicographic
  preferences.
\newblock In \emph{Proceedings of the International Joint Conference on
  Autonomous Agents and Multiagent Systems, AAMAS}, volume 2023, pages
  152--160, 2023.

\bibitem[Hsu()]{hsu2024efxmanuscript}
K.~Hsu.
\newblock {EFX} orientations of multigraphs.
\newblock To appear in ECAI 2025.

\bibitem[Hsu(2024{\natexlab{a}})]{hsu2024efxarxiv}
K.~Hsu.
\newblock {EFX} orientations of multigraphs.
\newblock \emph{arXiv preprint arXiv:2410.12039}, 2024{\natexlab{a}}.

\bibitem[Hsu(2024{\natexlab{b}})]{hsu2024existence}
K.~Hsu.
\newblock Existence of {MMS} allocations of mixed manna.
\newblock In \emph{ECAI 2024}, pages 3573--3580. IOS Press, 2024{\natexlab{b}}.

\bibitem[Hsu and King()]{hsu2025polynomialmanuscript}
K.~Hsu and V.~King.
\newblock A polynomial-time algorithm for {EFX} orientations of chores.
\newblock To appear in ECAI 2025.

\bibitem[Hsu and King(2025)]{hsu2025polynomialarxiv}
K.~Hsu and V.~King.
\newblock A polynomial-time algorithm for {EFX} orientations of chores.
\newblock \emph{arXiv preprint arXiv:2501.13481}, 2025.

\bibitem[Hummel(2023)]{hummel2023lower}
H.~Hummel.
\newblock On lower bounds for maximin share guarantees.
\newblock In \emph{Proceedings of the Thirty-Second International Joint
  Conference on Artificial Intelligence}, pages 2747--2755, 2023.

\bibitem[Kaneko and Nakamura(1979)]{kaneko1979nash}
M.~Kaneko and K.~Nakamura.
\newblock The {N}ash social welfare function.
\newblock \emph{Econometrica: Journal of the Econometric Society}, pages
  423--435, 1979.

\bibitem[Karp(1972)]{karp1972reducibility}
R.~Karp.
\newblock Reducibility among combinatorial problems.
\newblock \emph{Complexity of Computer Computations}, pages 85--103, 1972.

\bibitem[Kaviani et~al.(2024)Kaviani, Seddighin, and
  Shahrezaei]{kaviani2024almost}
A.~Kaviani, M.~Seddighin, and A.~Shahrezaei.
\newblock Almost envy-free allocation of indivisible goods: A tale of two
  valuations.
\newblock \emph{arXiv preprint arXiv:2407.05139}, 2024.

\bibitem[Kulkarni et~al.(2021)Kulkarni, Mehta, and
  Taki]{kulkarni2021indivisible}
R.~Kulkarni, R.~Mehta, and S.~Taki.
\newblock Indivisible mixed manna: On the computability of {MMS}+{PO}
  allocations.
\newblock In \emph{Proceedings of the 22nd ACM Conference on Economics and
  Computation}, pages 683--684, 2021.

\bibitem[Kurokawa et~al.(2016)Kurokawa, Procaccia, and Wang]{KPW16}
D.~Kurokawa, A.~D. Procaccia, and J.~Wang.
\newblock When can the maximin share guarantee be guaranteed?
\newblock In \emph{Thirtieth AAAI Conference on Artificial Intelligence}, 2016.

\bibitem[Lee(2017)]{lee2017apx}
E.~Lee.
\newblock {APX}-hardness of maximizing {N}ash social welfare with indivisible
  items.
\newblock \emph{Information Processing Letters}, 122:\penalty0 17--20, 2017.

\bibitem[Li et~al.(2025)Li, Sun, Suzuki, and Xing]{li2025subsidy}
B.~Li, A.~Sun, M.~Suzuki, and S.~Xing.
\newblock On the subsidy of envy-free orientations in graphs.
\newblock \emph{arXiv preprint arXiv:2502.13671}, 2025.

\bibitem[Lipton et~al.(2004)Lipton, Markakis, Mossel, and
  Saberi]{lipton2004approximately}
R.~J. Lipton, E.~Markakis, E.~Mossel, and A.~Saberi.
\newblock On approximately fair allocations of indivisible goods.
\newblock In \emph{Proceedings of the 5th ACM Conference on Electronic
  Commerce}, pages 125--131, 2004.

\bibitem[Mahara(2024)]{mahara2024extension}
R.~Mahara.
\newblock Extension of additive valuations to general valuations on the
  existence of {EFX}.
\newblock \emph{Mathematics of Operations Research}, 49\penalty0 (2):\penalty0
  1263--1277, 2024.

\bibitem[Mahara(2025)]{mahara2025existence}
R.~Mahara.
\newblock Existence of fair and efficient allocation of indivisible chores.
\newblock \emph{arXiv preprint arXiv:2507.09544}, 2025.

\bibitem[McGlaughlin and Garg(2020)]{mcglaughlin2020improving}
P.~McGlaughlin and J.~Garg.
\newblock Improving {N}ash social welfare approximations.
\newblock \emph{Journal of Artificial Intelligence Research}, 68:\penalty0
  225--245, 2020.

\bibitem[Moore(1959)]{moore1959shortest}
E.~F. Moore.
\newblock The shortest path through a maze.
\newblock In \emph{Proc. of the International Symposium on the Theory of
  Switching}, pages 285--292. Harvard University Press, 1959.

\bibitem[Murhekar and Garg(2021)]{murhekar2021fair}
A.~Murhekar and J.~Garg.
\newblock On fair and efficient allocations of indivisible goods.
\newblock In \emph{Proceedings of the AAAI Conference on Artificial
  Intelligence}, volume~35, pages 5595--5602, 2021.

\bibitem[Nash(1950)]{nash1950bargaining}
J.~Nash.
\newblock The bargaining problem.
\newblock \emph{Econometrica}, 18\penalty0 (2):\penalty0 155--162, 1950.

\bibitem[Plaut and Roughgarden(2020)]{plaut2020almost}
B.~Plaut and T.~Roughgarden.
\newblock Almost envy-freeness with general valuations.
\newblock \emph{SIAM Journal on Discrete Mathematics}, 34\penalty0
  (2):\penalty0 1039--1068, 2020.

\bibitem[Procaccia(2009)]{procaccia2009thou}
A.~D. Procaccia.
\newblock Thou shalt covet thy neighbor's cake.
\newblock In \emph{Proceedings of the 21st International Joint Conference on
  Artificial Intelligence}, pages 239--244, 2009.

\bibitem[Procaccia(2020)]{procaccia2020technical}
A.~D. Procaccia.
\newblock Technical perspective: An answer to fair division's most enigmatic
  question.
\newblock \emph{Communications of the ACM}, 63\penalty0 (4):\penalty0 118--118,
  2020.

\bibitem[Ramezani and Endriss(2009)]{ramezani2009nash}
S.~Ramezani and U.~Endriss.
\newblock {N}ash social welfare in multiagent resource allocation.
\newblock In \emph{International Workshop on Agent-Mediated Electronic
  Commerce}, pages 117--131. Springer, 2009.

\bibitem[Rawls(1971)]{rawls1971}
J.~Rawls.
\newblock \emph{A Theory of Justice}.
\newblock The Belknap Press of Harvard University Press, Cambridge,
  Massachusetts, 1971.

\bibitem[Roos and Rothe(2010)]{roos2010complexity}
M.~Roos and J.~Rothe.
\newblock Complexity of social welfare optimization in multiagent resource
  allocation.
\newblock In \emph{Proceedings of the 9th International Conference on
  Autonomous Agents and Multiagent Systems: volume 1-Volume 1}, pages 641--648,
  2010.

\bibitem[Sgouritsa and Sotiriou(2025)]{sgouritsa2025existence}
A.~Sgouritsa and M.~M. Sotiriou.
\newblock On the existence of {EFX} allocations in multigraphs.
\newblock \emph{arXiv preprint arXiv:2502.09777}, 2025.

\bibitem[Steinhaus(1948)]{steinhaus1948problem}
H.~Steinhaus.
\newblock The problem of fair division.
\newblock \emph{Econometrica}, 16:\penalty0 101--104, 1948.

\bibitem[Tao et~al.(2025)Tao, Wu, Yu, and Zhou]{tao2025existence}
B.~Tao, X.~Wu, Z.~Yu, and S.~Zhou.
\newblock On the existence of {EFX} (and {P}areto-{O}ptimal) allocations for
  binary chores.
\newblock \emph{Theoretical Computer Science}, 1042:\penalty0 115248, 2025.

\bibitem[Woeginger(1997)]{woeginger1997polynomial}
G.~J. Woeginger.
\newblock A polynomial-time approximation scheme for maximizing the minimum
  machine completion time.
\newblock \emph{Operations Research Letters}, 20\penalty0 (4):\penalty0
  149--154, 1997.

\bibitem[Woeginger and Sgall(2007)]{woeginger2007complexity}
G.~J. Woeginger and J.~Sgall.
\newblock On the complexity of cake cutting.
\newblock \emph{Discrete Optimization}, 4\penalty0 (2):\penalty0 213--220,
  2007.

\bibitem[Zeng and Mehta(2024)]{zeng2024structure}
J.~A. Zeng and R.~Mehta.
\newblock On the structure of envy-free orientations on graphs.
\newblock \emph{arXiv preprint arXiv:2404.13527}, 2024.

\bibitem[Zhou et~al.(2024)Zhou, Wei, Li, and Li]{zhou2024complete}
Y.~Zhou, T.~Wei, M.~Li, and B.~Li.
\newblock A complete landscape of {EFX} allocations on graphs: goods, chores
  and mixed manna.
\newblock In \emph{Proceedings of the Thirty-Third International Joint
  Conference on Artificial Intelligence}, pages 3049--3056, 2024.

\end{thebibliography}

\end{document}